\documentclass[english,11pt]{article}
\usepackage{amsmath}
\usepackage{amsthm}
\usepackage{amssymb}
\usepackage{fontspec}
\usepackage{float}
\usepackage{bm}
\usepackage{graphicx}
\usepackage{geometry}
\usepackage{rotfloat}
\usepackage{setspace}
\usepackage[authoryear]{natbib}
\usepackage[bookmarks=true,bookmarksnumbered=true,bookmarksopen=false,
 breaklinks=true,pdfborder={0 0 0},pdfborderstyle={},backref=false,colorlinks=false]
 {hyperref}
\hypersetup{pdftitle={Robust Signal Maximization in Spillover Experiments},
 pdfauthor={K. Borusyak, P. Hull, E. Munro}}

\makeatletter

\providecommand{\tabularnewline}{\\}
\floatstyle{ruled}
\newfloat{algorithm}{tbp}{loa}
\providecommand{\algorithmname}{Algorithm}
\floatname{algorithm}{\protect\algorithmname}

\usepackage[strict]{changepage}
\newcommand\independent{\protect\mathpalette{\protect\independenT}{\perp}}
\def\independenT#1#2{\mathrel{\rlap{$#1#2$}\mkern2mu{#1#2}}}

\usepackage{pdflscape}
\usepackage{afterpage}
\usepackage{flafter}

\usepackage{dcolumn,booktabs}
\usepackage{multirow,booktabs}
\usepackage{threeparttable}
\usepackage[justification=centering]{caption}
\newcolumntype{d}[1]{D{.}{.}{#1}}

\PassOptionsToPackage{position=top}{subfig}

\usepackage{tikz}
\usetikzlibrary{arrows.meta, positioning, calc}

\usepackage{chngcntr}

\makeatother

\theoremstyle{plain}
\newtheorem{thm}{\protect\theoremname}
\newtheorem{prop}{\protect\propositionname}
\newtheorem{lem}{\protect\lemmaname}
\newtheorem{cor}{\protect\corollaryname}
\newtheorem{assumption}{\protect\assumptionname}
\usepackage{polyglossia}
\setdefaultlanguage[variant=american]{english}
\providecommand{\assumptionname}{Assumption}
\providecommand{\corollaryname}{Corollary}
\providecommand{\lemmaname}{Lemma}
\providecommand{\propositionname}{Proposition}
\providecommand{\theoremname}{Theorem}

\begin{document}
\title{Robust Signal Maximization in Spillover Experiments}
\author{\vspace{1.5cm}
}
\author{Kirill Borusyak\\
UC Berkeley\and Peter Hull\\
Brown\and Evan Munro\\
Chicago Booth\thanks{Contact: k.borusyak@berkeley.edu, peter\_hull@brown.edu, and evan.munro@chicagobooth.edu. We thank David Atkin, Michael Best, and Eric Verhoogen for encouraging this project; we thank Dmitry Arkhangelsky, Gabriel Kreindler, Shuangning Li, David McKenzie, Ted Miguel, and  Davide Viviano for useful comments. OpenAI's ChatGPT contributed valuable insights. Refine.ink was used to check for consistency and clarity.}}
\date{\vspace{0.25cm}
November 2026}

\maketitle
\vspace{0.25cm}

\begin{abstract}
\begin{singlespace} \noindent\begin{adjustwidth*}{0.9cm}{0.9cm} {\normalsize We
study the optimal design and analysis of experiments for estimating
spillover effects. Assuming a linear exposure model, we characterize
the treatment-assignment distribution and regression-based estimator
that minimize worst-case approximate variance against a broad class
of distributions of unobservables. The design problem yields an intuitive
solution in which the planner trades off spillover signal strength
against dispersion of spillover variation. The analysis problem yields
a simple recentered instrumental variable estimator. We illustrate
these new tools in semi-synthetic experiments based on three applications
from development economics. Our approach yields large variance reductions,
increasing effective sample sizes by 44-266\%.}\end{adjustwidth*} \end{singlespace} \vfill{}
\thispagestyle{empty}
\end{abstract}
\newpage
\global\long\def\expec#1#2{\mathbb{E}_{#1}\left[#2\right]}%
\global\long\def\Pr#1#2{\mathrm{Pr}_{#1}\left[#2\right]}%
\global\long\def\var#1#2{\mathrm{Var}_{#1}\left[#2\right]}%
\global\long\def\cov#1#2{\mathrm{Cov}_{#1}\left[#2\right]}%
\global\long\def\sd#1#2{\mathrm{SD}_{#1}\left[#2\right]}%
\global\long\def\corr#1#2{\mathrm{Corr}_{#1}\left[#2\right]}%
\global\long\def\V#1#2{\mathcal{V}_{#1}\left[#2\right]}%
\global\long\def\one{\mathbf{1}}%
\global\long\def\diag{\operatorname{diag}}%
\global\long\def\tr{\operatorname{tr}}%
\global\long\def\plim{\operatorname*{plim}}%
\global\long\def\col{\operatorname{Col}}%
\global\long\def\Im{\operatorname*{Im}}%
\global\long\def\op#1{\left\Vert #1\right\Vert _{\infty}}%
\global\long\def\tr{\operatorname{tr}}%
\global\long\def\rank{\operatorname{rank}}%
\global\long\def\Frob#1{\left\Vert #1\right\Vert _{2}}%
\global\long\def\toP{\xrightarrow{p}}%
\global\long\def\toD{\Rightarrow}%
\global\long\def\fhat{\widehat{f}}%
\global\long\def\bhat{\widehat{\beta}}%
\global\long\def\PR{\mathbb{P}}%
 
\global\long\def\dist{\operatorname{dist}}%
\global\long\def\interior#1{\mathrm{int}\left(#1\right)}%

\section{Introduction}

Researchers increasingly estimate spillovers in experiments. Sometimes,
an intervention is randomized to one set of units while the spillover
effects are measured in another set. In such “bipartite” experiments
the intervention and outcome units are connected via a network: e.g.,
of friendship across individuals \citep{cai_social_2015} or of supplier
relationships across firms \citep{best_spillover_2025}. In other
cases, the intervention and outcome units coincide and direct and
spillover effects are jointly estimated in a single experimental sample
\citep{Miguel2004,chaurey_social_2025}. Power is a common issue in
such studies: spillover effect estimates can be noisy, while increasing
the sample size is costly. It is therefore important to make the best
use of the experimental variation, as well as to design the randomization
to target the spillover parameter of interest.

This paper develops new tools for experimental design and estimation
when spillover effects are of primary interest. Our baseline assumption
is that the researcher adopts a linear model of spillovers, in which
the key treatment variable is the fraction or number of “friends”
(i.e., network connections) who have been randomly selected for the
intervention—perhaps augmented with weights reflecting the importance
of each connection.\footnote{Specifically, we assume the spillover treatment can be written $x_{i}=\sum_{k}w_{ik}g_{k}$
where $g_{k}$ is an indicator that unit $k$ is assigned to treatment.
The $w_{ik}$ are fully unrestricted, and the intervention units $k=1,\dots,K$
can differ from the outcome units $i=1,\dots,N$. We also discuss
how our approach can be used for nonlinear formula treatments $x_{i}=X_{i}(g_{1},\dots,g_{K})$
for known $X_{i}(\cdot)$ that can implicitly depend on some $w_{ik}$
\citep{borusyak_design-based_2025}.} Such linear specifications are widespread in economics, including
in experiments (e.g., \citealp{Miguel2004,cai_social_2015,atkin_reducing_2024}).
We assume the researcher has access to the network measure (or at
least a correlated proxy) at the design stage. We further assume the
researcher wants to identify the spillover effect using the experimental
variation; hence, they use some recentered estimator \citep{BH1}.
It is increasingly recognized that recentering is necessary to guarantee
a causal interpretation of spillover effect estimates without assuming
that the network is exogenous, with appropriate regression or instrumental
variable (IV) estimators employed either intuitively \citep{Miguel2004}
or formally \citep{chaurey_social_2025,atkin_reducing_2024}.

We then ask: which randomization design and which recentered estimator
are likely to yield good power? Our answer reflects two countervailing
intuitions. On the one hand, when an outcome unit $i$ is exposed
to the shocks of multiple intervention units $k$, it is desirable
to positively correlate the treatment assignments of those intervention
units. Doing so increases the variance of $i$'s spillover treatment
(i.e., the spillover “signal”). On the other hand, when multiple
outcome units are exposed to the same intervention unit, their exposures
are correlated and this can increase estimation noise if their unobservables
are also correlated; correlating assignments of multiple intervention
units exacerbates this clustering problem. This makes it desirable
for exposures to have minimal correlations and more equal variances
(“isotropy”). A likely-powerful design is thus one which maximizes
spillover signals while ensuring sufficient dispersion in the spillover
variation. The estimator choice can also help with isotropy, reweighting
data to spread out the induced spillover variation at the cost of
a lower signal.

We formalize these intuitions in a minimax approach: we look for the
experimental design and recentered estimator that minimize the worst-case
approximate variance over a large class of error distributions. We
parameterize this class parsimoniously, by a scalar that captures
the extent to which the researcher is willing to rule out mutual dependencies,
dependence on the network, and heteroskedasticity of the errors. This
approach has three advantages. First, it reflects the reality that
researchers tend to have at best a rough sense of how errors are clustered,
especially at the experimental design stage. Second, it allows for
general network structures and yields a tractable closed-form characterization
of worst-case variance, both as a function of the estimator (holding
the design fixed) and as a function of the design (provided the optimal
estimator will be used). Although implementing our solution generally
involves numerical optimization, it is amenable to practical algorithms.
Finally, the minimax approach yields non-degenerate and intuitive
solutions in contrast to some corner solutions in the literature (e.g.,
\citet{kiefer_general_1974}). For example, independent randomization
and cluster-level randomization arise as special cases when, respectively,
there are no spillovers and the exposure is the cluster-level treatment
saturation. When exposure is a leave-one-out average of cluster shocks,
the solution is to optimally mix these two designs.

We characterize the minimax-optimal design and estimator in two steps.
First, we derive the optimal recentered IV estimator for any experimental
design. Its construction entails taking the spillover treatment itself,
recentering it by its expectation over the experimental design, and
partially “whitening” it to spread out the remaining exogenous
variation. While the whitening step weakens the first stage, the induced
isotropy improves robustness against non-spherical errors. Second,
we characterize the best designs given the optimal recentered instrument
will be used. We show this optimal design satisfies natural properties,
such as a 50/50 marginal distribution for each shock and separability
across independent network blocks. The optimal IV is straightforward
to construct for any design, and we propose a computationally efficient
approximation to the optimal design; the entire procedure can be described
by a six-line algorithm (Algorithm \ref{alg:full}). In the parameterization
with the highest robustness to adversarial errors, this procedure
yields an especially intuitive closed-form solution in which the correlation
of shock assignments is determined by the cosine similarity of exposure
to them. We show how our estimators are asymptotically normal with
a sparse exposure network—even when the errors have strong mutual
correlations—and give a simple inference procedure.

We extend these solutions in several directions, including by characterizing
minimax-optimal designs and estimators when spillover effects are
estimated alongside direct effects or other linear exposure measures.
Here the optimal design balances signal and isotropy of the residual
variation in the focal spillover treatment after accounting for the
other exposures. Other extensions include optimizing estimation of
the focal spillover subject to a constraint on the other exposures,
estimators with predetermined covariates, nonlinear spillover treatments,
budget constraints in allocating shocks, and optimal design with instruments
that adapt to observed outcomes.

Finally, we illustrate the power gains from using our approach. To
do so we calibrate semi-synthetic experiments to data from the bipartite
experiment in \citet{cai_social_2015} and the joint estimation of
direct and spillover effects in \citet{Miguel2004}. Relative to independent
randomization and no whitening, we find large declines in standard
errors in both settings—equivalent to increasing effective sample
sizes by between 44\% and 237\% in our preferred approach. We find
similarly large gains relative to a saturation design benchmark in
the setting of \citet{egger2022general}.

This paper contributes to several related literatures. Most directly,
we add to a growing literature on optimal experimental design under
interference. Many papers in this literature focus on specific settings,
a restricted class of designs, and a pre-specified estimator (with
different models of interference and estimands). For example, \citet{Baird2018}
and \citet{Cruces} focus on partial interference, where spillovers
are confined within clusters, and derive optimal saturation designs
(discussed more below) for estimating direct and spillover effects,
including by regression. \citet{pouget-abadie_variance_2019} and
\citet{harshaw_design_2023} instead focus on bipartite experiments
and derive optimal clustered assignments under a linear exposure-response
model with heterogeneous treatment effects. \citet{thiyageswaran_optimal_2026}
study optimal worst-case estimation of the global average treatment
effect (GATE) over a general class of designs using the standard Horvitz-Thompson
estimator, which does not account for spillovers and is generally
biased. \citet{faridani2026linear} study experimental design and
estimation of the GATE when the structure of interference is unknown
but decays with distance sufficiently fast, deriving minimax convergence
rates.\footnote{Other work with general networks includes \citet{viviano_experimental_2025}
who studies two-wave experiments for estimands including average treatment
and spillover effects, using a pilot wave to estimate error variances,
and \citet{viviano_causal_2025} who choose treatment clusterings
to trade off worst-case asymptotic bias and variance when estimating
GATEs.}

Our approach departs from this literature in four key ways. First,
we provide a unified and tractable characterization of optimal experimental
designs for general network structures, nesting bipartite graphs as
a special case. Second, we do not restrict the class of designs \emph{a
priori}. Third, we jointly optimize the design and the estimator to
achieve further power improvement. We minimize the finite-sample
approximate variance of the estimator rather than only its convergence
rate, as in \citet{faridani2026linear} and \citet{kandiros2026design}.
Like \citet{harshaw_design_2023}, we employ recentering to ensure
identification without assuming network exogeneity. Finally, we target
the parameters of a linear spillover model instead of GATE-style estimands.
While more restrictive, such constant-effects models are widely used
in practice with evidence for linearity in different contexts (e.g.,
in \citet{Miguel2004} and \citet{egger2022general}).\footnote{In the presence of heterogeneous treatment effects, our baseline estimators
identify their convex averages. This is shown in Appendix \ref{appx:HetFX},
in line with earlier results on design-based estimators \citep{borusyak_negative_2024}.} In bipartite experiments, the GATE is proportional to the linear
spillover parameter; otherwise, when both direct and indirect effects
are included, our baseline approach focuses on isolating them rather
than obtaining the total. However, we show that our approach can be
reformulated to target the GATE under a known linear-exposure structure.
This added structure likely yields more power than designs targeting
GATEs under more non-parametric interference, e.g. in \citet{faridani2026linear}.

Our solution also differs from saturation designs: a popular approach
for clustered networks with complete within-cluster spillovers in
which units are randomized according to cluster-specific saturation
levels, which are themselves drawn randomly from some distribution.
Our approach applies to general networks, yielding power gains both
in non-clustered settings (e.g., with geographic spillovers) and when
spillovers are incomplete within clusters. Even in complete-cluster
settings, we improve on standard saturation designs in two ways. First,
we allow cluster-specific assignment correlations and apply partial
whitening. Second, we derive an analytically tractable characterization
of the minimax solution, in contrast to the existing literature which
offers limited practical guidance for choosing a good saturation design
in linear exposure models.\footnote{\citet{Baird2018} derive optimal saturation levels for estimating
slope effects under equal cluster sizes and strong restrictions on
the error structure (see also \citealp{hirano_design_2010}). \citet{Cruces}
relax these restrictions, but optimize only the shares of clusters
assigned to fixed saturation levels and only with fully nonparametric
estimation.}

We also add to a literature using minimax arguments to justify different
experimental designs. In settings without interference, \citet{kallus_optimality_2021}
and \citet{Bai2021} show that complete or blockwise randomization
is minimax-optimal absent strong assumptions on error dependence while
\citet{hull2026optimal} shows independent randomization is optimal
when errors are bounded (see also \citet{wu1981robustness} and \citet{li1983minimaxity};
\citet{bojinov2023design} study minimax designs in switchback experiments).
We extend this logic to spillover estimation with recentered IV, where
we show how the minimax-optimal design depends on the allowed adversarial
clustering of errors. As in \citet{thiyageswaran_optimal_2026}, we
impose a Schatten $p$-norm constraint on the second-moment matrix
of model residuals to define our minimax problem; choosing a smaller
$p$ leads to greater robustness to cross-unit dependence by making
the spillover variation more dispersed. Like \citet{thiyageswaran_optimal_2026},
we solve a relaxed version of an infeasible optimization problem and
then use Gaussian rounding \citep{goemans_improved_1995}; see also
\citet{guo2026gaussianized}. The ex-ante choice of $p$ by the researcher
is similar to choices made in the broader literature on minimax-optimal
estimators.\footnote{In non-parametric regression, for example, researchers choose the
smoothness class of the data-generating process—determining the tradeoff
between robustness and confidence interval length \citep{armstrong2018optimal}.} It also connects to the design-based estimation theory of \citet{chang2023design},
who propose the Schatten $\infty$-norm of the design matrix as a
measure of experimental complexity.

The minimax approach allows us to avoid degenerate solutions that
often arise in other analyses. For example, a classic literature on
optimal regression design \citep[e.g., ][]{kiefer_general_1974} chooses
where to place observations in covariate space to minimize a functional
of the ordinary least squares (OLS) covariance matrix (e.g. D-optimality,
which minimizes its determinant; \citet{kiefer1959optimum}). This
literature assumes independent and homoskedastic errors, and optimal
designs tend to be concentrated on a few extreme points in the design
space.\footnote{Our focus on non-spherical errors echoes \citet{bickel_robustness_1979}
who recover non-degenerate designs when errors are serially correlated.} In the interference setting, the optimal design in \citet{pouget-abadie_variance_2019}
can similarly produce very coarse randomization: e.g., grouping all
units into two clusters only and assigning treatment at the cluster
level. \citet{harshaw_design_2023} employ a heuristic correction
to the optimization criterion to reduce the issue, while our minimax
approach resolves it naturally—yielding a guarantee of asymptotic
normality with a sparse exposure network and mild regularity conditions.

Lastly, we add to a literature on robust and powerful spillover estimation
with recentered IV. \citet{BH1} show how recentering via knowledge
of the (quasi-)experimental design can identify spillover effects
under arbitrary endogeneity of the network, while \citet{borusyak_optimal_2026}
and \citet{borusyak_estimating_2025} show how asymptotically optimal
recentered instruments can be constructed given a particular design.
Here we optimize both the design and instrument, with our applications
showing that both levers can meaningfully improve estimator precision.
Moreover, while the implementation of the \citet{borusyak_optimal_2026}
optimal estimator is hindered by its dependence on the (likely unknown)
dependence of error terms, our minimax solution provides a parsimonious
solution robust to a wide range of error structures.

The remainder of this paper is organized as follows. In Section \ref{sec:optimal_designs}
we develop our theoretical framework and results. In Section \ref{sec:computation_clt},
we propose tractable methods of computing the optimal design and provide
a central limit theorem for the recentered IV estimator under the
feasible design. In Section \ref{sec:Applications} we illustrate
these tools in three semi-synthetic experiments. Section \ref{sec:Conclusion}
concludes. The proofs of main results are given in Appendix \ref{appx:Proofs-of-Main}.
Additional theoretical results are given in Appendix \ref{appx:Additional-Results},
with proofs in Appendix \ref{appx:Proofs-additional}. Appendix \ref{sec:appx-empirics}
gives additional empirical results.

\section{Theory\label{sec:optimal_designs}}

\subsection{Setting}

We consider estimation of the spillover effects of a set of treatment
shocks $g_{k}\in\{0,1\}$, $k=1,\dots,K$, on a set of outcomes $y_{i}$,
$i=1,\dots,N$. This might be in a bipartite setting, in which the
$K$ intervention units are distinct from the $N$ outcome units,
or in a setting like \citet{Miguel2004} in which we estimate spillovers
within a single experimental sample (with $N=K$).

We parameterize spillovers with a linear exposure model, which we
assume is correctly specified: 
\begin{align}
y_{i} & =\beta x_{i}+\varepsilon_{i},\qquad x_{i}=w^{\prime}_{i}g,\label{eq:model_baseline}
\end{align}
where the spillover treatment $x_{i}$ combines the shocks $g=(g_{1},\dots,g_{K})^{\prime}$
with exposure weights $w_{i}\in\mathbb{R}^{K}$, which may or may
not add to one, and $\varepsilon_{i}$ is an unobserved error.\footnote{Here correct specification includes an implicit exclusion restriction:
that the shocks only affect outcomes through $x_{i}$. While we focus
on a linear constant-effects model, Appendix \ref{appx:HetFX} discusses
the interpretation of our estimands under heterogeneous effects. The
weights given to different units' treatment effects, derived there,
can be computed and compared to the weights of other more conventional
estimands; in some cases these weights coincide.} For example, \citet{cai_social_2015} estimate a spillover effect
$\beta$ from information sessions randomized to a set of $K$ farmers
in China on their peers' decisions to adopt weather insurance $y_{i}$;
here $x_{i}=w^{\prime}_{i}g$ gives the share of farmer $i$'s neighbors
who were assigned to the information session, with $w_{ik}$ being
entries of the row-normalized adjacency matrix of the peer graph.\footnote{We assume that the relevant network is measured at the design stage,
but require only that we observe the connections between outcome and
intervention units. We do not need, for example, connections between
one intervention unit and another at the design stage. This is not
a strong requirement, as the connections between outcome and intervention
units are necessary to estimate the regression in (\ref{eq:model_baseline})
anyway (although not necessarily at the design stage). Moreover, in
some contexts, such as when spillovers are based on geographical distance
in \citet{Miguel2004}, no additional data collection is required.
Efficiency gains are also likely if the design is based on a noisy
measure of the network, measured at baseline. If the network can itself
endogenously change after the experiment, the researcher can use the
ex-post network for the endogenous variable in the regression while
using the ex-ante network to construct the IV, as in \citet{gao2024endogenousinterferencerandomizedexperiments}.} In Section \ref{subsec:Direct-Effects} we extend the model to include
direct effects of treatment, for settings like \citet{Miguel2004}.
Below we discuss how our main results extend to spillover treatments
that are nonlinear functions of the shocks.

In most of our analysis, we condition on the exposure matrix $W=(w_{1},\dots,w_{N})^{\prime}$;
we make this conditioning implicit to simplify notation. We allow
the error vector $\varepsilon=(\varepsilon_{1},\dots,\varepsilon_{N})^{\prime}$
to be either fixed (as in a standard design-based setup) or drawn
from some distribution that can implicitly depend on $W$.\footnote{All of our results hold in both cases but the restriction on the error
distribution introduced below is more natural when $\varepsilon$
are viewed as stochastic. Nevertheless, our approach is still design-based
in the sense that the moment conditions we leverage derive their validity
from the randomization of $g$ alone.} Without loss of generality, we assume that no column of $W$ is entirely
zero.

We consider the problem of choosing an experimental design, i.e.,
a probability distribution for $g$: $\delta\in\mathcal{D}$ where
$\mathcal{D}=\Delta(\left\{ 0,1\right\} ^{K})$ and $\Delta(\cdot)$
denotes the simplex. By virtue of randomization, we have $g\independent\varepsilon$.
We then estimate $\beta$ by recentered IV \citep{BH1}: i.e., we
choose a set of instrument functions $z_{i}(\cdot):\{0,1\}^{K}\rightarrow\mathbb{R}$
where $\expec{\delta}{z_{i}(g)}=0$ for all $i$, and estimate:
\begin{align*}
\hat{\beta}\left[z\right] & =\frac{\sum^{N}_{i=1}z_{i}(g)y_{i}}{\sum^{N}_{i=1}z_{i}(g)x_{i}}.
\end{align*}
We let $\mathcal{Z}_{\delta}$ be the set of all such $z=(z_{i})^{N}_{i=1}$
for a given $\delta$.

Recentering (i.e., the mean-zero property of $z_{i}(g)$, implicitly
conditional on $W$) ensures we identify $\beta$ just using the experimental
variation in $g$: the network $W$ can be arbitrarily correlated
with the unobserved $\varepsilon$. Indeed, the moment condition $\expec{\delta}{\sum^{N}_{i=1}z_{i}\varepsilon_{i}}=0$
is satisfied for any such $z$. Recentered $z\in\mathcal{Z}_{\delta}$
can be obtained from any fixed function, $\tilde{z}_{i}$, of both
$g$ and $W$ by subtracting its (known) expectation under a given
design: i.e. $z_{i}=\tilde{z}_{i}-\expec{\delta}{\tilde{z}_{i}}$.
The set of recentered IV estimators $\hat{\beta}\left[z\right]$,
$z\in\mathcal{Z}_{\delta}$, is large, including estimators that weight
by functions of $W$ (by rescaling $z$), control for functions of
$W$ (by the Frisch–Waugh–Lovell theorem), as well as non-instrumented
estimators (e.g., OLS of $y_{i}$ on $\tilde{z}_{i}=x_{i}$ controlling
for $\expec{\delta}{\tilde{z}_{i}}$, again by the Frisch–Waugh–Lovell
theorem).\footnote{Appendix \ref{appx:non-recentered} considers estimation without recentering.
While non-recentered IVs can in principle improve on mean-squared
error grounds by introducing a small amount of bias to estimation,
we show the corresponding optimal designs can have undesirable properties
like unusually high treatment rates. We also show that when $W\mathbf{1}=1$
and an intercept is included in estimation (or, more generally, when
$W\mathbf{1}$ is linearly spanned by the included controls), a joint
optimum over designs and instruments can be attained with a recentered
instrument.}

With this setup, we ask: what design $\delta$ and instrument $z$
should we choose in order to get the most precise estimate of the
spillover effect $\beta$?

\subsection{Optimal Design and Instrument}

We follow \citet{borusyak_optimal_2026} in studying the finite-sample
approximate variance of the recentered IV estimator: 
\begin{align*}
\V{\delta,\mathcal{E}}z & =\frac{\var{\delta,\mathcal{E}}{\sum^{N}_{i=1}z_{i}(g)\varepsilon_{i}}}{\expec{\delta}{\sum^{N}_{i=1}z_{i}(g)x_{i}}^{2}},
\end{align*}
where here we let $\mathcal{E}$ denote the (unknown) distribution
of $\varepsilon$, with $\V{\delta,\mathcal{E}}z=\infty$ whenever
the denominator is zero. Proposition 2 of \citet{borusyak_optimal_2026}
shows that $\V{\delta,\mathcal{E}}z$ gives a good approximation to
the appropriately scaled asymptotic variance of $\hat{\beta}[z]$
so long as the estimator is well-behaved in a particular sense, which
primarily means it converges to $\beta$ at some rate with a non-vanishing
first stage.\footnote{\label{fn:well-defined-estimator}As a just-identified IV estimator,
$\hat{\beta}\left[z\right]$ may have no finite-sample moments. More
precisely, in our setting with discrete $g$ the problems may arise
from $\hat{\beta}\left[z\right]$ being undefined on the event that
$\sum^{N}_{i=1}z_{i}(g)x_{i}=0$.}

We look for the experimental design and instrument that minimize the
worst-case $\V{\delta,\mathcal{E}}z$ under a weak restriction on
the error distribution. Specifically, we suppose the researcher believes
$\mathcal{E}$ belongs to some class of distributions $\mathcal{F}_{p}(\sigma)$
and solve the minimax problem:
\begin{align}
\delta^{*},z^{*} & \in\arg\min_{\delta\in\mathcal{D},z\in\mathcal{Z}_{\delta}}\max_{\mathcal{E}\in\mathcal{F}_{p}(\sigma)}\V{\delta,\mathcal{E}}z.\label{eq:minimax}
\end{align}
The class of error distributions, parameterized by $p\in[1,\infty]$
and $\sigma>0$, is given by: 
\begin{align}
\mathcal{F}_{p}(\sigma) & =\left\{ \mathcal{E}\colon\left\Vert \expec{\mathcal{E}}{\varepsilon\varepsilon'}\right\Vert _{p}\le N^{1/p}\sigma^{2}\right\} ,\label{eq:F_p}
\end{align}
where $\left\Vert \cdot\right\Vert _{p}$ is the Schatten-$p$ matrix
norm defined for positive semi-definite matrices as\footnote{For general matrices, it is defined as $\tr\left(\left(A'A\right)^{p/2}\right)^{1/p}$.
Schatten norms should not be confused with other matrix norms that
can be denoted in the same way.}
\[
\left\Vert A\right\Vert _{p}=\left(\tr\left(A^{p}\right)\right)^{1/p}.
\]
Three important special cases are the nuclear norm (a.k.a., the trace
norm) with $p=1$, the Frobenius norm with $p=2$, and the operator
norm with $p=\infty$. Denoting the eigenvalues of $\expec{\mathcal{E}}{\varepsilon\varepsilon'}$
by $\lambda_{1},\dots,\lambda_{N}\ge0$, we can rewrite (\ref{eq:F_p})
as an upper bound on their power mean:
\[
\mathcal{F}_{p}(\sigma)=\left\{ \mathcal{E}\colon\left(\frac{1}{N}\sum^{N}_{i=1}\lambda^{p}_{i}\right)^{1/p}\le\sigma^{2}\right\} ,
\]
with the limit $\mathcal{F}_{p}(\sigma)=\left\{ \mathcal{E}\colon\lambda_{\text{max}}\left(\expec{\mathcal{E}}{\varepsilon\varepsilon'}\right)\le\sigma^{2}\right\} $
where $\lambda_{\max}=\max\left\{ \lambda_{1},\dots,\lambda_{N}\right\} $
for $p=\infty$.

\begin{figure}[t]
\caption{Boundary of $\mathcal{F}_{p}(\sigma)$ for $\sigma^{2}=1$, $N=2$,
for Different Choices of $p$\label{fig:Fp_illustration}}

\begin{centering}
\includegraphics[width=0.3\paperwidth]{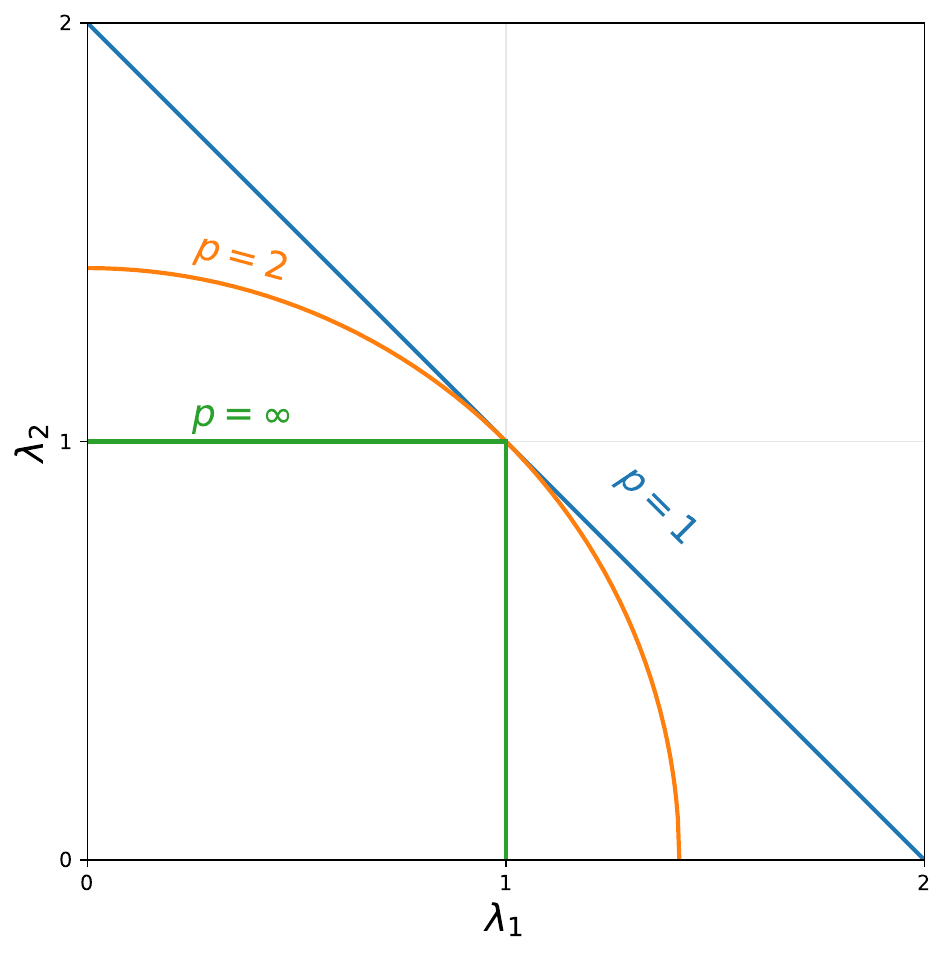}\smallskip{}
\par\end{centering}
{\small\emph{Notes}}{\small : This figure shows boundaries of $\mathcal{F}_{p}(\sigma)$
for different Schatten-$p$ norms, with $N=2$ observations and $\sigma=1$,
in terms of the two eigenvalues of the $\expec{\mathcal{E}}{\varepsilon\varepsilon'}$
matrix. The $(1,1)$ point corresponds to spherical errors.}{\small\par}
\end{figure}

To interpret $\mathcal{F}_{p}(\sigma)$, note first that spherical
(i.e., homoskedastic and mutually uncorrelated) errors with mean zero
(implicitly conditional on $W$) and variance $\sigma^{2}$, i.e.
$\expec{\mathcal{E}}{\varepsilon\varepsilon'}=\sigma^{2}I_{N}$, are
at the boundary of this set for any $p$. Different choices of $p$
then place different weights on the correlations among the errors.
When $p=1$, only the average second moment of $\varepsilon_{i}$
is constrained, allowing fully unrestricted correlations across observations.
At the other extreme, when $p=\infty$, correlations are strongly
penalized such that a worst-case scenario in (\ref{eq:minimax}) is
always with spherical errors. Intermediate values of $p$ correspond
to smaller penalties on correlations. Figure \ref{fig:Fp_illustration}
provides a visualization for $N=2$ by plotting the boundary of $\mathcal{F}_{p}(\sigma)$
in terms of the two eigenvalues of $\expec{\mathcal{E}}{\varepsilon\varepsilon'}$;
increasing the correlation across errors makes the eigenvalues more
asymmetric without changing their sum; see also Appendix \ref{appx:Equicorrelated}
for analytical results with equicorrelated $\varepsilon_{i}$. Overall,
$p$ captures the extent to which the researcher is willing to rule
out strong correlations across the errors when designing the experiment.
Similarly, deviations from homoskedasticity and from network exogeneity
(i.e., the zero conditional mean of the errors) are penalized more
strongly when $p$ is larger.

\paragraph{Optimal Instrument.}

Solving the problem (\ref{eq:minimax}) backwards, we first characterize
the optimal instrument for a given design $\delta$:
\begin{thm}
\label{prop:OptimalIV}Fix $\delta$ and let
\begin{align*}
\tilde{x}_{\delta} & =W\tilde{g}_{\delta},\qquad\tilde{g}_{\delta}=g-\expec{\delta}g,\\
S_{\delta} & =\var{\delta}x=W\var{\delta}gW^{\prime}.
\end{align*}
Suppose $S_{\delta}\neq0$. Then for any $p\in[1,\infty]$ and $\sigma>0$,
the optimal instrument is:
\begin{align*}
z^{*}_{\delta}(g) & \propto\left(S^{\dag}_{\delta}\right)^{\frac{1}{p+1}}\tilde{x}_{\delta}\in\arg\min_{z\in\mathcal{Z}_{\delta}}\max_{\mathcal{E}\in\mathcal{F}_{p}(\sigma)}\V{\delta,\mathcal{E}}z,
\end{align*}
where $\dag$ denotes the pseudoinverse. The worst-case approximate
variance is:
\[
\max_{\mathcal{E}\in\mathcal{F}_{p}(\sigma)}\V{\delta,\mathcal{E}}{z^{*}_{\delta}}=N^{1/p}\sigma^{2}\left(\tr(S^{p/(p+1)}_{\delta})\right)^{-(p+1)/p}.
\]
\end{thm}
\noindent The Appendix \ref{appx:Proof-Thm1} proof provides an explicit
characterization of the worst-case error distribution. In this result
and for the rest of the paper, we use the convention that $1/\infty=0$,
$p/(p+1)=1$ for $p=\infty$, and $A^{0}=A^{\dag}A$ for a matrix
$A$.

Intuitively, $z^{*}_{\delta}$ takes the spillover treatment $x$,
recenters it by $\expec{\delta}x$, and then reweights the resulting
$\tilde{x}_{\delta}$ inversely by a fractional power of $\var{\delta}x$,
$S^{1/(p+1)}_{\delta}$. The goal in reweighting is to spread out
the useful variation in $\tilde{x}_{\delta}$ across different directions
that could be adversarially targeted by clustering in $\varepsilon$.
When $p=1$ and $\var{\delta}x$ is an invertible matrix, this fully
“whitens” $\tilde{x}_{\delta}$ (i.e., rotates and rescales
it to make $\var{\delta}{z^{*}_{\delta}}=I$) in order to guard against
the fully-unconstrained clustering of $\varepsilon$,\footnote{When $\var{\delta}x$ is not invertible, e.g. because $\delta$ perfectly
correlates $x_{i}$ for several $i$, $z^{\ast}$ preserves that dependence.} while as $p$ increases there is less whitening as errors with strong
mutual correlations are excluded. As a result, at $p=\infty$, the
optimal IV is the unwhitened $\tilde{x}_{\delta}$.\footnote{Note this holds even with rank-deficient $S_{\delta}$ since $S^{\dagger}_{\delta}S_{\delta}\tilde{x}_{\delta}=\tilde{x}_{\delta}$.
The worst-case optimality of this instrument for $p=\infty$, i.e.
when $\mathcal{F}_{p}(\sigma)$ bounds the operator norm of the error
second-moment matrix, was previously shown in \citet[Lemma 2]{borusyak_optimal_2026}.} Another intuition for $z^{*}_{\delta}$ comes from the Appendix \ref{appx:Proof-Thm1}
proof, which shows that the partial whitening of $\tilde{x}_{\delta}$
corresponds to generalized least squares (GLS) weighting with a weight
matrix proportional to the inverse worst-case error matrix, $\expec{\mathcal{E}}{\varepsilon\varepsilon'}^{\dag}$,
corresponding to the design.

\paragraph{Optimal Design.}

We next characterize the optimal design when the optimal instrument
is used:
\begin{prop}
\label{prop:OptimalDesign}Define $\tilde{x}_{\delta}$ and $S_{\delta}$
as in Theorem \ref{prop:OptimalIV} and suppose there is $\delta\in\mathcal{D}$
such that $S_{\delta}\ne0$. Then 
\[
\delta^{*},z^{*}\in\arg\min_{\delta\in\mathcal{\mathcal{D}},z\in\mathcal{Z}_{\delta}}\max_{\mathcal{E}\in\mathcal{F}_{p}(\sigma)}\V{\delta,\mathcal{E}}z
\]
is solved by
\begin{align*}
\delta^{*} & \in\arg\max_{\delta\in\mathcal{D}}\tr\left(S^{\frac{p}{p+1}}_{\delta}\right)
\end{align*}
together with the optimal instrument $z^{\ast}=z^{\ast}_{\delta^{\ast}}$
from Theorem \ref{prop:OptimalIV}.
\end{prop}
\noindent This result follows directly from Theorem \ref{prop:OptimalIV}.
For $p=1$, the optimal design maximizes $\tr\left(\var{\delta}x^{1/2}\right)$.
At $p=\infty$, the sum of variances, $\tr\left(\var{\delta}x\right)$,
is maximized instead.

In what follows we tend to avoid $p=\infty$ as the optimal design
problem has a degenerate rank-1 solution. Specifically, Appendix \ref{appx:Rank-1-Designs}
shows that there exists $\delta^{*}$ that puts all the mass on just
one pair of complementary assignments $g^{*}\in\left\{ 0,1\right\} ^{K}$
and $\one-g^{*}$. Moreover, if $W$ is non-negative, one solution
corresponds to $g^{\ast}=\one$, i.e. treating all intervention units
with probability 0.5 and treating none otherwise. While this is not
problematic under spherical errors, it is highly undesirable when
errors can be correlated since the estimator may not be consistent.
Choosing $p<\infty$ guards against this possibility and helps with
consistency, as we show in Section \ref{subsec:clt-inference}.\footnote{Another problem with degenerate designs is that the asymptotic approximation
that makes the approximate variance a useful concept may fail. Moreover,
the estimator itself may not be well-defined with a probability that
is not asymptotically small, as $\sum_{i}z_{i}(g)x_{i}=0$ whenever
$g=\boldsymbol{0}$; see footnote \ref{fn:well-defined-estimator}.}

A\textbf{ }further result helps characterize optimal designs:
\begin{prop}
\label{prop:DesignProperties}There is a $\delta^{*}$ from Proposition
\ref{prop:OptimalDesign} with $\expec{\delta^{\ast}}{g_{k}}=0.5$
for all $k$. Further, if $W$ is block-diagonal after a joint partition
of rows and columns, there is such a $\delta^{*}$ with independent
subvectors of $g$ across blocks and the optimal design can be solved
separately block-by-block.
\end{prop}
\noindent Intuitively, any design with \emph{$\expec{\delta^{\ast}}{g_{k}}\ne0.5$}
can be symmetrized by averaging it with its complement and this can
never hurt the criterion. Similarly, any across-block dependence in
$g_{k}$ can get removed without hurting the criterion, since $S_{\delta^{*}}=W\var{\delta^{\ast}}gW^{\prime}$.\footnote{While $\mathcal{F}_{p}(\sigma)$ allows for error dependence even
across blocks, Proposition \ref{prop:DesignProperties} implies that
the same design would be optimal even if one were to \emph{a priori
}rule out cross-block dependence.} This result is generalized in Appendix \ref{appx:Symmetry-Normalization}:
whenever $W$ is invariant under a group of relabelings over rows
and columns, there exists a $\delta^{*}$ which is invariant under
the same group. In particular, if $W$ is invariant to permutations
of rows and columns in a cluster, one may restrict attention to exchangeable
designs within that cluster.

\subsection{Special Cases\label{subsec:Special-Cases}}

We now illustrate the results of Theorem \ref{prop:OptimalIV} and
Proposition \ref{prop:OptimalDesign} in several insightful special
cases.

\paragraph{No Spillovers.}

A useful benchmark is the standard experimental case with $N=K$ and
$W=I$. Since $W$ is diagonal, Proposition \ref{prop:DesignProperties}
immediately implies that independent random assignment, $g_{i}\stackrel{iid}{\sim}Bernoulli(0.5)$,
is optimal in this case, and the optimal instrument is just $z^{*}_{i}(g)\propto g_{i}-0.5$.\footnote{Note that complete randomization, with exactly half of the units treated,
is not optimal for $p<\infty$ as it induces a negative correlation
across assignments which can be exploited by the adversary choosing
the errors. See \citet{hull2026optimal} for a related analysis of
minimax experimental designs with bounded outcomes.}

\paragraph{Group-Specific Shocks.}

This and the remaining cases are visualized in Figure \ref{fig:SpecialCases}.
Consider a simple bipartite graph in which observations are partitioned
into $J$ groups, $C_{1},\dots,C_{J}$, with shocks assigned at the
group level: $K=J$ and $x_{i}=g_{j(i)}$ corresponding to $w_{i}=e_{j(i)}$,
where $e_{j}$ is the $j$th standard basis vector in $\mathbb{R}^{K}$
and $j(i)\in\left\{ 1,\dots,K\right\} $ gives the group of observation
$i$. Again Proposition \ref{prop:DesignProperties} implies independent
$g_{j}\stackrel{iid}{\sim}Bernoulli(0.5)$ is optimal. The optimal
instrument is now proportional to $N^{-1/(1+p)}_{j(i)}\cdot(g_{j(i)}-0.5)$
where $N_{j}=\left|C_{j}\right|$. Thus when $p=1$ the estimator
down-weights large groups by a factor of $1/\sqrt{N_{j(i)}}$ to ensure
the random variation is equally spread out; not doing so would make
the estimator noisy if the errors are very correlated within clusters.
With higher $p$ less of this reweighting is needed, with no reweighting
in the limit $p=\infty$.

\begin{figure}
\caption{Graphs in Special Cases\label{fig:SpecialCases}}

\begin{centering}
\begin{tabular}{ccc}
(a) Group-specific shocks\medskip{}
 &  & (b) Group-average shocks\tabularnewline
\begin{tikzpicture}[
    >=Stealth,
    obs/.style={circle, fill=black, inner sep=0pt, minimum size=4.5pt},
    grp/.style={circle, draw, line width=0.8pt, inner sep=1.2pt, minimum size=18pt, font=\small},
    edge/.style={thin}
]

\node[obs] (i1) at (0,3.0) {};
\node[left=3pt of i1] {\small $i$};

\node[obs] (i2) at (0,2.4) {};
\node[obs] (i3) at (0,1.8) {};
\node[obs] (i4) at (0,1.2) {};
\node[obs] (i5) at (0,0.3) {};
\node[obs] (i6) at (0,-0.3) {};
\node[obs] (i7) at (0,-1.2) {};
\node[obs] (i8) at (0,-1.8) {};

\node[grp] (j1) at (4,2.1) {$C_1$};
\node[grp] (j2) at (4,0.3) {$C_2$};
\node[grp] (j3) at (4,-1.5) {$C_3$};

\draw[edge] (i1) -- (j1);
\draw[edge] (i2) -- (j1);
\draw[edge] (i3) -- (j1);
\draw[edge] (i4) -- (j1);

\draw[edge] (i5) -- (j2);
\draw[edge] (i6) -- (j2);

\draw[edge] (i7) -- (j3);
\draw[edge] (i8) -- (j3);

\end{tikzpicture}\vspace{1cm}
 &  & \begin{tikzpicture}[
    >=Stealth,
    obs/.style={circle, fill=black, inner sep=0pt, minimum size=4.5pt},
    edge/.style={thin}
]

\node[obs] (l1) at (0,3.0) {};
\node[obs] (l2) at (0,2.4) {};
\node[obs] (l3) at (0,1.8) {};
\node[obs] (l4) at (0,1.2) {};
\node[obs] (l5) at (0,0.3) {};
\node[obs] (l6) at (0,-0.3) {};
\node[obs] (l7) at (0,-1.2) {};
\node[obs] (l8) at (0,-1.8) {};
\node[left=4pt of l1] {\small $i$};

\node[obs] (r1) at (4,3.0) {};
\node[obs] (r2) at (4,2.4) {};
\node[obs] (r3) at (4,1.8) {};
\node[obs] (r4) at (4,1.2) {};
\node[obs] (r5) at (4,0.3) {};
\node[obs] (r6) at (4,-0.3) {};
\node[obs] (r7) at (4,-1.2) {};
\node[obs] (r8) at (4,-1.8) {};
\node[right=4pt of r1] {\small $k$};


\foreach \a in {1,2,3,4}
  \foreach \b in {1,2,3,4}
    \draw[edge] (l\a) -- (r\b);

\foreach \a in {5,6}
  \foreach \b in {5,6}
    \draw[edge] (l\a) -- (r\b);

\foreach \a in {7,8}
  \foreach \b in {7,8}
    \draw[edge] (l\a) -- (r\b);

\end{tikzpicture}\tabularnewline
(c) Leave-out averages\medskip{}
 &  & (d) Spatial spillovers\tabularnewline
\begin{tikzpicture}[
    >=Stealth,
    obs/.style={circle, fill=black, inner sep=0pt, minimum size=4.5pt},
    edge/.style={thin}
]

\node[obs] (l1) at (0,3.0) {};
\node[obs] (l2) at (0,2.4) {};
\node[obs] (l3) at (0,1.8) {};
\node[obs] (l4) at (0,1.2) {};
\node[obs] (l5) at (0,0.3) {};
\node[obs] (l6) at (0,-0.3) {};
\node[obs] (l7) at (0,-1.2) {};
\node[obs] (l8) at (0,-1.8) {};
\node[left=4pt of l1] {\small $i$};

\node[obs] (r1) at (4,3.0) {};
\node[obs] (r2) at (4,2.4) {};
\node[obs] (r3) at (4,1.8) {};
\node[obs] (r4) at (4,1.2) {};
\node[obs] (r5) at (4,0.3) {};
\node[obs] (r6) at (4,-0.3) {};
\node[obs] (r7) at (4,-1.2) {};
\node[obs] (r8) at (4,-1.8) {};
\node[right=4pt of r1] {\small $k$};


\draw[edge] (l1) -- (r2);
\draw[edge] (l1) -- (r3);
\draw[edge] (l1) -- (r4);

\draw[edge] (l2) -- (r1);
\draw[edge] (l2) -- (r3);
\draw[edge] (l2) -- (r4);

\draw[edge] (l3) -- (r1);
\draw[edge] (l3) -- (r2);
\draw[edge] (l3) -- (r4);

\draw[edge] (l4) -- (r1);
\draw[edge] (l4) -- (r2);
\draw[edge] (l4) -- (r3);

\draw[edge] (l5) -- (r6);
\draw[edge] (l6) -- (r5);

\draw[edge] (l7) -- (r8);
\draw[edge] (l8) -- (r7);

\end{tikzpicture} &  & \raisebox{-1.5cm}[0pt][0pt]{%
\begin{tikzpicture}[
  x=1cm,y=1cm,
  line cap=round,
  line join=round,
  scale=0.90,
  transform shape
]
  \def\R{2.0}
  \def\N{10}

  \draw[line width=0.35pt,draw=black!60,dash pattern=on 2pt off 1.5pt]
    (0,0) circle (\R);

  \foreach \j in {0,...,9} {
    \coordinate (v\j) at ({90-360*\j/\N}:\R);
  }

  \newcommand{\drawbundleat}[3]{%
    \pgfmathtruncatemacro{\jmone}{mod(#1-1+\N,\N)}%
    \pgfmathtruncatemacro{\jmtwo}{mod(#1-2+\N,\N)}%
    \pgfmathtruncatemacro{\jpone}{mod(#1+1,\N)}%
    \pgfmathtruncatemacro{\jptwo}{mod(#1+2,\N)}%
    \pgfmathsetmacro{\ang}{90-360*(#1)/\N}%
    \begin{scope}[draw=#2,line width=#3]
      \draw (v#1) .. controls ($ (0,0)!2.08!(v#1) $) and ($ (0,0)!1.62!(v\jmtwo) $) .. (v\jmtwo);
      \draw (v#1) .. controls ($ (0,0)!1.72!(v#1) $) and ($ (0,0)!1.40!(v\jmone) $) .. (v\jmone);
      \draw (v#1) .. controls ($ (0,0)!1.72!(v#1) $) and ($ (0,0)!1.40!(v\jpone) $) .. (v\jpone);
      \draw (v#1) .. controls ($ (0,0)!2.08!(v#1) $) and ($ (0,0)!1.62!(v\jptwo) $) .. (v\jptwo);

      \coordinate (loopcenter) at ($ (0,0)!1.75!(v#1) $);
      \draw (v#1)
        .. controls ($ (loopcenter)+(\ang+90:0.60) $)
                  and ($ (loopcenter)+(\ang-90:0.60) $)
        .. (v#1);
    \end{scope}
  }

  \foreach \j in {1,...,9} {
    \drawbundleat{\j}{black!20}{0.45pt}
  }

  \drawbundleat{0}{black}{0.45pt}

  \foreach \j in {0,...,9} {
    \fill (v\j) circle (1.9pt);
  }
  \fill (v0) circle (2.5pt);
\end{tikzpicture}%
}\tabularnewline
\end{tabular}
\par\end{centering}
\begin{centering}
\bigskip{}
\bigskip{}
\par\end{centering}
{\small\emph{Notes}}{\small : This figure shows network graphs in the
special cases considered in Section \ref{subsec:Special-Cases}.}{\small\par}
\end{figure}

\paragraph{Group-Average Shocks.}

Now consider a case in which observations are partitioned into $J$
groups $C_{1},\dots,C_{J}$, but the shocks $g_{i}$ are assigned
to individuals ($K=N$) and fully propagate through the cluster-average:
$x_{i}=\bar{g}_{j(i)}$ for $\bar{g}_{j}=\frac{1}{N_{j}}\sum_{i\in C_{j}}g_{i}$,
which corresponds to $W=\diag(P_{N_{1}},\dots,P_{N_{J}})$ for $P_{n}=\mathbf{1}_{n}\mathbf{1}^{\prime}_{n}/n$.
This is isomorphic to the previous case, with shocks that are perfectly
correlated within groups and independent across groups being optimal:
$g_{i}=\check{g}_{j(i)}$ with $\check{g}_{j}\stackrel{iid}{\sim}Bernoulli(0.5)$.
The optimal instrument is again proportional to $N^{-1/(1+p)}_{j(i)}(\check{g}_{j(i)}-0.5)$.
Intuitively, there is no need to vary the shocks within groups since
treatment is at the group level; perfect within-group shock correlation
ensures the maximal variability of the spillover treatment while we
still reweight by $N^{-1/(1+p)}_{j(i)}$ to more evenly spread out
this variation.

\paragraph{Leave-Out Averages.}

Suppose in the previous case we instead have $N_{j}>1$ for all $j$
and $x_{i}=\bar{g}_{-i}=\frac{1}{N_{j(i)}-1}\sum_{k\in C_{j(i)},k\neq i}g_{k}$.
Appendix \ref{appx:Leave-Out-Averages} shows that an optimal design
induces a within-group shock correlation of $\psi_{j}=(1-r_{j})/(1+(N_{j}-1)r_{j})$
for $r_{j}=(N_{j}-1)^{-2p}$. This correlation can be implemented
by a variety of designs, for example a saturation design that mixes
over saturation rates of 0, 0.5, and 1. This design is close to a
more conventional saturation design (studied, e.g., in \citet{Baird2018})
with saturation rates of $0$, $0.5$, and $1$, except with independent
(rather than complete) randomization within groups. The optimal instrument
is now proportional to $\left(1+(N_{j(i)}-1)\psi_{j(i)}\right){}^{-1/(1+p)}\cdot\left((\bar{g}_{j(i)}-\frac{1}{2})+(N_{j(i)}-1)(\bar{g}_{j(i)}-g_{i})\right)$.
It is easy to verify that it is optimal to correlate the shocks more
in larger groups: $\psi_{j}$ increases in $N_{j}$, from $\psi_{j}=0$
when $N_{j}=2$ to $\psi_{j}\to1$ when $N_{j}\to\infty$. Intuitively,
in the smallest group of size $N_{j}=2$, we are effectively back
in the no-spillover case, while with $N_{j}\to\infty$ the leave-out
average is equivalent to the group average from the previous case.
For intermediate sizes, within-group correlation increases in $p$:
as before, it is higher when less clustering is allowed in $\varepsilon$.

\paragraph{Spatial Spillovers.}

Finally, suppose the $N=K$ units are equally spaced points on a circle,
with treatment being an average of shocks to units within the distance
of $L$: $x_{i}=\frac{1}{2L+1}\sum^{i+L}_{k=i-L}g_{k}$, where index
summation is understood modulo $N$. As in Appendix \ref{appx:Spatial-Spillovers},
assume for simplicity that $2L+1$ divides $N$ and $2L+1\le N/2$.
The appendix shows that, for $p=1$, the optimal design correlates
shock assignments at the same spatial scale as the exposure mapping.
Specifically, the correlation between assignments at circular distance
$d$ is $\max\left\{ 1-d/(2L+1),0\right\} $, linearly decaying in
$d$. The design can be implemented by splitting the circle into blocks
of $2L+1$ consecutive units, randomly shifting this partition, and
assigning shocks at the block level. The optimal instrument is $z^{\ast}=W^{\dag}\tilde{g}$.
Under the shifted-block implementation of the design, it takes values
$(2L+1)\tilde{g}_{i}-2L\bar{\tilde{g}}$ for units at the center of
a block and $\bar{\tilde{g}}$ for all other units, where $\bar{\tilde{g}}$
is the average of all recentered shocks. Intuitively, the IV undoes
the smoothing done by spatial averaging to mitigate the correlations
that could otherwise be exploited by the adversary. These results
can be extended to a two-dimensional grid, with $x_{i}$ averaging
the shocks within a certain distance in each coordinate.

\subsection{\label{subsec:Direct-Effects}Multiple Exposures}

We now consider a model in which the shocks $g$ affect the outcomes
through multiple linear exposure measures. For brevity, we consider
two exposures but all results generalize immediately to three or more.
For a known fixed $N\times K$ matrix $U=(u_{i}')^{N}_{i=1}$, we
now have:
\begin{align}
y_{i} & =\beta x_{i}+\tau h_{i}+\varepsilon_{i},\label{eq:multiple_exposure}\\
x_{i} & =w^{\prime}_{i}g,\qquad h_{i}=u^{\prime}_{i}g.\nonumber 
\end{align}
Our goal is still to minimize the worst-case approximate estimation
variance of the effect of $x_{i}$.

While we consider a general formulation, this setting is especially
relevant when the intervention and outcome units are the same and
$\tau$ corresponds to the direct effect of the intervention:
\begin{equation}
y_{i}=\beta x_{i}+\tau g_{i}+\varepsilon_{i}\label{eq:model_direct}
\end{equation}
which corresponds to $U=I_{N}$.\footnote{Another class of applications is to multiplex networks (cf. \citet{zenou2025peer}).}
In this context, there are three reasons why one would focus on the
estimation power of $\beta$. First, the researcher may be most interested
in the spillover effect, but still wants to acknowledge the presence
of the direct effect. Second, even if direct and spillover effects
are of equal interest, the spillover effect is usually more challenging
to estimate. A design that yields sufficient power for it would usually
deliver precise estimates of the direct effect, too. Finally, if the
researcher is interested in the Global Average Treatment Effect (GATE),
defined as the effect of switching from $g=\boldsymbol{0}_{K}$ to
$g=\boldsymbol{1}_{K}$ via both direct and indirect effects, that
problem can be reformulated as the one we solve here.\footnote{Specifically, the GATE equals $\gamma=\bar{w}\beta+\bar{u}\tau$ for
known constants $\bar{w}=\frac{1}{N}\sum_{i,k}w_{ik}$ and $\bar{u}=\frac{1}{N}\sum_{i,k}u_{ik}$.
Provided $\bar{w}\ne0$, such that the indirect effect matters, one
can then rewrite (\ref{eq:multiple_exposure}) as $y_{i}=\gamma\frac{x_{i}}{\bar{w}}+\tau\left(h_{i}-\frac{\bar{u}}{\bar{w}}x_{i}\right)+\varepsilon_{i}$
and apply our results to this reformulation, optimizing estimation
power for $\gamma$ while also including $h_{i}-\frac{\bar{u}}{\bar{w}}x_{i}$
in the model.}

To see how Theorem \ref{prop:OptimalIV} and Propositions \ref{prop:OptimalDesign}–\ref{prop:DesignProperties}
generalize, collect the parameters and explanatory variables in $\theta=(\beta,\tau)^{\prime}$
and $\boldsymbol{x}_{i}=(x_{i},h_{i})^{\prime}$. Absent restrictions
on how the error term can correlate with $W$ and $U$, we again use
recentered IV. For some recentered instrument vector $\boldsymbol{z}(\cdot)=(\boldsymbol{z}_{i}(\cdot))^{N}_{i=1}$,
where $\boldsymbol{z}_{i}(\cdot):\{0,1\}^{K}\rightarrow\mathbb{R}^{2}$
and $\expec{\delta}{\boldsymbol{z}_{i}(g)}=0$, consider the estimator:
\begin{align*}
\hat{\theta}\left[\boldsymbol{z}\right] & =\left(\sum^{N}_{i=1}\boldsymbol{z}_{i}(g)\boldsymbol{x}^{\prime}_{i}\right)^{-1}\sum^{N}_{i=1}\boldsymbol{z}_{i}(g)y_{i}.
\end{align*}
Define the approximate variance of $\hat{\beta}$ as:
\begin{align*}
\boldsymbol{\mathcal{V}}_{\delta,\mathcal{E}}\left[\boldsymbol{z}\right] & =e^{\prime}_{1}\expec{\delta}{\sum^{N}_{i=1}\boldsymbol{z}_{i}(g)\boldsymbol{x}^{\prime}_{i}}^{-1}\var{\delta,\mathcal{E}}{\sum^{N}_{i=1}\boldsymbol{z}_{i}(g)\varepsilon_{i}}\expec{\delta}{\sum^{N}_{i=1}\boldsymbol{z}_{i}(g)\boldsymbol{x}^{\prime}_{i}}^{-1\prime}e_{1},
\end{align*}
where $e_{1}=(1,0)^{\prime}$, with $\boldsymbol{\mathcal{V}}_{\delta,\mathcal{E}}\left[\boldsymbol{z}\right]=\infty$
whenever $\expec{\delta}{\sum^{N}_{i=1}\boldsymbol{z}_{i}(g)\boldsymbol{x}^{\prime}_{i}}$
is singular. Also, as before, define
\begin{align*}
\mathcal{V}_{\delta,\mathcal{E}}[z] & =\frac{\var{\delta,\mathcal{E}}{\sum^{N}_{i=1}z_{i}(g)\varepsilon_{i}}}{\expec{\delta}{\sum^{N}_{i=1}z_{i}(g)x_{i}}^{2}}.
\end{align*}

An initial result shows that the problem of searching over recentered
IV vectors $\boldsymbol{z}$ is equivalent to a constrained version
of the original problem of searching over scalar $z$ that requires
$z$ to be orthogonal to the “nuisance” exposure $h$:
\begin{lem}
\label{lemma:directFX}Fix $\delta$ and $\mathcal{E}$. Let $\mathcal{Z}^{(2)}_{\delta}$
be the set of recentered $\boldsymbol{z}$ and suppose $\tr(\var{\delta}h)>0$
for $h=(h_{i})^{N}_{i=1}$. Then:
\begin{align*}
\left\{ \boldsymbol{\mathcal{V}}_{\delta,\mathcal{E}}\left[\boldsymbol{z}\right]:z\in\mathcal{Z}^{(2)}_{\delta}\right\}  & =\left\{ \V{\delta,\mathcal{E}}z:z\in\mathcal{Z}_{\delta},\expec{\delta}{z^{\prime}h}=0\right\} .
\end{align*}
That is, the set of attainable $\boldsymbol{\mathcal{V}}_{\delta,\mathcal{E}}\left[\boldsymbol{z}\right]$
equals the set of attainable $\V{\delta,\mathcal{E}}z$ under the
constraint $\expec{\delta}{z^{\prime}h}=0$. Moreover, this equivalence
is constructive in both directions: any $\boldsymbol{z}\in\mathcal{Z}^{(2)}_{\delta}$
can be transformed into an instrument vector whose first component
is some scalar $z\in\mathcal{Z}_{\delta}$ with $\expec{\delta}{z'h}=0$,
without changing the value of the objective, while any such scalar
$z$ can be implemented by the vector instrument $\boldsymbol{z}_{i}(g)=\bigl(z_{i}(g),\ h_{i}-\expec{\delta}{h_{i}}\bigr)$.
\end{lem}
Imposing the orthogonality constraint yields the analog of Theorem
\ref{prop:OptimalIV}:

\renewcommand{\thethm}{1$\,^\prime$} 
\begin{thm}
\label{prop:OptimalIV_direct}Fix $\delta$ and, for $\Sigma_{\delta}=\var{\delta}g$,
let 
\begin{align}
c_{\delta} & \in\arg\min_{c\in\mathbb{R}}\tr\left(\left((W-cU)\Sigma_{\delta}(W-cU)^{\prime}\right)^{p/(p+1)}\right),\label{eq:c_delta}\\
\tilde{x}^{\perp}_{\delta} & =(W-c_{\delta}U)(g-\expec{\delta}g),\qquad S^{\perp}_{\delta}=(W-c_{\delta}U)\Sigma_{\delta}(W-c_{\delta}U)^{\prime}.\nonumber 
\end{align}
Suppose $S^{\perp}_{\delta}\neq0$ and $\tr(U\Sigma_{\delta}U')>0$.
Then, for any $p\in[1,\infty]$ and $\sigma>0$:

(a) The worst-case approximate variance with the optimal instrument
is:
\[
\min\limits_{\bm{z}\in\mathcal{Z}^{(2)}_{\delta}}\max_{\mathcal{E}\in\mathcal{F}_{p}(\sigma)}\mathcal{\boldsymbol{\mathcal{V}}}_{\delta,\mathcal{E}}\left[\bm{z}\right]=N^{1/p}\sigma^{2}\tr\left(\left(S^{\perp}_{\delta}\right)^{p/(p+1)}\right)^{-(p+1)/p}.
\]

(b) Assume that either $p>1$ or $p=1$ and $\rank\!\left((W-c_{\delta}U)\Sigma^{1/2}_{\delta}\right)=\min\left\{ N,\rank(\Sigma_{\delta})\right\} $.
Then the instrument $\boldsymbol{z}^{\ast}_{i}(g)=(z^{*}_{\delta}(g)_{i},h_{i}-\expec{\delta}{h_{i}})$
minimizes $\max_{\mathcal{E}\in\mathcal{F}_{p}(\sigma)}\boldsymbol{\mathcal{V}}_{\delta,\mathcal{E}}\left[\boldsymbol{z}\right]$
over $\mathcal{Z}^{(2)}_{\delta}$, where:
\begin{align*}
z^{*}_{\delta}(g) & =\left((S^{\perp}_{\delta})^{\dag}\right)^{1/(p+1)}\tilde{x}^{\perp}_{\delta}.
\end{align*}
\end{thm}
\renewcommand{\theprop}{1$\,^\prime$} 
The key difference from the optimal instrument in the single-exposure
case is that the matrix $W$ is replaced with an appropriate residualization
$W-c_{\delta}U$ where $c_{\delta}$ (which plays the role of a Lagrange
multiplier on the $\expec{\delta}{z^{\prime}h}=0$ constraint) can
be solved for by the scalar optimization (\ref{eq:c_delta}). The
additional technical assumption for $p=1$ is imposed because the
objective (\ref{eq:c_delta}) can be non-differentiable, which may
lead to a different analytical form of $z^{*}_{\delta}(g)$. To avoid
this issue, we require that the residualized exposure matrix retains
as much assignment-induced variation as its dimensions allow. Given
orthogonality between the optimal IV for $x_{i}$ and the exposure
$h_{i}$, the IV for $h_{i}$ is simply $h_{i}-\expec{\delta}{h_{i}}$.

Analogs of Propositions \ref{prop:OptimalDesign}–\ref{prop:DesignProperties}
also follow:
\begin{prop}
\label{prop:OptimalDesign_direct}In the Theorem \ref{prop:OptimalIV_direct}
setting, suppose there is $\delta^{*}\in\mathcal{D}$ such that
\begin{align*}
\delta^{*} & \in\arg\max_{\delta\in\mathcal{D}}\tr{\left((S^{\perp}_{\delta})^{p/(p+1)}\right)},
\end{align*}
with $\tr(U\Sigma_{\delta^{*}}U')>0$ and $\min_{c\in\mathbb{R}}\tr{\left(((W-cU)\Sigma_{\delta^{*}}(W-cU)^{\prime})^{p/(p+1)}\right)}>0$.
Then for $p>1$ 
\[
\delta^{*},\boldsymbol{z}^{*}\in\arg\min_{\delta\in\mathcal{D},\boldsymbol{z}\in\mathcal{Z}^{(2)}_{\delta}}\max_{\mathcal{E}\in\mathcal{F}_{p}(\sigma)}\mathcal{\boldsymbol{\mathcal{V}}}_{\delta,\mathcal{E}}\left[\boldsymbol{z}\right]
\]
where $\boldsymbol{z}^{\ast}=\boldsymbol{z}^{\ast}_{\delta^{\ast}}$
is the optimal instrument from Theorem \ref{prop:OptimalIV_direct}.
This also holds for $p=1$ provided \textup{$\rank\!\left((W-c_{\delta}U)\Sigma^{1/2}_{\delta}\right)=\min\left\{ N,\rank(\Sigma_{\delta})\right\} $
}.
\end{prop}
\renewcommand{\theprop}{2$\,^\prime$} 
\begin{prop}
\label{prop:DesignProperties_direct}There is a $\delta^{*}$ from
Proposition \ref{prop:OptimalDesign_direct} with $\expec{\delta^{\ast}}{g_{k}}=0.5$
for all $k$. Further, if $W$ and $U$ are jointly block diagonal
under the same row-column partition, there is such a $\delta^{*}$
with independent subvectors of $g$ across blocks.
\end{prop}
\renewcommand{\theprop}{\arabic{prop}}\setcounter{prop}{2} \renewcommand{\thethm}{\arabic{thm}}\setcounter{thm}{1}

\noindent We note that, unlike with Proposition \ref{prop:DesignProperties},
the optimal design cannot generally be solved block-by-block in the
block-diagonal case because $c_{\delta}$ is common across all blocks.

We illustrate these results in a simple but nontrivial special case:
when $x_{i}$ is the leave-out average of shocks in the cluster to
which $i$ belongs and $h_{i}=g_{i}$ captures the direct effect.
Appendix \ref{appx:Leave-Out-Averages} shows that, at least when
$N_{j}$ is even, the optimal amount of assignment correlation within
the cluster can be implemented by a saturation design with saturation
rates 0, 0.5, and 1, albeit with probabilities that generally differ
from the case without the direct effect. In the special case where
all clusters are of the same size, $N_{j}=n$, the optimal within-cluster
correlation takes a closed form: $\psi^{\perp}=(1-r^{\perp})/(1+(n-1)r^{\perp})$
for $r^{\perp}=(n-1)^{-2p/(2p-1)}$. The optimal spillover instrument
is proportional to $(\bar{g}_{j(i)}-0.5)+(n-1)^{1/(2p-1)}(\bar{g}_{j(i)}-g_{i})$.
Relative to the case without direct effects, there is generally more
randomization within clusters ($\psi^{\perp}\le\psi$, with strict
inequality for $n>2$ and $p>1$), which helps isolate spillover effects
from the direct effects. The exception is when $p=1$: with both
homogeneous and heterogeneous $N_{j}$, the design and instrument
are then the same as in the original problem.\footnote{\label{fn:direct-doesnt-change-design}This follows from two properties
of the single-exposure $z^{\ast}$ under $p=1$. First, full whitening
means that $z^{\ast}_{i}$ is uncorrelated with all $z^{\ast}_{j}$
for all $j\ne i$. Second, applying Corollary \ref{cor:p1}, below,
to the leave-out exposure implies that $\tilde{g}_{i}$ can be represented
as a weighted sum of $z^{\ast}_{j}$ of $j\ne i$. Thus, $z^{\ast}_{i}$
is uncorrelated with $\tilde{g}_{i}$ for each $i$ and hence $\expec{\delta}{(z^{\ast})'\tilde{g}}=0$.
The additional orthogonality constraint in the two-exposure problem
is therefore non-binding at the single-exposure optimum, so the same
design and spillover IV remain optimal.}

\subsection{Other Extensions}

We derive five further extensions to the baseline results. First,
Appendix \ref{appx:predetermined_covs} shows how the results extend
when the researcher includes an intercept and possibly other predetermined
covariates $r_{i}$ in estimation. We formalize this idea by enriching
the class of error distributions to $\varepsilon_{i}=r^{\prime}_{i}\gamma+\tilde{\varepsilon}_{i}$
where the distribution of $\tilde{\varepsilon}=\left(\tilde{\varepsilon}_{i}\right)$
belongs to $\mathcal{F}_{p}(\sigma)$ while $\gamma$ is unrestricted.
The results in Theorem \ref{prop:OptimalIV} and Propositions \ref{prop:OptimalDesign}–\ref{prop:DesignProperties}
continue to hold, with the exposure matrix $W$ residualized columnwise
on the covariates.

Second, Appendix \ref{appx:nonlinear} discusses how Theorem \ref{prop:OptimalIV}
and Proposition \ref{prop:OptimalDesign} extend to the general formula
treatment setting where the $x_{i}$ are nonlinear (but still known)
$i$-specific formulas combining $W$ and $g$. This allows, for example,
$x_{i}$ to be an indicator for $i$ having at least one treated friend
in their social network or the more elaborate formulas considered
in \citet{BH1}. The proofs of Theorem \ref{prop:OptimalIV} and Proposition
\ref{prop:OptimalDesign} turn out to extend verbatim after defining
the general $\tilde{x}_{\delta}=x-\expec{\delta}x$ and $S_{\delta}=\var{\delta}x$,
making them a robust characterization of optimal design with formula
treatments. Proposition \ref{prop:DesignProperties}, however, does
not extend cleanly: it is no longer without loss of generality to
consider designs with $0.5$ marginals. The following discussion of
feasible design computation and inference is also reliant on linearity
of the spillover treatment formula.

Third, Appendix \ref{appx:budgets} considers settings where the researcher
faces a constraint of a marginal treatment probability $q\neq0.5$
(common across units), e.g. due to a fixed budget or government target.
Theorem \ref{prop:OptimalIV} applies to any design and is unaffected
by the constraints on feasible designs. We therefore show that Proposition
\ref{prop:OptimalDesign} extends naturally with such constraints,
as well as the second claim of Proposition \ref{prop:DesignProperties}
and the results on computation, below.

Fourth, Appendix \ref{appx:direct_fx_constraint} considers a multiple-exposure
setting where the researcher wants to optimize the worst-case precision
of the indirect effect estimate while ensuring the worst-case precision
of the direct effect estimate stays above some baseline value (e.g.,
its precision under independent randomization and homoskedastic errors).
We show this leads to a natural extension of the unconstrained problem
in Section \ref{subsec:Direct-Effects}, which can also be used to
optimize a weighted average of worst-case variances of indirect and
direct effects.

Finally, Appendix \ref{appx-adaptivity} considers adaptive IV choice,
after outcomes are observed. While the design is inevitably chosen
\emph{ex ante}, without detailed knowledge of the error structure,
the choice of instrument can in principle be informed by outcome data.
For example, using the unwhitened recentered exposure $\tilde{x}$
as the IV can yield a more efficient estimator when the errors are
found to be close to spherical. We show that oracle adaptation does
not affect the ex-ante minimax design choice whenever the adaptive
procedure does not increase the approximate variance under any $\mathcal{E}\in\mathcal{F}_{p}(\sigma)$,
compared to the optimally whitened estimator. In particular, the Proposition
\ref{prop:OptimalDesign} design remains optimal even if the researcher
could use an infeasible optimal IV ex-post. More practically, we discuss
how one can combine the Proposition \ref{prop:OptimalDesign} optimal
design with an optimally-weighted generalized method of moments (GMM)
estimator based on two instruments: the partially whitened $z^{\ast}$
from Theorem \ref{prop:OptimalIV} and the unwhitened $\tilde{x}$.
Their variance-minimizing combination preserves the optimal design
while adapting towards the unwhitened IV when the error structure
is simple.\footnote{\label{fn:GMM-finite-sample}However, estimating the efficient GMM
weighting matrix can lead to small-sample biases \citep{hansen1996finite}.}

\section{Feasible Designs and Inference \label{sec:computation_clt}}

We now introduce a relaxed version of the optimal design problem.
This relaxation admits a closed-form solution in special cases and
a computationally efficient approximation in general. Moreover, this
approximation can yield $\sqrt{K}$-consistent recentered IV estimators
with $p<\infty$, as well as a simple asymptotic inference procedure.

\subsection{The Relaxed Problem\label{subsec:Computation}}

Solving the Proposition \ref{prop:OptimalDesign} problem can be computationally
intractable. While its objective depends on the experimental design
$\delta$ only through the $K\times K$ covariance matrix $\var{\delta}g$,
the class of all $K\times K$ covariance matrices achievable with
binary treatments, $\mathcal{C}_{K}=\left\{ \var{\delta}g\colon\delta\in\mathcal{D}\right\} $,
is high-dimensional. The computation time required to find an optimal
$\var{\delta}g$, by searching over convex combinations of the $2^{K}$
possible treatment vectors in $\left\{ 0,1\right\} ^{K}$, is generally
exponential in $K$.

To make progress, we follow \citet{goemans_improved_1995} and \citet{thiyageswaran_optimal_2026}
in considering a relaxed problem with a polynomial-time solution in
$K$. Specifically, we replace the optimization over $\mathcal{C}_{K}$
with optimization over $\mathcal{Q}_{K}=\left\{ \Sigma\in\mathbb{S}^{K}:\Sigma\succeq0,\Sigma_{kk}=1/4\text{ for all }k\right\} $,
where $\mathbb{S}^{K}$ is the set of $K\times K$ real and symmetric
matrices. Thus $\mathcal{Q}_{K}$ is the set of $K\times K$ covariance
matrices with $1/4$ on the diagonal, containing $\mathcal{C}_{K}$
when we restrict attention to designs with $\expec{\delta}{g_{k}}=0.5$
(without loss of generality, by Proposition \ref{prop:DesignProperties}).
The relaxed Proposition \ref{prop:OptimalDesign} problem is then:
\begin{equation}
\max_{\Sigma\in\mathcal{Q}_{K}}\tr\left((W\Sigma W')^{\frac{p}{p+1}}\right).\label{eq:relax2}
\end{equation}

Searching over $\mathcal{Q}_{K}$ is feasible for moderate $K$ using
an interior point solver for conic optimization problems. Two additional
results further improve computational efficiency without any additional
approximation cost. First, if $W$ can be split into independent blocks
$W_{(1)},\dots,W_{(B)}$ with $K(b)$ shocks in each, we can use Proposition
\ref{prop:DesignProperties} to solve the optimal design problem separately
by block. Second, as the following result shows, we can characterize
the solution explicitly in terms of a lower-dimensional convex optimization
problem over the $K-1$-dimensional simplex $\Delta_{K}=\left\{ \ell\in\mathbb{R}^{K}\colon\ell_{k}\ge0,\sum^{K}_{k=1}\ell_{k}=1\right\} $
rather than $\mathcal{Q}_{K}$; this problem is amenable to gradient-based
optimization approaches that scale more efficiently with $K$.
\begin{prop}
\label{prop:Optimize_l}For $p<\infty$, let 
\[
\Sigma^{\ast}=\frac{1}{4}D^{-1/2}_{\ell^{\ast}}\frac{M^{p}_{\ell^{\ast}}}{\tr\left(M^{p}_{\ell^{\ast}}\right)}D^{-1/2}_{\ell^{\ast}},
\]
where, for $\ell>0$,
\[
D_{\ell}=\diag\left(\ell_{1},\dots,\ell_{K}\right),\qquad M_{\ell}=D^{-1/2}_{\ell}W'WD^{-1/2}_{\ell},
\]
and 
\begin{equation}
\ell^{\ast}\in\arg\min_{\ell\in\text{\emph{int}(}\Delta_{K})}\tr\left(M^{p}_{\ell}\right),\label{eq:lstar}
\end{equation}
with $\ell^{*}>0$. Then $\Sigma^{\ast}$ solves the relaxed problem
(\ref{eq:relax2}). Moreover, the Theorem \ref{prop:OptimalIV} optimal
IV corresponding to $\Sigma^{\ast}$ satisfies $\tilde{g}_{k}\propto\sum_{i}w_{ik}z^{\ast}_{i}/\ell^{\ast}_{k}$.

\emph{This characterization is especially helpful in two special cases
where it yields closed-form solutions of the relaxed problem:}
\end{prop}
\begin{cor}
\label{cor:p1}For $p=1$, the relaxed problem (\ref{eq:relax2})
is solved by:
\[
\Sigma^{\ast}_{kl}=\frac{1}{4}\frac{w_{\cdot k}'w_{\cdot l}}{\left\Vert w_{\cdot k}\right\Vert _{2}\left\Vert w_{\cdot l}\right\Vert _{2}},\qquad k,l\in\left\{ 1,\dots,K\right\} ,
\]
where $w_{\cdot k}$ is the $k$th column of $W$. That is, in the
relaxed problem's solution, the correlation of assignments of any
two shocks equals the cosine similarity of the vectors of exposures
to those shocks. The corresponding optimal IV satisfies $\tilde{g}_{k}\propto\sum_{i}w_{ik}z^{\ast}_{i}/\left\Vert w_{\cdot k}\right\Vert _{2}$.
\end{cor}
\begin{cor}
\label{cor:sym_diagonal}For $p<\infty$, if all diagonal elements
of $\left(W'W\right)^{p}$ are equal, the relaxed problem (\ref{eq:relax2})
is solved by:
\[
\Sigma^{\ast}=\frac{1}{4}\frac{\left(W'W\right)^{p}}{\tr{\left(\left(W'W\right)^{p}\right)}/K}.
\]
The corresponding optimal IV is $z^{\ast}\propto a\cdot\left(W'\right)^{\dag}\tilde{g}$
which depends on $p$ only through $a=\left(4\tr{\left(\left(W'W\right)^{p}\right)}/K\right)^{1/(p+1)}$.
Moreover, when $W$ consists of multiple independent blocks with the
diagonal elements of $\left(W'W\right)^{p}$ equal within blocks,
these expressions apply block-by-block, with $a$ possibly varying
across blocks.
\end{cor}
Corollary \ref{cor:p1} captures a simple intuition: shocks to two
intervention units should be correlated to the extent that they are
connected, in the sense that the same outcome units are exposed to
them. This turns out to be the optimal solution to the relaxed problem
when $p=1$. Corollary \ref{cor:sym_diagonal} extends this result
to arbitrary $p$, as long as each block of $W$ is sufficiently symmetric—a
restrictive condition that nevertheless includes all of the special
cases in Section \ref{subsec:Special-Cases}. The corresponding solution
is particularly intuitive when $W$ is nonnegative and $p$ is an
integer. In that case, the matrix $(W'W)^{p}$ captures $p$th-degree
connections between intervention units $k$. When $p=2$, for example,
the solution suggests correlating two shocks when they are either
connected directly or connected to the same third shock. The worst-case
loss from such longer-range shock correlations is lower with higher
$p$, as error correlations are more restricted. As $p\to\infty$,
this solution correlates all shocks in the same connected component
of the bipartite network defined by $W$.

Proposition \ref{prop:Optimize_l} and its corollaries further characterize
the IV that is optimal if the relaxed problem's solution can be implemented
(an issue we discuss below). In the simplest case of Corollary \ref{cor:sym_diagonal},
this optimal IV satisfies $W'z^{\ast}\propto\tilde{g}$, meaning that
partial whitening entails the inverse operation to the exposure mapping,
$\tilde{x}=W\tilde{g}$. For instance, when $x$ involves some averaging
of the (correlated) shocks, $z^{\ast}$ performs their deconvolution.\footnote{Interestingly, the optimal IV is invariant to $p$, up to the proportionality
constant $a$. With larger $p$, the shocks are more strongly correlated;
thus, the same IV that fully whitens the shocks when $p=1$ does only
partial whitening for larger $p$. We also note that when $W$ consists
of independent blocks the constant $a$ varies across blocks; thus,
$p$ affects the reweighting of blocks but not the optimal IV within
blocks. The optimal IV is similar in the general case of Proposition
\ref{prop:Optimize_l}; the rescaling by $\ell^{\ast}_{k}$ ensures
$\var{\delta}{\tilde{g}_{k}}=1/4$.}

All of these results generalize immediately to settings with an intercept
and possibly other predetermined covariates in estimation, by first
residualizing $W$ on those covariates column by column. For instance,
for $p=1$ and when estimation includes an intercept as the only covariate,
Corollary \ref{cor:p1} applied to the demeaned $w_{\cdot k}$ implies
that shock correlations under the relaxed problem's solution are equal
to the correlations of the $N\times1$ exposures to those shocks (rather
than cosine similarities without the intercept). This solution positively
correlates shocks with similar exposure while introducing negative
correlations between shocks with non-overlapping exposure.\footnote{That is, if $w_{\cdot k},w_{\cdot l}\ge0$ and $w_{\cdot k}'w_{\cdot l}=0$
(no exposure overlap), $\Sigma^{\ast}_{kl}$ is proportional to the
sample correlation of $w_{\cdot k}$ and $w_{\cdot l}$, which is
negative.}

We can also generalize the relaxed problem to multiple exposures,
where it remains computationally efficient. By the same arguments
as in the Proposition \ref{prop:OptimalDesign_direct} proof, a relaxed
version of the multiple exposures problem is:
\begin{equation}
\min\limits_{c\in\mathbb{R}}\max_{\Sigma\in\mathcal{Q}_{K}}\tr{\Big(}\left((W-cU)\Sigma(W-cU)'\right)^{\frac{p}{p+1}}\Big).\label{eq:relaxp}
\end{equation}
Here the inner problem has the same structure as before with a single
exposure, with $W$ replaced by $W-cU$; the computationally efficient
form of Proposition \ref{prop:Optimize_l} and its closed-form special
cases thus apply.\footnote{The Proposition \ref{prop:Optimize_l} problem becomes $\min_{c\in\mathbb{R}}\min_{\ell\in\interior{\Delta_{K}}}\tr\left\{ \left(D^{-1/2}_{\ell}\left(W-cU\right)'\left(W-cU\right)D^{-1/2}_{\ell}\right)^{p}\right\} $.}
Additionally, when $W$ and $U$ are jointly block-diagonal, the inner
problem can be split block-by-block for any $c$. The outer problem
is then a scalar optimization, amenable to a golden section search
over a compact subset of $\mathbb{R}$.\footnote{The minimizer is guaranteed to exist in a compact subset of $\mathbb{R}$
as long as $\max\limits_{\Sigma\in\mathcal{Q_{K}}}\tr{\Big(}\left((W-cU)\Sigma(W-cU)'\right)^{\frac{p}{p+1}}\Big)\ge\tr{\Big(}\left((W-cU)I/4(W-cU)'\right)^{\frac{p}{p+1}}\Big)\to\infty$
when $\left|c\right|\to\infty$, which requires only that $U\neq0$.}

\subsection{Feasible Implementation\label{subsec:Feasible-Implementation}}

Two challenges remain once a solution $\Sigma^{\ast}$ to the relaxed
problems in Equation (\ref{eq:relax2}) or (\ref{eq:relaxp}) is found.
First, this $\Sigma^{\ast}$ may not correspond to an element of $\mathcal{C}_{K}$:
i.e., it may not be implementable via binary shock vectors. Second,
even if it is theoretically implementable, there is no known computationally
feasible procedure for sampling binary vectors from a generic covariance
matrix. These problems can be overcome in simpler special cases: the
designs in Section \ref{subsec:Special-Cases} are indeed implementations
of the relaxed problem solution.

Outside special cases, we address both challenges by using a simple
Gaussian rounding procedure to sample binary assignments with the
variance matrix approximating $\Sigma^{\ast}$, following \citet{goemans_improved_1995}.
Specifically, we draw $\xi\sim\mathcal{N}(0,\Sigma^{\ast})$ and set
$g^{GR}_{k}=\one{\left[\xi_{k}\ge0\right]}{}$.\footnote{This procedure simplifies further in the Corollary \ref{cor:p1} case:
draw $\eta\sim\mathcal{N}(0,I_{N})$ and set $g^{GR}_{k}=\one{\left[(W'\eta)_{k}\ge0\right]}{}$.
Indeed, $g^{GR}_{k}=\one{\left[\xi_{k}\ge0\right]}{}$ for $\xi_{k}=\sum_{i}w_{ik}\eta_{i}/2\left\Vert w_{\cdot k}\right\Vert _{2}$,
where $(\xi_{1},\dots,\xi_{K})\sim\mathcal{N}(0,\Sigma^{\ast})$ for
$\Sigma^{\ast}$ from Corollary \ref{cor:p1}.} This procedure distorts correlations in a known way: $\corr{}{g^{GR}_{k},g^{GR}_{l}}=\arcsin\left[\corr{}{\xi_{k},\xi_{l}}\right]/\frac{\pi}{2}$,
which does not equal $\corr{}{\xi_{k},\xi_{l}}$ in general; see \citet{goemans_improved_1995}
and the illustration in Appendix Figure \ref{fig:arcsin}. Given this
Gaussian rounding design, we find the optimal IV based on $S_{GR}=W\var{}{g^{GR}}W^{\prime}$
for
\begin{equation}
\Sigma^{GR}=\var{}{g^{GR}}=\frac{1}{2\pi}\arcsin\left[4\Sigma^{\ast}\right],\label{eq:V}
\end{equation}
with $\arcsin\left[\cdot\right]$ applied entrywise. Figure \ref{fig:alg_summary}
summarizes the inputs and outputs of the procedure; see Algorithm
\ref{alg:full} in the appendix for a more detailed summary.

\begin{figure}
\caption{\label{fig:alg_summary}Inputs and Outputs of the Proposed Algorithm}
\smallskip{}

\begin{centering}
\resizebox{\linewidth}{!}{%
\begingroup
\definecolor{rctdiagink}{HTML}{25333B}
\definecolor{rctdiagmuted}{HTML}{52616A}
\definecolor{rctdiagrule}{HTML}{B9C2C8}
\definecolor{rctdiaginputfill}{HTML}{EDF3F9}
\definecolor{rctdiagoutputfill}{HTML}{EDF3F9}
\definecolor{rctdiagoutputrule}{HTML}{A1B6CC}
\definecolor{rctdiagalgorithmfill}{HTML}{DFEAF6}
\definecolor{rctdiagalgorithmrule}{HTML}{8CA9C6}
\definecolor{rctdiagblueink}{HTML}{294D73}

\newcommand{\rctdiagcardtitle}[1]{{\fontsize{12.5}{14.5}\selectfont\bfseries #1}}
\newcommand{\rctdiagbody}[1]{\\[4pt]{\color{rctdiagmuted}#1}}
\newcommand{\rctdiagone}{\mathbf{1}}

\normalfont
\linespread{1}\selectfont
\frenchspacing
\fontsize{11}{13.5}\selectfont\color{rctdiagink}
\begin{tikzpicture}[
  card/.style={draw=rctdiagrule,line width=.6pt,rounded corners=5pt,
               inner xsep=10pt,inner ysep=8pt,align=left},
  input/.style={card,draw=rctdiagoutputrule,fill=rctdiaginputfill,text width=208pt,
                minimum width=228pt},
  output/.style={card,draw=rctdiagoutputrule,fill=rctdiagoutputfill,text width=200pt,
                 minimum width=220pt},
  arrow/.style={-{Stealth[length=5pt,width=4pt]},draw=rctdiagmuted,line width=.8pt},
  optional/.style={dash pattern=on 3pt off 2pt}
]
  \node[input] (w) at (-228pt,100pt)
    {\rctdiagcardtitle{Indirect exposure}\\[7pt]
     {\fontsize{14}{16}\selectfont
       \makebox[\linewidth]{$\displaystyle
         x_i=\sum_k
         \underbrace{\vphantom{w_{ik}g_k}w_{ik}}_{\color{rctdiagmuted}\substack{
           \text{\vphantom{Ag}known exposure}\\\text{\vphantom{Ag}of unit }i\\\text{\vphantom{Ag}to treatment }k}}
         \;\cdot\;
         \underbrace{\vphantom{w_{ik}g_k}g_k}_{\color{rctdiagmuted}\substack{
           \text{\vphantom{Ag}binary}\\\text{\vphantom{Ag}treatment}\\\text{\vphantom{Ag}assignment}}}
       $}}};

  \node[input,optional,fill=white,below=14pt of w] (other)
    {\rctdiagcardtitle{Other exposures}\hfill
     {\sffamily\fontsize{8.5}{10}\selectfont\color{rctdiagmuted}OPTIONAL}
     \rctdiagbody{Direct or additional indirect, if included}};

  \node[input,below=14pt of other] (target)
    {\rctdiagcardtitle{Target effect}
     \rctdiagbody{Coefficient on $x_i$ by default}};

  \node[input,below=14pt of target] (p)
    {\rctdiagcardtitle{Tuning parameter $p\in[1,\infty)$}
     \rctdiagbody{Lower $p$: more robust to correlated\\
       and heteroskedastic errors.\\[2pt]
       Default: $p=2$}};

  \coordinate (stackmid) at ($(w.north)!.5!(p.south)$);
  \coordinate (maincenter) at (stackmid -| 0pt,0pt);
  \node[card,draw=rctdiagalgorithmrule,fill=rctdiagalgorithmfill,text width=144pt,
        minimum width=164pt,minimum height=96pt,align=center] (alg) at (maincenter)
    {{\color{rctdiagblueink}\rctdiagcardtitle{Algorithm A1}}\\[6pt]
     Minimize worst-case\\ approximate variance\\ for target effect};

  \node[output] (sig) at ($(maincenter)+(224pt,59pt)$)
    {{\color{rctdiagblueink}\rctdiagcardtitle{Randomization design}}
     \rctdiagbody{Captured by covariance matrix $\Sigma^*$\\[3pt]
       Draw $(\xi_1,\dots,\xi_K)\sim\mathcal{N}(0,\Sigma^*)$\\[3pt]
       Assign shocks $g_k=\rctdiagone[\xi_k\geq 0]$}};

  \node[output] (z) at ($(maincenter)+(224pt,-57pt)$)
    {{\color{rctdiagblueink}\rctdiagcardtitle{Recentered instrument $z_i^*$}}\\[1pt]
     {\color{rctdiagblueink}\fontsize{12.5}{14.5}\selectfont\bfseries for indirect exposure $x_i$}
     \rctdiagbody{Estimate $\hat\beta$ and $\operatorname{SE}(\hat\beta)$}};

  \draw[arrow,rounded corners=3pt]
    (w.east) -- ++(16pt,0) |- ([yshift=44pt]alg.west);
  \draw[arrow,optional] (other.east) -- (other.east -| alg.west);
  \draw[arrow,rounded corners=3pt]
    (target.east) -- ++(20pt,0) |- ([yshift=-12pt]alg.west);
  \draw[arrow,rounded corners=3pt]
    (p.east) -- ++(24pt,0) |- ([yshift=-42pt]alg.west);

  \draw[arrow,rounded corners=3pt]
    ([yshift=30pt]alg.east) -- ++(16pt,0) |- (sig.west);
  \draw[arrow,rounded corners=3pt]
    ([yshift=-26pt]alg.east) -- ++(16pt,0) |- (z.west);

  \draw[arrow] (sig.south) -- (z.north)
    node[midway,fill=white,draw=rctdiagrule,line width=.5pt,rounded corners=4pt,
         inner xsep=6pt,inner ysep=3pt,text=rctdiagmuted,
         font=\fontsize{10.5}{12}\selectfont] {Run experiment; observe outcomes};
\end{tikzpicture}
\endgroup%
}\bigskip{}
\par\end{centering}
{\small\emph{Notes}}{\small : This figure summarizes the inputs and
outputs of the feasible implementation. See Algorithm \ref{alg:full}
in the appendix for details.}{\small\par}
\end{figure}

Beyond simplicity in implementation, this procedure offers two advantages.
First, Appendix \ref{appx:approx_bound} shows the cost of using this
relaxed and rounded version of the problem is bounded: the worst-case
approximate variance of the resulting estimator is at most $\pi/2$
times that of the original problem.\footnote{We find much smaller approximation errors in applications: for all
designs in Sections \ref{subsec:Cai}–\ref{subsec:mk}, the worst-case
approximate variance exceeds the variance of the original problem
by no more than 6\%. This is computed by comparing the worst-case
variance arising from the Gaussian rounded design to that of the unrounded
solution to the relaxed problem (\ref{eq:relax2}), which is a lower
bound for the variance of the optimal solution to the original problem.} Second, as we next show, the fact that our shocks are a simple transformation
of a correlated Gaussian vector is helpful for establishing a central
limit theorem (CLT) for the recentered IV estimator and asymptotically
valid inference, even when all errors are mutually correlated.

\subsection{\label{subsec:clt-inference}Asymptotic Normality and Inference}

We now consider the asymptotic behavior of the estimator $\hat{\beta}=\hat{\beta}\left[z^{\ast}\right]$
corresponding to the Gaussian rounding implementation of the Section
\ref{subsec:Computation} relaxed-optimal design, with the Theorem
\ref{prop:OptimalIV} optimal IV $z^{\ast}$. We show that choosing
a small integer $p$ can yield sparsity of the implied shock covariance
matrix, which in turn can yield a CLT for $\hat{\beta}$ and a simple
inference procedure.

Let $d(\Sigma)=\max_{k}\sum^{K}_{l=1}\mathbf{1}\left\{ \Sigma_{kl}\neq0\right\} $
measure the sparsity of a covariance matrix $\Sigma$ by the maximum
number of non-zero covariances across all rows. We consider a sequence
of data-generating processes indexed by $K$. For $K\to\infty$ (which
under Assumption \ref{ass:asymptotic} will require $N\to\infty$)
we assume:

\begin{assumption}
\label{ass:asymptotic}There exist $\underline{\lambda}>0$, \textup{$\bar{\lambda}\in(0,\infty)$,}
$\bar{d}<\infty$, $h^{*}\in(0,\infty)$, and $v^{*}\in(0,\infty)$
such that:
\begin{enumerate}
\item[(a)] \textup{$d(\Sigma^{\ast})\le\bar{d}$ and }$\lambda_{\max}(W'W)\le\bar{\lambda}$
\textup{;}
\item[(b)] \textup{}$\lambda^{+}_{\min}(S_{GR})>\underline{\lambda}$ with $\lambda^{+}_{\min}(\cdot)$
denoting the smallest positive eigenvalue;
\item[(c)] \textup{}$h_{K}=K^{-1}\tr(S^{p/(p+1)}_{GR})\toP h^{*}$;
\item[(d)] \textup{}$v_{K}=K^{-1}\varepsilon'S^{(p-1)/(p+1)}_{GR}\varepsilon\stackrel{p}{\rightarrow}v^{*}$;
\item[(e)] \textup{}$\frac{1}{K^{3/2}}\sum^{K}_{k=1}|b_{k}|^{3}=o_{p}(1)$ for
$b=W'(S^{\dag}_{GR})^{1/(p+1)}\varepsilon$,
\end{enumerate}
where statements (a) and (b) hold with probability $1-o(1)$ and all
probability statements are with respect to the joint distribution
of $(W,\varepsilon,\xi)$.
\end{assumption}
The first part of Assumption \ref{ass:asymptotic}(a) is the key sparsity
condition imposed on the solution to the relaxed problem in Equation
(\ref{eq:relax2}). While the sparsity of $\Sigma^{*}$ is formally
an asymptotic condition, the researcher can heuristically check whether
$d(\Sigma^{\ast})$ is small compared to $K$; at the cost of more
complex primitive conditions, we can accommodate $d(\Sigma^{\ast})$
growing slowly in $K$.\footnote{Theorem \ref{thm:main} can be extended to cases where $\Sigma^{*}$
is approximately sparse, in that it contains many small entries but
has absolute row and column sums that are bounded or grow very slowly.
This extension is particularly useful when predetermined covariates
are included, since then the effective $W$ and corresponding $\Sigma^{*}$
will not generally be sparse. In such cases we could allow the operator
norm of $\Sigma^{*}$ to grow slowly by a smoothing argument and the
second order Poincaré inequality of \citet{chatterjee2009fluctuations},
as applied in \citet{lei2018asymptotics}. We thank Lihua Lei for
pointing this out.} Since $\frac{1}{N}\left\Vert x\right\Vert ^{2}_{2}=\frac{1}{N}\left\Vert Wg\right\Vert ^{2}_{2}\le\lambda_{\max}(W'W)\cdot\frac{1}{K}\left\Vert g\right\Vert ^{2}\cdot\frac{K}{N}$,
the second part precludes $x$ from diverging provided $K\asymp N$.
Appendix \ref{app:Theorem1} provides sufficient conditions for both
parts of this assumption. It shows that $d(\Sigma^{\ast})$ is upper-bounded
by $\left(d(W'W)\right)^{p}$ under the assumptions of Proposition
\ref{prop:Optimize_l}, for integer $p$. With $p=\infty$, no sparsity
of $\Sigma^{\ast}$ is guaranteed within any connected component of
the graph, and so designs that are optimal for $p=\infty$ may not
yield consistent $\hat{\beta}$ when $W$ does not consist of many
independent blocks and when errors are also strongly dependent. But
for smaller integer $p$, the optimal design only correlates units
that are connected by paths in $W'W$ of lengths $p$ or less. Thus,
for finite integer $p$, both parts of Assumption \ref{ass:asymptotic}(a)
are guaranteed when maximum row and column degree of $W$ as well
as the maximum absolute value of its elements are all bounded. This
set of conditions in turn requires $K\asymp N$ in non-degenerate
cases, as also shown in Appendix \ref{app:Theorem1}.

In addition to sparsity, Assumption \ref{ass:asymptotic} imposes
regularity conditions to ensure the joint distribution of spillover
treatments $x_{i}$ and errors $\varepsilon_{i}$ is well-behaved.
Assumption \ref{ass:asymptotic}(b) precludes near-collinearity of
exposures to ensure that the whitening matrix remains well-conditioned.
Yet, this condition allows $S_{GR}$ to be degenerate, which happens,
e.g., when some observations have the same exposure to all shocks.
Noting that $h_{K}=\expec{\delta}{\frac{1}{K}x'z^{\ast}\mid W}$,
Assumption \ref{ass:asymptotic}(c) implies that $W$ and the experimental
design are such that there is sufficient variation in the spillover
treatments $x_{i}$ so that the first stage is not degenerate asymptotically.
Noting further that $v_{K}=\frac{1}{K}\var{\delta}{\varepsilon'z^{\ast}\mid W,\varepsilon}$,
Assumption \ref{ass:asymptotic}(d) ensures that the finite-sample
variance of $\varepsilon'z^{\ast}$ conditional on $W$ and $\varepsilon$
(scaled appropriately) converges to a deterministic constant. This
is a law of large numbers for the quadratic form $v_{K}$ and it imposes
regularity on the joint distribution of $W$ and $\varepsilon$. For
example, conditional on a sequence of $W$, if $\varepsilon$ is drawn
from a mixture of two deterministic sequences that lead to different
limits for $v_{K}$, then this assumption is violated; under suitable
moment and stabilization conditions, sufficiently weak or local dependence
can ensure this assumption. Assumption \ref{ass:asymptotic}(e) is
an anti-concentration condition for the implicit weights of the numerator
$\frac{1}{K}\varepsilon'z^{\ast}$; combined with the sparsity condition,
it ensures that the numerator can be asymptotically approximated by
sums of independent components.

Under these conditions, $\widehat{\beta}$ is $\sqrt{K}$-consistent:
\begin{thm}
\label{thm:main} Under Assumption~ \ref{ass:asymptotic} and the
Gaussian rounding design based on the relaxed problem's solution from
Proposition \ref{prop:Optimize_l}, the Theorem \ref{prop:OptimalIV}
optimal recentered IV estimator $\hat{\beta}$ satisfies 
\[
\sqrt{K}(\bhat-\beta)\toD\mathcal{N}\!\left(0,\frac{v^{*}}{h^{*2}}\right).
\]
\end{thm}
\noindent Asymptotically valid inference is straightforward from
this result using the normal approximation and a simple plug-in estimator
for the standard error, $\sqrt{\hat{v}/(Kh^{2}_{K})}$. Here $S_{GR}$
and therefore $h_{K}$ are known. Moreover, the plug-in estimator
for $v^{\ast}$ is straightforward to compute: $\hat{v}=\frac{1}{K}\hat{\varepsilon}\,'S^{(p-1)/(p+1)}_{GR}\hat{\varepsilon}$
for $\hat{\varepsilon}=y-\hat{\beta}x$. The convergence in probability
of $\hat{v}$ to $v^{*}$ follows from consistency of $\widehat{\beta}$
and Assumption \ref{ass:asymptotic} (see Appendix \ref{app:Theorem1}).
Appendix \ref{app:Theorem2_direct} provides a CLT for the multiple
exposures case and the corresponding plug-in standard errors.

\section{Applications\label{sec:Applications}}

\subsection{\label{subsec:Cai}Bipartite Experiment: Cai et al. (2015)}

\begin{figure}
\caption{Spillover Structure and Optimal Designs in Applications\label{fig:Applications}}

\begin{centering}
(a) An example village in \citet{cai_social_2015}, $K=34$, $N=44$
\par\end{centering}
\begin{centering}
\includegraphics[width=0.75\textwidth]{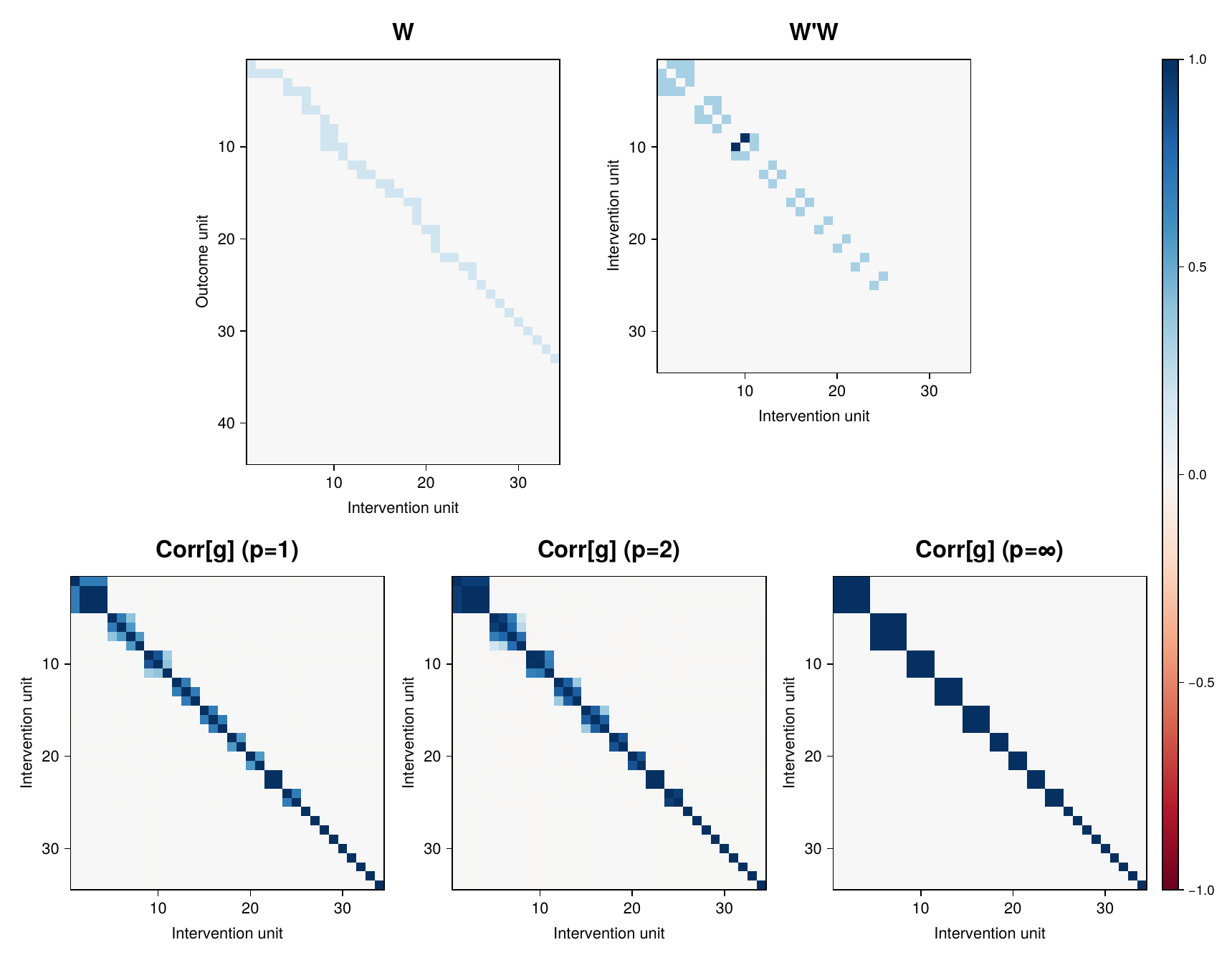}
\par\end{centering}
\begin{centering}
\medskip{}
(b) \citet{Miguel2004}, $K=N=49$
\par\end{centering}
\begin{centering}
\includegraphics[width=0.75\textwidth]{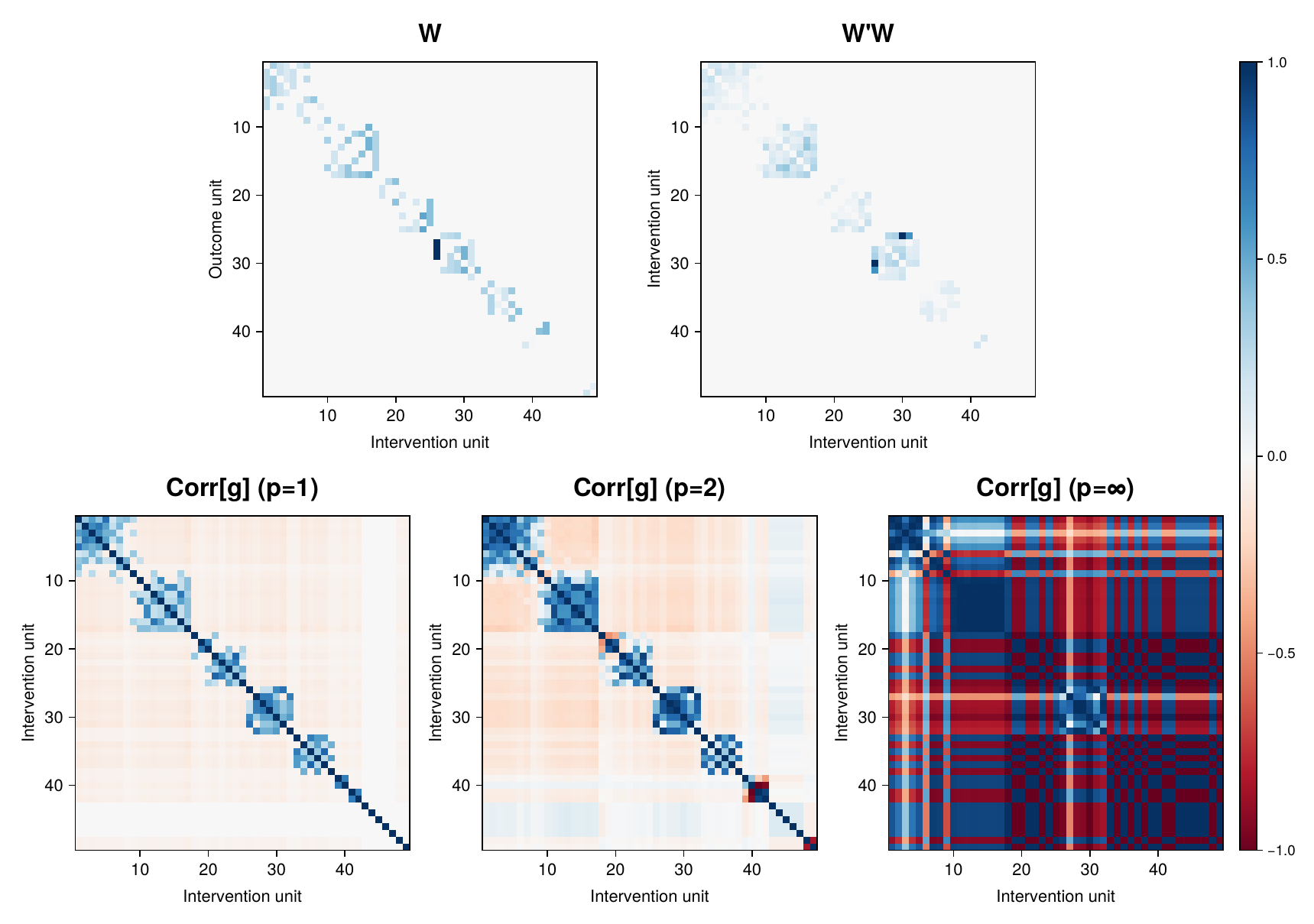}\smallskip{}
\par\end{centering}
{\small\emph{Notes}}{\small : Panel (a) shows an example of one village
(Xinlian) in the \citet{cai_social_2015} data. The first row shows
the exposure matrix $W$ and the matrix $W'W$ with the $kl$ element
measuring the extent to which $k$ and $l$ are friends of the same
farmers $i$. For $W'W$, the diagonal is replaced with zero and other
values are normalized by dividing by the maximum off-diagonal value.
The second row shows the correlation matrices of the Algorithm \ref{alg:full}
feasible optimal designs corresponding to $p=1$ and $p=2$, along
with the analogous solutions for $p=\infty$. Panel (b) reports the
same objects for the full set of schools in \citet{Miguel2004}, except
$W$ is also normalized to be between 0 and 1 by dividing each matrix
element-wise by its maximum value.}{\small\par}
\end{figure}

\paragraph{Setup.}

Our first application is the bipartite experiment of \citet{cai_social_2015}
who study how social networks affect the adoption of weather insurance
by randomly assigning rice farmers in villages in China to simple
vs. intensive information sessions over two rounds. They collect information
on demographics, friendship, and insurance take-up for each farmer.
We calibrate a simulation to a simplified version of the social network
effect specification in Table 2, Column 2 of their paper. For the
sample of farmers assigned to either the simple or intensive information
session in the second round, it examines how insurance take-up varies
with the fraction of friends that are treated in the first round.
Specifically, our outcome units $i$ are those randomized in the second
round ($N=1,274$),\footnote{More precisely, outcome units are those used to estimate spillover
effects on insurance take-up (from “Type 1” villages) rather
than price effects.} and our intervention units $k$ are those randomized in the first
round with at least one friend in the second round ($K=995$). As
in \citet{cai_social_2015}, we define $x_{i}=\sum_{k}w_{ik}g_{k}$
where $w_{ik}$ equals an indicator that $i$ and $k$ are friends
divided by the total number of friends $i$ has (including friends
outside the first-round sample). Thus, $x_{i}$ measures the percentage
of the second round unit's friends who were treated in the first round.
Each $i$ and $k$ is assigned to one of 44 administrative villages,
which are administrative units that each combine multiple local communities;
throughout this analysis we drop 10 of 6,251 friendships that are
across administrative villages so $W$ has block-diagonal structure
that eases computational burden. The top row of Figure \ref{fig:Applications}(a)
shows the exposure matrix $W$ for one example village, along with
the $W'W$ matrix that captures similarity in exposure weights across
intervention units.

We estimate the specification $y_{i}=\alpha+\beta x_{i}+\upsilon_{i}$
using recentered IV without whitening. Our simple specification has
a smaller estimated coefficient of $\beta^{\ast}=0.1408$ compared
to the specification in \citet{cai_social_2015}, which uses OLS as
an estimator and includes an additional set of controls and village
fixed effects, rather than recentering.

For $S=10,000$ simulations, we generate the outcome as $y_{i}=\beta^{*}x_{i}+\varepsilon_{i}$,
where $g_{k}$ varies by the experimental design, $W$ is fixed, and
we draw $\varepsilon=(\varepsilon_{i})$ from six different data-generating
processes (DGPs), i.e., joint distributions that vary the structure
of the $N\times N$ matrix $\expec{\mathcal{E}}{\varepsilon\varepsilon^{\prime}}$.\footnote{While the original outcomes in \citet{cai_social_2015} are binary,
our generated outcomes are continuous.} For the first five DGPs, the distribution of the errors is calibrated
so that the Monte Carlo standard deviation of the unwhitened recentered
IV estimator of $\beta$ under a 50\% Bernoulli RCT matches the plug-in
standard error of that estimator in the data. The error distributions
for the simulation are:
\begin{itemize}
\item Homoskedastic: $\varepsilon_{i}\stackrel{iid}{\sim}\mathcal{N}(0,\nu^{2})$,
with $\nu=0.53$;
\item Heteroskedastic: $\varepsilon_{i}\sim\mathcal{N}(0,\nu^{2}(1+\gamma d_{i}))$,
independent across units $i$, where $d_{i}$ is the number of $i$'s
first-round friends, $\nu=0.32$ and $\gamma=0.82$, which implies
50\% of the variance is driven by the heteroskedastic component;
\item Degree-in-Mean: $\varepsilon_{i}=\gamma d_{i}+\nu\zeta_{i}$, where
$\zeta_{i}\stackrel{iid}{\sim}\mathcal{N}(0,1)$, $\gamma=0.27$,
and $\nu=0.29$, with $\gamma$ calibrated such that 50\% of the cross-sectional
variation of $\varepsilon_{i}$ arises from the mean;
\item Village-Correlated: $\varepsilon_{i}=\nu\cdot(\zeta_{i}+u_{v(i)})$,
where $\nu=0.30$, $v(i)$ is the village of individual $i$, and
$u_{v}\stackrel{iid}{\sim}\mathcal{N}(0,1)$, $\zeta_{i}\stackrel{iid}{\sim}\mathcal{N}(0,1)$;
\item Network-Correlated: $\varepsilon=\nu\zeta+\gamma W\eta$, where $\zeta_{i}\stackrel{iid}{\sim}\mathcal{N}(0,1)$
and $\eta_{k}\stackrel{iid}{\sim}\mathcal{N}(0,1)$. The parameters
$\gamma=1.19$ and $\nu=0.27$ are calibrated so that 50\% of the
variance comes from the network component;
\item Estimated Residuals: $\varepsilon_{i}=\hat{\upsilon}_{i}$, the residuals
from the regression on the real data.
\end{itemize}

\paragraph{Results.}

We solve for the feasible optimal design and recentered IV estimator
of the indirect effect for several choices of $p$. The bottom row
of Figure \ref{fig:Applications}(a) plots a correlation heatmap for
these designs. As our theory suggests, $p=\infty$ involves perfectly
correlating all treatments within clusters of connected farmers (with
independent randomization across villages), while lower $p$ implies
more diffuse correlation structures.\footnote{Appendix Figure \ref{fig:frontier}(a) shows how the optimal designs
trade off isotropy, as measured by $\tr(S_{\delta^{*}})/\lambda_{\max}(S_{\delta^{*}})$,
and signal $\tr(S_{\delta^{*}})$. As expected the frontier moves
down and to the right as $p\rightarrow\infty$. Notably, the independent
Bernoulli design (indicated by the red dot) has higher isotropy than
any optimal design but also much lower signal.}

\begin{table}[tp]
\caption{RMSE and Effective Sample Size Gains for the \citet{cai_social_2015}
Application\label{tab:cai}}

\begin{centering}
\begin{tabular}{l*{6}{c}}
\toprule
 & \shortstack{Homo- \\ skedastic} & \shortstack{Hetero- \\ skedastic} & \shortstack{Deg.-in- \\ Mean} & \shortstack{Village- \\ Corr.} & \shortstack{Network- \\ Corr.} & \shortstack{Estimated \\ Residuals} \\
 & (1) & (2) & (3) & (4) & (5) & (6) \\
\midrule
Effective Sample Gain: \\
\quad Optimal, $p=1$ & +19\% & +33\% & +40\% & +70\% & +70\% & +32\% \\
\quad Optimal, $p=2$ & +55\% & +60\% & +44\% & +84\% & +78\% & +67\% \\
\quad Optimal, $p=\infty$ & +119\% & +92\% & +1\% & +31\% & +70\% & +85\% \\
\addlinespace
RMSE for RCT Benchmark & 0.135 & 0.133 & 0.131 & 0.133 & 0.132 & 0.133 \\
\bottomrule
\end{tabular}
\smallskip{}
\par\end{centering}
{\small\emph{Notes}}{\small : This table reports, for the \citet{cai_social_2015}
application, differences in effective sample size for different optimal
designs and estimators relative to the benchmark of the Bernoulli
RCT with the unwhitened recentered IV estimator. The rows correspond
to the Algorithm \ref{alg:full} feasible optimal designs and optimal
IV with $R=\one{}{}$ (the intercept), for different finite values
of $p$ (and the analogous solutions for $p=\infty$). The last row
gives the RMSE of the benchmark design and estimator. The columns
are the six error DGPs described in the text. Differences in effective
sample sizes are computed, for each DGP, as the squared ratio of inverse
RMSE of a given scenario relative to the benchmark, minus one.}{\small\par}
\end{table}

Table \ref{tab:cai} reports the gain in effective sample size for
the optimal recentered IV estimator from Theorem \ref{prop:OptimalIV};
the gain in effective sample size is defined as the squared inverse
RMSE ratio relative to the benchmark of a simple Bernoulli experiment
with a 50\% treatment probability and a recentered IV estimator without
whitening.\footnote{That is, we report $100\left[\left(\frac{\mathrm{RMSE}_{\mathrm{benchmark}}}{\mathrm{RMSE}_{\mathrm{optimal}}}\right)^{2}-1\right]$.
All estimators include an intercept: i.e., Algorithm \ref{alg:full}
is applied with $R=\one{}{}$. We use the optional part of step 2
to block-diagonalize the residualized $W$. This allows us to solve
for the optimal design village by village, with independent randomization
across villages.} The columns correspond to the error structures described above, while
the rows correspond to the optimal designs. The last row reports the
RMSE of the benchmark, which represents the standard deviation of
the recentered IV estimator since the bias is negligible.

All optimal designs beat a standard RCT, which is not optimized to
estimate spillover effects. Under homoskedastic errors, the design
with $p=\infty$ (i.e., perfectly correlating treatments within blocks)
works best. But under more complex and realistic error structures,
the advantage of extreme shock correlations deteriorates: the $p=2$
design is best for the DGP in columns 3–5 with network-dependent or
correlated errors. This choice of $p$ also shows robust power gains
across all columns, equivalent to increasing the sample size by 44–84\%
in all columns.

To isolate gains from the optimal design vs. the optimal estimator,
we first note that the $p=\infty$ design involves no whitening yet
achieves substantial improvements across most error distributions—suggesting
a key role of design. Appendix Table \ref{tab:cai-extra} more directly
examines the role of whitening by reporting effective sample size
gains of recentered IV without whitening for the $p=1$ and $p=2$
optimal designs. Whitening helps with more complex error distributions
(columns 3–5) while reducing efficiency in other cases. The table
also shows results from GMM estimation using both whitened and unwhitened
instruments, as discussed in Appendix \ref{appx-adaptivity}. As expected,
the overidentified IV specification tends to be close to—and sometimes
slightly better than—the best of the two alternatives.

Appendix Table \ref{tab:Coverage-Cai} further studies the coverage
of confidence intervals built using the normal approximation from
Theorem \ref{thm:main}. Since the network of \citet{cai_social_2015}
is made up of many independent villages, all designs and simulations
have close to nominal coverage, even when shocks are perfectly correlated
within villages under $p=\infty$. Appendix \ref{subsec:appx-more-Cai}
reports a range of other supplementary results. It illustrates the
proposed designs in one simulation draw, visualizes the variances
and mutual covariances of $\tilde{x}_{i}$ and their whitened versions,
and measures worst-case variance gains for a range of error distributions
of increasing complexity. 

\subsection{\label{subsec:mk}Direct and Indirect Effects: Miguel and Kremer
(2004)}

\paragraph{Setup.}

Our second application is in the setting of \citet{Miguel2004} who
estimate the direct and spillover effects of a randomized school-level
deworming treatment on helminth infection rates among children in
Kenya. The intervention shocks $g_{i}$ are at the school level; slightly
simplifying the original analysis, we define the outcome $y_{i}$
as the average infection rate in school $i$. Thus, the outcome and
intervention units are the same set of $N=K=49$ schools (specifically,
Groups 1 and 2 of the experiment for which health outcomes are observed).
Following \citet{Miguel2004}, the spillover treatment $x_{i}$ is
the number (rather than the share) of students in nearby schools that
were treated, rescaled by 1,000. This is formalized by the $N\times N$
adjacency matrix $W$ with entry $w_{ik}=E_{k}/1000$ if school $i$
is within 3km of school $k\ne i$, where $E_{k}$ is the number of
pupils eligible for the treatment in school $k$, and 0 otherwise
(including $w_{ii}=0$ for all $i$). The top row of Figure \ref{fig:Applications}(b)
shows the exposure matrix $W$ and the matrix of similarity in exposure
weights $W'W$.

We calibrate our simulation to a simplified version of column (1)
in Table VII of \citet{Miguel2004}. We estimate a school-level regression
$y_{i}=\alpha+\beta x_{i}+\tau g_{i}+\upsilon_{i}$ using recentered
IV without whitening. This yields coefficients $\beta^{\ast}=-0.226$
and $\tau^{\ast}=-0.258$.\footnote{Our specification differs from \citet{Miguel2004} in several ways.
We use a linear model rather than probit, we recenter rather than
including additional controls, and we drop a third treatment variable
based on exposure to schools between 3 and 6 km away. Despite these
simplifications, our estimates are very similar to the corresponding
average marginal effects in the original paper, of -0.26 for the indirect
effect and -0.25 for the direct effect.}

For $S=10,000$ simulations, we generate the outcome as
\begin{equation}
y_{i}=\beta^{*}x_{i}+\tau^{*}g_{i}+\varepsilon_{i},\label{eq:MK_specification}
\end{equation}
where again the distribution of $g$ depends on the experimental design,
$W$ is fixed, and we vary $\expec{\mathcal{E}}{\varepsilon\varepsilon'}$
according to the same six DGPs described above, with a few minor changes.
We replace the village-correlated DGP with a “cluster-correlated”
one by creating five clusters of schools using spectral clustering
on $W$ and use random cluster shocks. We calibrate the simulations
like in the first application: to match the estimated standard error
of the unwhitened recentered IV estimator and targeting the same shares
of variation explained by different error components.\footnote{The corresponding parameters are{\small{} as follows. Homoskedastic:
$\nu=0.25$; heteroskedastic: $\gamma=0.4$, $\nu=0.17$; degree-in-mean:
$\gamma=0.081$, $\nu=0.12$, cluster-correlated: $\nu=0.14$; network-correlated:
$\gamma=0.17,\nu=0.11$.}}

\begin{table}[tp]
\caption{RMSE and Effective Sample Size Gains for the \citet{Miguel2004} Application\label{tab:RMSE-MK}}

\begin{centering}
(a) Spillover Effect\smallskip{}
\par\end{centering}
\begin{centering}
\begin{tabular}{l*{6}{c}}
\toprule
 & \shortstack{Homo- \\ skedastic} & \shortstack{Hetero- \\ skedastic} & \shortstack{Deg.-in- \\ Mean} & \shortstack{Cluster- \\ Corr.} & \shortstack{Network- \\ Corr.} & \shortstack{Estimated \\ Residuals} \\
 & (1) & (2) & (3) & (4) & (5) & (6) \\
\midrule
Effective Sample Gain: \\
\quad Optimal, $p=1$ & +3\% & +16\% & +92\% & +94\% & +169\% & +118\% \\
\quad Optimal, $p=2$ & +106\% & +122\% & +190\% & +164\% & +237\% & +172\% \\
\quad Optimal, $p=\infty$ & +293\% & +298\% & +430\% & +156\% & +246\% & +202\% \\
\addlinespace
RMSE for RCT Benchmark & 0.152 & 0.159 & 0.165 & 0.162 & 0.159 & 0.188 \\
\bottomrule
\end{tabular}
\bigskip{}
\par\end{centering}
\begin{centering}
(b) Direct Effect\smallskip{}
\par\end{centering}
\begin{centering}
\begin{tabular}{l*{6}{c}}
\toprule
 & \shortstack{Homo- \\ skedastic} & \shortstack{Hetero- \\ skedastic} & \shortstack{Deg.-in- \\ Mean} & \shortstack{Cluster- \\ Corr.} & \shortstack{Network- \\ Corr.} & \shortstack{Estimated \\ Residuals} \\
 & (1) & (2) & (3) & (4) & (5) & (6) \\
\midrule
Effective Sample Gain: \\
\quad Optimal, $p=1$ & -20\% & -19\% & -6\% & -13\% & -21\% & -12\% \\
\quad Optimal, $p=2$ & -15\% & -12\% & +3\% & -12\% & -21\% & -8\% \\
\quad Optimal, $p=\infty$ & -3\% & +5\% & -4\% & -7\% & -22\% & -56\% \\
\addlinespace
RMSE for RCT Benchmark & 0.075 & 0.072 & 0.058 & 0.062 & 0.049 & 0.064 \\
\bottomrule
\end{tabular}
\smallskip{}
\par\end{centering}
{\small\emph{Notes}}{\small : Panel (a) reports, for the \citet{Miguel2004}
application, differences in effective sample size for different optimal
designs and estimators for the spillover effect relative to the benchmark
of the Bernoulli RCT with the unwhitened recentered IV estimator.
Panel (b) reports the gains for the direct effect under the same design
and estimation choice. The rows correspond to the Algorithm \ref{alg:full}
feasible optimal designs and optimal IV with $R=\one$ (the intercept)
in the presence of direct effects, for different finite values of
$p$ (and the analogous solutions for $p=\infty$). The last row of
each panel gives the RMSE of the benchmark design and estimator. The
columns are the six error DGPs described in the text. Differences
in effective sample sizes are computed, for each DGP, as the squared
ratio of inverse RMSE of a given scenario relative to the benchmark,
minus one.}{\small\par}
\end{table}

\paragraph{Results.}

We solve for the feasible design optimized for the \emph{indirect}
effect in a linear model (\ref{eq:MK_specification}) that also includes
a direct effect. As in the first application, we include the intercept
for all estimators.\footnote{We again use Algorithm \ref{alg:full} with $R=\one{}{}$. Here, however,
we do not need the optional part of step 2.}

The bottom row of Figure \ref{fig:Applications}(b) presents the correlation
heatmap for the optimal designs used in the simulations. Now that
the direct effect is in the model, we no longer have a perfectly correlated
design when $p=\infty$. As $p$ decreases, the optimal design from
Proposition \ref{prop:OptimalDesign_direct} weakens both positive
and negative correlations of the assignments to protect against adversarially
correlated (again, positively or negatively) error distributions.

The results for the indirect and direct effects are in Table \ref{tab:RMSE-MK},
which reports the effective sample size gain compared to the baseline
RCT design with an unwhitened recentered IV estimator. For estimating
the spillover effect, Panel (a) shows that all optimal designs beat
the simple RCT under all error structures. Selecting $p=\infty$ gives
the best performance in all columns except column (4). Error correlations
are not adversarial enough for $p=2$ to dominate in DGPs (3) and
(5), as it did in the previous application,\footnote{A possible reason is that here $x_{i}$ is the number (rather than
share) of treated neighbors, so there is more scope for increasing
signal by correlating the assignments—which is what the $p=\infty$
design targets.} but it remains a good choice. Overall, the power gains are even more
substantial than in the first application: with the $p=2$ choice,
they range from +106\% in column 1 to +237\% in column 5.\footnote{To again isolate the performance impact of the optimal design vs.
the optimal estimator, Appendix Table \ref{tab:MK-extra} reports
the effective sample size gain of the recentered IV estimator without
whitening for the $p$-optimal designs along with the overidentified
GMM specification. In this application, GMM has an efficiency cost
compared to the better of the whitened and unwhitened instruments
in some columns, for the standard reason that the weighting matrix
has to be estimated from a relatively small sample (see footnote \ref{fn:GMM-finite-sample}).} Yet, Appendix Table \ref{tab:Coverage-MK} provides a reason to prefer
a conservative choice of $p$: under $p=\infty$, the resulting covariance
matrix of the shocks is too dense to meet the conditions of Theorem
\ref{thm:main}. Confidence intervals built using the normal approximation
under the $p=\infty$ design severely undercover the target parameter,
while the $p=1$ design has close to nominal coverage under all DGPs.

Since our proposed designs are optimized for the indirect effect,
they can lead to some deterioration of performance of the direct effect
estimator, as reported in Panel (b) of Table \ref{tab:RMSE-MK}. For
$p=\infty$ this deterioration can be substantial, equivalent to a
56\% reduction in sample size in column 6. However, for $p=1$ and
$p=2$ the deterioration is mild, of at most 21\%.\footnote{Appendix \ref{subsec:appx-more-MK} reports supplementary results
analogous to those reported in Appendix \ref{subsec:appx-more-Cai}.
We visualize the signal-isotropy tradeoff and the variances and mutual
covariances of $\tilde{x}_{i}$ and its whitened version, and measure
worst-case variance gains for a range of error distributions of increasing
complexity.}

\subsection{\label{subsec:egger}Saturation Designs: Egger et al. (2022)}

Finally, we explore the efficiency gains of our approach over saturation
designs with an additional application in the setting of \citet{egger2022general},
developed in Appendix \ref{subsec:appx-Egger}. \citet{egger2022general}
estimate direct and spillover effects of unconditional cash transfers
using a saturation design which varied treatment assignment probabilities
randomly between 1/3 and 2/3 across groups of villages, with complete
randomization within village groups. We show our optimal designs continue
to deliver large improvements in spillover effect RMSE relative to
this benchmark for $p=2$ and $p=\infty$. However, these results
alone are not sufficiently reassuring because direct effect estimation
is more severely degraded in this setting relative to the \citet{Miguel2004}
application. To address this, we apply the constrained version of
our approach from Appendix \ref{appx:direct_fx_constraint} requiring
that the worst-case approximate variance of the direct-effect estimator
be no more than 25\% higher than that of unwhitened recentered IV
under independent assignment. We find that our design and estimation
strategy continue to deliver large gains in the effective sample size
for spillover estimation, on the order of 100-250\%.

\section{Conclusion\label{sec:Conclusion}}

We have shown how researchers can design experiments and choose recentered
estimators to more precisely estimate spillover effects, under weak
restrictions on the distribution of unobservables. The core insight
of this approach is that spillover signal can be increased, by correlating
the assignments of connected units, without overly concentrating variation
in directions where the unobservables may be adversarially clustered.
Balancing this tradeoff between signal and isotropy yields a tractable
minimax problem with intuitive special cases and practical approaches
to computation. In three semi-synthetic experiments, based on high-profile
applications, we find large decreases in spillover effect RMSE relative
to independent randomization and saturation design benchmarks. These
gains are attained without large deterioration in the standard errors
for the direct effects, when present.

Several open questions remain. First, while we have derived extensions
of the main theory to nonlinear spillover treatment constructions,
computation of the optimal designs in such nonlinear settings remains
a challenge. Second, while we have shown that the optimal design problem
is invariant to certain adaptive instrument constructions, there is
more to explore in how a researcher might leverage some prior information
on the distribution of unobservables such as from a first-wave pilot
(\citealp{viviano_experimental_2025}) or how to pool information
from prior experiments. Relatedly, we have not considered adaptive
experimentation designs, in which the intervention units arrive sequentially.
Finally, we have not studied how the minimax approach changes when
a researcher is interested in estimating a more complicated parametric
model of spillovers, such as from a structural model (\citealp{borusyak_estimating_2025}).
We leave such extensions to future work.

\begin{singlespace}
\addcontentsline{toc}{section}{References}

\bibliographystyle{aer}
\bibliography{OptimalRCT}

\end{singlespace}

\appendix
\newpage\setcounter{equation}{0} \renewcommand{\theequation}{A\arabic{equation}}
\setcounter{prop}{0} \renewcommand{\theprop}{A\arabic{prop}}
\setcounter{lem}{0} \renewcommand{\thelem}{A\arabic{lem}}
\setcounter{cor}{0} \renewcommand{\thecor}{A\arabic{cor}}
\setcounter{algorithm}{0} \renewcommand{\thealgorithm}{A\arabic{algorithm}}\setcounter{table}{0} \renewcommand{\thetable}{A\arabic{table}}
\setcounter{figure}{0} \renewcommand{\thefigure}{A\arabic{figure}}
\begin{center}
{\LARGE\textbf{Appendix}}{\LARGE\par}
\par\end{center}

\section{Proofs of Main Results\label{appx:Proofs-of-Main}}

\subsection{\label{appx:Proof-Thm1}Theorem \ref{prop:OptimalIV}}

We drop the $\delta$ subscript from $\tilde{x}$, $S$, and $z^{\ast}$
to simplify notation. Write $\Omega_{\varepsilon}=\expec{\mathcal{E}}{\varepsilon\varepsilon'}$
and, for any recentered instrument $z$, $Q_{z}=\expec{\delta}{zz'}$.
The approximate variance of the recentered IV estimator can be written
\[
\V{\delta,\mathcal{E}}z{}=\frac{\expec{\delta,\mathcal{E}}{z'\varepsilon\varepsilon'z}}{\expec{\delta}{z'x}^{2}}=\frac{\tr{\left(Q_{z}\Omega_{\varepsilon}\right)}}{\expec{\delta}{z'\tilde{x}}^{2}}
\]
where the denominator uses $\expec{\delta}{z'\expec{\delta}x}=0$.
By the duality of Schatten $p$- and $q$-norms where $q\in[1,\infty]$,
such that $1/p+1/q=1$, the minimax approximate variance satisfies:
\begin{align}
\max_{\mathcal{E}\in\mathcal{F}_{p}(\sigma)}\V{\delta,\mathcal{E}}z{} & =\sup_{\Omega_{\varepsilon}\succeq0,\left\Vert \Omega_{\varepsilon}\right\Vert _{p}\le N^{1/p}\sigma^{2}}\frac{\text{tr}(Q_{z}\Omega_{\varepsilon})}{\expec{\delta}{z'\tilde{x}}^{2}}=N^{1/p}\sigma^{2}\frac{\left\Vert Q_{z}\right\Vert _{q}}{\expec{\delta}{z'\tilde{x}}^{2}}.\label{eq:minimax_var}
\end{align}
Note that the additional constraint that $\Omega_{\varepsilon}$ is
positive semi-definite does not affect the supremum because $Q_{z}\succeq0$.

Next, we propose a worst-case $\Omega_{\varepsilon}$ for a given
$z$ and show that it attains the supremum. Fix an instrument $z$
for which the second moment is non-zero, so that $\left\Vert Q_{z}\right\Vert _{q}\neq0.$
Then, conjecture the following worst-case matrix: 
\begin{equation}
\Omega^{\ast}_{\varepsilon}=\begin{cases}
N\sigma^{2}\frac{B_{z}}{\tr(B_{z})}, & p=1,\\
N^{1/p}\sigma^{2}\frac{Q^{q-1}_{z}}{\|Q_{z}\|^{q-1}_{q}}, & 1<p<\infty,\\
\sigma^{2}Q^{0}_{z}, & p=\infty,
\end{cases}\label{eq:int_cov}
\end{equation}
where $B_{z}\neq0$ is any positive semidefinite matrix with $\Im(B_{z})\subseteq\mathcal{M}_{z}$
where $\mathcal{M}_{z}$ is the eigenspace associated with the largest
eigenvalue of $Q_{z}$. For $1<p<\infty$, $p(q-1)=q$ and $q/p=q-1$,
so:
\begin{align*}
\left\Vert Q^{q-1}_{z}\right\Vert _{p} & =\left\{ \tr(Q^{p(q-1)}_{z})\right\} ^{1/p}\\
 & =\left\{ \tr(Q^{q}_{z})\right\} ^{1/p}\\
 & =\|Q_{z}\|^{q-1}_{q}.
\end{align*}
This means $\left\Vert \Omega^{\ast}_{\varepsilon}\right\Vert _{p}=N^{1/p}\sigma^{2}$
so $\Omega^{\ast}_{\varepsilon}$ is feasible under $\mathcal{F}_{p}(\sigma)$.
In addition, 
\begin{align*}
\tr{\left(Q_{z}\Omega^{\ast}_{\varepsilon}\right)} & =N^{1/p}\sigma^{2}\frac{\tr(Q^{q}_{z})}{\|Q_{z}\|^{q-1}_{q}}\\
 & =N^{1/p}\sigma^{2}\frac{\|Q_{z}\|^{q}_{q}}{\|Q_{z}\|^{q-1}_{q}}\\
 & =N^{1/p}\sigma^{2}\left\Vert Q_{z}\right\Vert _{q},
\end{align*}
which attains the supremum from the optimization problem in (\ref{eq:minimax_var}).
Similarly, for $p=1$ and $q=\infty$, we have $\left\Vert \Omega^{\ast}_{\varepsilon}\right\Vert _{1}=N\sigma^{2}$
and the matrix is feasible. Furthermore, we have $Q_{z}B_{z}=\left\Vert Q_{z}\right\Vert _{\infty}B_{z}$,
so $\tr(Q_{z}\Omega^{*}_{\varepsilon})=N\sigma^{2}\left\Vert Q_{z}\right\Vert _{\infty}$,
which attains the supremum. For $p=\infty$ and $q=1$, $\left\Vert \Omega^{\ast}_{\varepsilon}\right\Vert _{\infty}=\sigma^{2}$
and so it is again feasible, and $\tr(Q_{z}\Omega^{*}_{\varepsilon})=\sigma^{2}\tr(Q_{z})=\sigma^{2}\left\Vert Q_{z}\right\Vert _{1}$
which again attains the supremum.

To find the optimal instrument, we proceed in two steps. We first
show that there is an instrument that minimizes (\ref{eq:minimax_var})
that takes the form $z=V\tilde{x}$. We then find the optimal $V$
matrix. Define the linear projection of $z$ on $\tilde{x}$ as $v_{z}=V_{z}\tilde{x}$
for $V_{z}=\expec{\delta}{z\tilde{x}'}S^{\dag}$, with the residual
$u_{z}=z-v_{z}$ that satisfies $\expec{\delta}{\tilde{x}u_{z}'}=\expec{\delta}{v_{z}u_{z}'}=0$.
Thus, 
\[
\expec{\delta}{v_{z}\tilde{x}'}=\expec{\delta}{\left(z-u_{z}\right)\tilde{x}'}=\expec{\delta}{z\tilde{x'}}
\]
and
\[
Q_{z}=\expec{\delta}{\left(v_{z}+u_{z}\right)\left(v_{z}+u_{z}\right)'}=\expec{\delta}{v_{z}v_{z}'}+\expec{\delta}{u_{z}u_{z}'}\succeq\expec{\delta}{v_{z}v_{z}'}=Q_{v}.
\]
By the monotonicity of the Schatten $q$-norm on the positive semidefinite
cone, $\left\Vert Q_{z}\right\Vert _{q}\ge\left\Vert Q_{v}\right\Vert _{q}$.

Thus, the objective in (\ref{eq:minimax_var}) is weakly lower for
the instrument $v$ than $z$ and it is without loss to restrict attention
to the instruments that take the form $z=V\tilde{x}$. For them, $Q_{z}=VSV'$
and $\expec{\delta}{z^{\prime}\tilde{x}}=\tr{\left(VS\right)}$ such
that the optimal $V$ matrix solves
\begin{equation}
\min_{V}\frac{\left\Vert VSV'\right\Vert _{q}}{\tr(VS)^{2}}.\label{eq:VSV}
\end{equation}
Let $H=VS^{1/2}$ and $r=2p/(p+1)\in[1,2]$ be the Hölder conjugate
of $2q$, so that $\frac{1}{r}+\frac{1}{2q}=1$. By the definition
of the Schatten norm, $\left\Vert VSV'\right\Vert _{q}=\left\Vert HH'\right\Vert _{q}=\left\Vert H\right\Vert ^{2}_{2q}$
while $\tr{\left(VS\right)}=\tr{\left(HS^{1/2}\right)}$. By the Hölder
inequality for Schatten norms (Corollary IV.2.6 of \citet{bhatia2013matrix}),
\begin{equation}
\tr{\left(VS\right)}^{2}\le\left\Vert H\right\Vert ^{2}_{2q}\left\Vert S^{1/2}\right\Vert ^{2}_{r}=\left\Vert VSV'\right\Vert _{q}\cdot\left\{ \tr{\left(S^{r/2}\right)}\right\} ^{2/r}.\label{eq:inst_inequality}
\end{equation}
Plugging this inequality into (\ref{eq:VSV}), we have that for any
$V$, $\frac{\left\Vert VSV'\right\Vert _{q}}{\tr(VS)^{2}}\geq\left\{ \tr{\left(S^{r/2}\right)}\right\} ^{-2/r}$.
This lower bound is attained by $V^{*}=\left(S^{\dag}\right)^{1-r/2}=\left(S^{\dag}\right)^{1/(p+1)}$
since $\left\Vert V^{*}SV^{*}{}'\right\Vert _{q}=\tr{\left(S^{p/(p+1)}\right)}^{(p-1)/p}$
for $p>1$ and $\left\Vert V^{*}SV^{*}{}'\right\Vert _{q}=1$ for
$p=1$, while $\tr(V^{*}S)=\tr{\left(S^{p/(p+1)}\right)}$. Therefore,
$V^{\ast}$ solves (\ref{eq:VSV}). The optimal instrument is then
$z^{\ast}=\left(S^{\dag}\right)^{1/(p+1)}\tilde{x}$ and by plugging
in the optimum of (\ref{eq:VSV}) into (\ref{eq:minimax_var}), we
have $\max_{\mathcal{E}\in\mathcal{F}_{p}(\sigma)}\V{\delta,\mathcal{E}}{z^{*}}=N^{1/p}\sigma^{2}\left\{ \tr{\left(S^{r/2}\right)}\right\} ^{-2/r}$.

To get an optimal matrix $\Omega^{\ast}_{\varepsilon}$ corresponding
to $z^{\ast}$, we use (\ref{eq:int_cov}). For $1<p<\infty$, we
have 
\[
\Omega^{*}_{\varepsilon}=N^{1/p}\sigma^{2}\,\frac{Q^{q-1}_{z^{*}}}{\|Q_{z^{*}}\|^{q-1}_{q}}.
\]
Then, plugging in the expression for $Q_{z^{*}}=S^{(p-1)/(p+1)}$
, we have that the worst-case error second-moment matrix corresponding
to $z^{\ast}$ is
\begin{align}
\Omega^{*}_{\varepsilon} & =N^{1/p}\sigma^{2}\frac{S^{1/(p+1)}}{\left(\tr{\left(S^{p/(p+1)}\right)}\right)^{1/p}}.\label{eq:omega_star}
\end{align}
For $p=1$, $Q_{z^{*}}=S^{0}$. The eigenspace associated with the
largest eigenvalue of $Q_{z}$ is $\Im(S)$, so to match the formula
for general $p$ while still meeting the conditions of (\ref{eq:int_cov})
we can choose $B_{z}=S^{1/2}$ and $\Omega^{*}_{\varepsilon}=N\sigma^{2}\frac{{S^{1/2}}}{\tr(S^{1/2})}$.
The worst-case matrix is not unique here, because $B_{z}$ is not
unique ($B_{z}=S$ also works, for example), but it is convenient
to match the general formula. Finally, for $p=\infty$, $\Omega^{*}_{\varepsilon}=\sigma^{2}S^{0}$.
Note the optimal instrument yields the recentered IV estimator with
a GLS weight matrix that is inversely proportional to this worst-case
$\Omega^{*}_{\varepsilon}$: $\hat{\beta}=(\tilde{x}^{\prime}\Omega^{*\dag}_{\varepsilon}x)^{-1}\tilde{x}^{\prime}\Omega^{*\dag}_{\varepsilon}y$.

\subsection{Proposition \ref{prop:OptimalDesign}\label{appx:Proof-Prop1}}

Follows immediately from the expression for the worst-case approximate
variance in Theorem \ref{prop:OptimalIV}, which is finite whenever
$S_{\delta}\ne0$.

\subsection{Proposition \ref{prop:DesignProperties}}

For the first claim, take any $\delta$ and define its complement
$\delta^{c}$ by drawing $\one-g$ whenever $\delta$ would draw $g$.
Now symmetrize:
\begin{align*}
\delta^{sym} & =\frac{1}{2}\delta+\frac{1}{2}\delta^{c}.
\end{align*}
Clearly $\expec{\delta^{sym}}{g_{k}}=0.5$ for all $k$ while
\begin{align*}
\var{\delta^{sym}}g & =\var{\delta}g+\left(\expec{\delta}g-\frac{1}{2}\mathbf{1}\right)\left(\expec{\delta}g-\frac{1}{2}\mathbf{1}\right)'\succeq\var{\delta}g.
\end{align*}
Hence
\begin{align*}
S_{\delta^{sym}} & =W\var{\delta^{sym}}gW^{\prime}\succeq W\var{\delta}gW^{\prime}=S_{\delta}.
\end{align*}
Furthermore, since $A\mapsto\tr(A^{p/(p+1)})$ is increasing on the
positive semidefinite cone,
\begin{align*}
\tr{\left(S^{p/(p+1)}_{\delta^{sym}}\right)} & \ge\tr{\left(S^{p/(p+1)}_{\delta}\right)}.
\end{align*}
Hence symmetrizing weakly improves the objective, so some $\delta^{*}$
has $\expec{\delta^{\ast}}{g_{k}}=0.5$ for all $k$.

For the second claim, suppose $W=\text{diag}(W_{(1)},\dots,W_{(B)})$
for blocks $W_{(b)}\in\mathbb{R}^{N_{b}\times K_{b}}$. For any design,
let $\delta_{b}$ be the marginal distribution of $g_{(b)}$. Define
the product design $\delta^{ind}=\delta_{1}\otimes\dots\otimes\delta_{B}$
and note
\begin{align*}
S_{\delta^{ind}} & =\diag(S_{\delta,11},\dots,S_{\delta,BB})
\end{align*}
for $S_{\delta,bb}=W_{(b)}\var{\delta}{g_{(b)}}W^{\prime}_{(b)}$
since cross-block covariances vanish under the product design. $S_{\delta^{ind}}$
is a pinching of $S_{\delta}$, following the definition in Chapter
2 of \citet{bhatia2013matrix}. By Equation II.39 of \citet{bhatia2013matrix},
we have a majorization relationship $\lambda(S_{\delta^{ind}})\prec\lambda(S_{\delta})$,
where $\lambda(A)$ gives the vector of eigenvalues of matrix $A$.
Since $x\mapsto x^{p/(p+1)}$ is concave for $p\in[1,\infty]$, by
the characterization of majorization in Theorem II.3.1 of \citet{bhatia2013matrix},
we have: 
\begin{align*}
\tr(S^{p/(p+1)}_{\delta^{ind}}) & \ge\tr(S^{p/(p+1)}_{\delta}).
\end{align*}
Hence any design is weakly improved by replacing it with the product
of its block marginals and we may restrict attention to product designs.

For the third claim, we have 
\[
\tr(S^{p/(p+1)}_{\delta^{ind}})=\sum^{B}_{b=1}\tr(S^{p/(p+1)}_{\delta^{ind},bb}).
\]
Thus, the design can be optimized block by block.

\subsection{\texorpdfstring{Lemma \ref{lemma:directFX}, Theorem \ref{prop:OptimalIV_direct},
Propositions \ref{prop:OptimalDesign_direct}–\ref{prop:DesignProperties_direct}}{Lemma 1, Theorem 1', Propositions 1'-2'}}

\paragraph{Lemma \ref{lemma:directFX}.}

Take any $\boldsymbol{z}\in\mathcal{Z}^{(2)}_{\delta}$. If $M_{\bm{z}}=\expec{\delta}{\boldsymbol{x}'\boldsymbol{z}}$
is nonsingular, we can define $z=\bm{z}M^{-1}_{\bm{z}}e_{1}$, where
$e_{1}=(1,0)'$. Since $\mathbb{E}_{\delta}[\bm{z}]=0$, $\mathbb{E}_{\delta}[z]=0$,
so it is re-centered. Also, $\expec{\delta}{\bm{x}'z}=M_{z}M^{-1}_{z}e_{1}=e_{1}$,
i.e. $\expec{\delta}{x^{\prime}z}=1$ and $\expec{\delta}{h^{\prime}z}=0$.
Next, to show the variance equality, we have that $\varepsilon'z=\varepsilon'\boldsymbol{z}M^{-1}_{z}e_{1}$.
Because $\expec{\delta}{x'z}=1$, 
\begin{align*}
\V{\delta,\mathcal{E}}z & =\var{\delta,\mathcal{E}}{\varepsilon'\boldsymbol{z}M^{-1}_{z}e_{1}}\\
 & =e_{1}'\left(M^{-1}_{\boldsymbol{z}}\right)'\var{\delta,\mathcal{E}}{\boldsymbol{z}'\varepsilon}M^{-1}_{\boldsymbol{z}}e_{1}\\
 & =\boldsymbol{\mathcal{V}}_{\delta,\mathcal{E}}\left[\boldsymbol{z}\right].
\end{align*}
Hence every vector instrument has a scalar instrument orthogonal to
$h$ with the same approximate variance for the spillover effect.
If $M_{z}$ is singular, $\boldsymbol{\mathcal{V}}_{\delta,\mathcal{E}}\left[\boldsymbol{z}\right]=\infty$
so we can take $z=0$ which achieves the same approximate variance.

Conversely, take any $z\in\mathcal{Z}_{\delta}$ with $\expec{\delta}{z^{\prime}h}=0$.
Let $\boldsymbol{z}=(z,h-\expec{\delta}h)$ and note that $\expec{\delta}{\boldsymbol{z}}=0$,
so $\boldsymbol{z}\in\mathcal{Z}^{(2)}_{\delta}$. If
\begin{align*}
\expec{\delta}{\boldsymbol{z}^{\prime}\boldsymbol{x}} & =\begin{bmatrix}\expec{\delta}{z^{\prime}x} & 0\\
\expec{\delta}{(h-\expec{\delta}h)^{\prime}x} & \tr(\var{\delta}h)
\end{bmatrix}
\end{align*}
is nonsingular, simple algebra verifies that
\begin{align*}
\V{\delta,\mathcal{E}}z & =\boldsymbol{\mathcal{V}}_{\delta,\mathcal{E}}\left[\boldsymbol{z}\right].
\end{align*}
If $\expec{\delta}{\boldsymbol{z}^{\prime}\boldsymbol{x}}$ is singular,
from $\tr(\var{\delta}h)>0$ we conclude that $\expec{\delta}{z'x}=0$.
Thus, we have $\V{\delta,\mathcal{E}}z=\boldsymbol{\mathcal{V}}_{\delta,\mathcal{E}}\left[\boldsymbol{z}\right]=\infty$.

\paragraph{Theorem \ref{prop:OptimalIV_direct}.}

By Lemma \ref{lemma:directFX}, it is enough to solve 
\begin{equation}
\inf_{\substack{z\in\mathcal{Z}_{\delta}:\,\expec{\delta}{z'h}=0}
}\max_{\mathcal{E}\in\mathcal{F}_{p}(\sigma)}\V{\delta,\mathcal{E}}z.\label{eq:const_minimax}
\end{equation}
By (\ref{eq:minimax_var}), this problem can be simplified to 
\begin{equation}
\inf\limits_{\substack{z\in\mathcal{Z}_{\delta}:\,\expec{\delta}{z'h}=0}
}N^{1/p}\sigma^{2}\frac{\left\Vert Q_{z}\right\Vert _{q}}{\expec{\delta}{z'\tilde{x}}^{2}}.\label{eq:prop1prime_step}
\end{equation}
Since $\expec{\delta}z=0$, for $\tilde{h}=h-\expec{\delta}h$, $\expec{\delta}{z'h}=\expec{\delta}{z'\tilde{h}}$,
and similarly, $\expec{\delta}{z'x}=\expec{\delta}{z'\tilde{x}}$.
This problem is invariant to rescaling of $z$: for any scalar $\alpha\ne0$,
$Q_{\alpha z}=\alpha^{2}Q_{z}$ and $\expec{\delta}{\left(\alpha z\right)'\tilde{x}}^{2}=\alpha^{2}\expec{\delta}{z'\tilde{x}}^{2}$.
Thus, it is without loss to impose the constraint $\left\Vert Q_{z}\right\Vert _{q}\le1$
and replace $\left\Vert Q_{z}\right\Vert _{q}$ with its maximum value
of one in (\ref{eq:prop1prime_step}): i.e., 
\[
\inf\limits_{\substack{z\in\mathcal{Z}_{\delta}:\,\expec{\delta}{z'h}=0}
}N^{1/p}\sigma^{2}\frac{\left\Vert Q_{z}\right\Vert _{q}}{\expec{\delta}{z'\tilde{x}}^{2}}=N^{1/p}\sigma^{2}\Bigl(\max_{z\in\mathcal{Z}_{\delta}:\expec{\delta}{z'\tilde{h}}=0,\ \left\Vert Q_{z}\right\Vert _{q}\le1}\expec{\delta}{z'\tilde{x}}\Bigr)^{-2}.
\]
We next drop constants and the exponent, which do not affect the solution,
and convert the constraint $\expec{\delta}{z'\tilde{h}}=0$ into the
Lagrangian form: 
\[
\max_{z\in\mathcal{Z}_{\delta}:\expec{\delta}{z'\tilde{h}}=0,\ \left\Vert Q_{z}\right\Vert _{q}\le1}\expec{\delta}{z'\tilde{x}}=\max_{\substack{z\in\mathcal{Z}_{\delta}:\left\Vert Q_{z}\right\Vert _{q}\le1}
}\inf_{c\in\mathbb{R}}\expec{\delta}{z'(\tilde{x}-c\tilde{h})}.
\]

We now apply the Sion minimax theorem to swap the order of optimization,
noting that the set $\mathcal{C}=\left\{ z\in\mathcal{Z}_{\delta}:\left\Vert Q_{z}\right\Vert _{q}\le1\right\} $
is convex (since $\mathcal{Z}_{\delta}$ is convex and by the properties
of the Schatten norm) and compact (because the support of $g$ is
finite), that $\mathbb{R}$ is convex, that $z'(\tilde{x}-c\tilde{h})$
is linear in $z$ and affine in $c$, and that the expectation is
a linear operator: 
\begin{equation}
\max_{\substack{z\in\mathcal{Z}_{\delta}:\left\Vert Q_{z}\right\Vert _{q}\le1}
}\inf_{c\in\mathbb{R}}\expec{\delta}{z'(\tilde{x}-c\tilde{h})}=\inf_{c\in\mathbb{R}}\max_{\substack{z\in\mathcal{Z}_{\delta}:\left\Vert Q_{z}\right\Vert _{q}\le1}
}\expec{\delta}{z'(\tilde{x}-c\tilde{h})}.\label{eq:duality}
\end{equation}
Although their optimal values are equal, we have not yet shown that
solutions for one side are optimal for the other side. We proceed
by finding a solution to the right side problem, and showing that
it is in fact optimal for the left-side problem. Applying the scale-invariance
argument again, this problem is equivalent to
\begin{equation}
\inf_{c\in\mathbb{R}}\Bigl(\min_{\substack{z\in\mathcal{Z}_{\delta}}
}N^{1/p}\sigma^{2}\frac{\left\Vert Q_{z}\right\Vert _{q}}{\expec{\delta}{z'(\tilde{x}-c\tilde{h})}^{2}}\Bigr)^{-1/2}.\label{eq:prop1prime_step2}
\end{equation}
For fixed $c$, since $\tilde{x}-c\tilde{h}=(W-cU)\tilde{g}$, the
inner optimization problem is therefore the Theorem \ref{prop:OptimalIV}
problem applied to the residualized exposure matrix $W-cU$. So we
can write: 
\begin{equation}
\min_{\substack{z\in\mathcal{Z}_{\delta}}
}N^{1/p}\sigma^{2}\frac{\left\Vert Q_{z}\right\Vert _{q}}{\expec{\delta}{z'(\tilde{x}-c\tilde{h})}^{2}}=N^{1/p}\sigma^{2}\left\{ \tr{\left(\left[(W-cU)\Sigma_{\delta}(W-cU)'\right]^{p/(p+1)}\right)}\right\} ^{-(p+1)/p}.\label{eq:multiple_exp_swapped}
\end{equation}

We next show that the outer optimization in (\ref{eq:prop1prime_step2}),
which can be equivalently written as $\inf\limits_{c\in\mathbb{R}}H(c)$
for $H(c)=\tr{\!}\left(\left[(W-cU)\Sigma_{\delta}(W-cU)'\right]^{p/(p+1)}\right)$,
attains its minimum, so that $c_{\delta}$ and $S^{\perp}_{\delta}$
are well-defined. This follows because we have assumed that $\tr(U\Sigma_{\delta}U')>0$,
and thus $H(c)\to\infty$ when $|c|\to\infty$, while $H(c)$ is continuous
in $c$ as a composition of continuous functions. For a fixed $c=c_{\delta}$,
we know from Theorem \ref{prop:OptimalIV} with $W$ replaced by $W-c_{\delta}U$,
that the minimal approximate variance is $N^{1/p}\sigma^{2}\left\{ \tr{\left(\left[(W-c_{\delta}U)\Sigma_{\delta}(W-c_{\delta}U)'\right]^{p/(p+1)}\right)}\right\} ^{-(p+1)/p}=N^{1/p}\sigma^{2}\left\{ \tr{\left[\left(S^{\perp}_{\delta}\right)^{p/(p+1)}\right]}\right\} ^{-(p+1)/p}$
, which completes the proof of the first part of the Theorem. For
the optimal IV, an optimal recentered instrument from Theorem \ref{prop:OptimalIV}
is 
\[
z^{*}_{\delta}(g)=\left((S^{\perp}_{\delta})^{\dagger}\right)^{1/(p+1)}\tilde{x}^{\perp}_{\delta}.
\]

This $z^{*}_{\delta}(g)$ is a solution to the right-side problem
in (\ref{eq:duality}). For it to also be a solution to the left-side
problem in (\ref{eq:duality}) and therefore to the original problem,
it is enough to show that it is feasible in the original problem.
It is clearly recentered, so it remains to verify that the proposed
solution $z^{*}_{\delta}(g)$ is orthogonal to $h$. Let $z_{c}=\left((S_{\delta,c})^{\dagger}\right)^{1/(p+1)}(\tilde{x}-c\tilde{h})$,
where $S_{\delta,c}=(W-cU)\Sigma_{\delta}(W-cU)'$. Let $d=\rank(\Sigma_{\delta})$
and choose a factor $R_{\delta}\in\mathbb{R}^{K\times d}$ satisfying
$R_{\delta}R_{\delta}'=\Sigma_{\delta}$. Define $A_{c}=(W-cU)R_{\delta},$
so $S_{\delta,c}=A_{c}A_{c}'$ and
\[
H(c)=\tr{\left(S^{p/(p+1)}_{\delta,c}\right)}=\|A_{c}\|^{r}_{r}
\]
for $r=2p/(p+1)$ as in Appendix \ref{appx:Proof-Thm1}. Let $s\in\mathbb{R}^{m}$
, where $m=\min\{N,d\}$, and define $f_{r}(s)=\sum\limits^{m}_{j=1}\left|s_{j}\right|^{r}$.
$H(c)$ is the function $f_{r}(\cdot)$ applied to the singular values
of $A_{c}$. Since $f_{r}(\cdot)$ is invariant to sign changes and
permutations, by Theorem 3.1 of \citet{lewis1995convex}, if $f_{r}(s)$
is differentiable at the vector of singular values, 
\[
\nabla_{A}\|A_{c}\|^{\,r}_{r}=r\left((A_{c}A_{c}')^{\dagger}\right)^{1-r/2}A_{c}.
\]
When $p>1$ and therefore $r>1$, $f_{r}(s)$ is always differentiable.
But when $p=1$, $f_{1}(s)$ is differentiable only if all singular
values of $A_{c_{\delta}}$ are non-zero. A sufficient condition is
the rank condition that we have assumed: $\rank\!\left((W-c_{\delta}U)R_{\delta}\right)=\rank\!\left((W-c_{\delta}U)\Sigma^{1/2}_{\delta}\right)=\min\{N,\rank(\Sigma_{\delta})\}$.
By the chain rule, given the derivative of the matrix function above
and that $dA_{c}/dc=-UR_{\delta}$, we have shown that $H(c)$ is
always differentiable at $c_{\delta}$.

Since $z_{c}$, up to a normalization, is a solution to the inner
maximization of the right side of (\ref{eq:duality}), the envelope
theorem implies $H'(c)\propto\expec{\delta}{z_{c}'\tilde{h}}$. Since
$H'(c_{\delta})=0$, we have that $\expec{}{z^{*}_{\delta}(g)'\tilde{h}}=0$,
as required for feasibility of $z^{*}_{\delta}(g)$ in the original
problem.

Finally, Lemma \ref{lemma:directFX} maps this constrained scalar
optimum back to the two-instrument problem: it is implemented by $\boldsymbol{z}^{*}_{i}(g)=(z^{*}_{\delta}(g)_{i},h_{i}-\expec{\delta}{h_{i}})$,
with the same optimal value, where $\tr(U\Sigma_{\delta}U')>0$ guarantees
that the first-stage matrix is non-singular.

\paragraph{Proposition \ref{prop:OptimalDesign_direct}.}

Follows immediately from Theorem \ref{prop:OptimalIV_direct}.

\paragraph{Proposition \ref{prop:DesignProperties_direct}.}

Follows analogously to Proposition \ref{prop:DesignProperties}.

\subsection{\label{appx:proof-relaxed}Characterization of Relaxed Solution from
Section \ref{subsec:Computation}}

We first note that, since $W$ has no zero columns, $W'W$ has all
positive diagonal elements. Thus, constraining $\ell^{\ast}>0$ is
without loss since $\tr{\left(M^{p}_{\ell}\right)}\to\infty$ whenever
any $\ell_{k}$ approaches zero. Next, we establish:
\begin{lem}
\label{lem:lstar_fixed_point}$\ell^{\ast}$ solves (\ref{eq:lstar})
if and only if it satisfies:
\[
\ell^{\ast}_{k}=\frac{\left(M^{p}_{\ell^{\ast}}\right)_{kk}}{\tr{\left(M^{p}_{\ell^{\ast}}\right)}},\qquad k=1,\dots,K.
\]
\end{lem}
\begin{proof}
We show below that the objective function in (\ref{eq:lstar}) is
convex. Furthermore, Slater's condition holds since the point $\bar{\ell}=1/K\cdot\bm{1}_{K}$
is strictly feasible. Thus, the KKT condition is necessary and sufficient.
Some $\ell\in\Delta_{K}$ minimizes $f(\ell)=\tr{\left(M^{p}_{\ell}\right)}$
subject to $\sum_{k}\ell_{k}=1$ if and only if 
\begin{equation}
\frac{\partial\tr{\left(M^{p}_{\ell}\right)}}{\partial\ell_{k}}=\text{const},\qquad k=1,\dots,K.\label{eq:lemma_const_deriv}
\end{equation}
By the chain rule and the derivative of spectral functions (i.e.,
functions of symmetric matrices that depend only on the eigenvalues
of the matrix) given in Theorem 1.1 of \citet{lewis1996derivatives},
\[
\frac{\partial}{\partial\ell_{k}}\tr(M^{p}_{\ell})=p\,\tr{\left(M^{p-1}_{\ell}\frac{\partial M_{\ell}}{\partial\ell_{k}}\right)}.
\]
Since $\frac{\partial M_{\ell}}{\partial\ell_{k}}=-\frac{1}{2\ell_{k}}\left(e_{k}e_{k}'M_{\ell}+M_{\ell}e_{k}e_{k}'\right)$,
we have $\frac{\partial}{\partial\ell_{k}}\tr(M^{p}_{\ell})=-p\frac{(M^{p}_{\ell})_{kk}}{\ell_{k}}$.
By (\ref{eq:lemma_const_deriv}), $\ell^{\ast}_{k}\propto\left(M^{p}_{\ell^{\ast}}\right)_{kk}$.
Using $\sum_{k}\ell^{\ast}_{k}=1$ and $\sum_{k}\left(M^{p}_{\ell^{\ast}}\right)_{kk}=\tr{\left(M^{p}_{\ell^{\ast}}\right)}$
yields the desired characterization.

To verify convexity, let $t\in[0,1]$ and $G_{\ell}=WD^{-1}_{\ell}W'$.
Since the matrices $G_{\ell}$ and $M_{\ell}$ have the same nonzero
eigenvalues, $f(\ell)=\tr(G^{p}_{\ell})$. Writing $w_{\cdot k}$
for the $k$-th column of $W$, we have 
\[
G_{\ell}=\sum^{K}_{k=1}\frac{w_{\cdot k}w_{\cdot k}'}{\ell_{k}}.
\]
For any $\ell,\tilde{\ell}>0$, the convexity of $x\mapsto x^{-1}$
gives $G_{t\ell+(1-t)\tilde{\ell}}\preceq tG_{\ell}+(1-t)G_{\tilde{\ell}}$.
The function $\tr(A^{p})=\|A\|^{p}_{p}$ is convex and monotonic on
the positive semidefinite cone for $p\geq1$, since it is the composition
of a norm with a convex function. Applying these arguments, 
\begin{align*}
f(t\ell+(1-t)\tilde{\ell}) & =\|G{}_{t\ell+(1-t)\tilde{\ell}}\|^{p}_{p}\\
 & \leq\|tG{}_{\ell}+(1-t)G_{\tilde{\ell}}\|^{p}_{p}\\
 & \leq t\|G_{\ell}\|^{p}_{p}+(1-t)\|G_{\tilde{\ell}}\|^{p}_{p}\\
 & =tf(\ell)+(1-t)f(\tilde{\ell}).
\end{align*}
Thus, the objective is convex.
\end{proof}

\paragraph{Proof of Proposition \ref{prop:Optimize_l}.}

Define 
\[
\bar{W}_{\ell}=WD^{-1/2}_{\ell},\qquad\bar{\Sigma}_{\ell}=D^{1/2}_{\ell}\Sigma D^{1/2}_{\ell},
\]
such that 
\[
W\Sigma W'=\bar{W}_{\ell}\bar{\Sigma}_{\ell}\bar{W}_{\ell}'.
\]
If $\Sigma\in\mathcal{Q}_{K}$, then 
\[
\tr(\bar{\Sigma}_{\ell})=\sum^{K}_{k=1}\ell_{k}\Sigma_{kk}=\frac{1}{4}.
\]
Using the Hölder inequality for Schatten norms (Corollary IV.2.6 of
\citet{bhatia2013matrix}) with $\left(\frac{2p}{p+1}\right)^{-1}=\left(2p\right)^{-1}+2^{-1}$,
we have 
\[
\tr{\left((W\Sigma W')^{p/(p+1)}\right)}=\left\Vert \bar{W}_{\ell}\bar{\Sigma}^{1/2}_{\ell}\right\Vert ^{2p/(p+1)}_{2p/(p+1)}\le\left\Vert \bar{W}_{\ell}\right\Vert ^{2p/(p+1)}_{2p}\cdot\left\Vert \bar{\Sigma}^{1/2}_{\ell}\right\Vert ^{2p/(p+1)}_{2}.
\]
Since 
\[
\left\Vert \bar{\Sigma}^{1/2}_{\ell}\right\Vert ^{2}_{2}=\tr(\bar{\Sigma}_{\ell})=\frac{1}{4}
\]
and 
\[
\left\Vert \bar{W}_{\ell}\right\Vert ^{2p}_{2p}=\tr{\left((\bar{W}_{\ell}'\bar{W}_{\ell})^{p}\right)}=\tr(M^{p}_{\ell}),
\]
it follows that 
\[
\tr{\left((W\Sigma W')^{p/(p+1)}\right)}\le4^{-p/(p+1)}\tr(M^{p}_{\ell})^{1/(p+1)}.
\]
This bound holds for every $\ell\in\interior{\Delta_{K}}$, so 
\begin{equation}
\tr{\left((W\Sigma W')^{p/(p+1)}\right)}\le4^{-p/(p+1)}\tr(M^{p}_{\ell^{*}})^{1/(p+1)}.\label{eq:relaxed_upper_bound}
\end{equation}

It remains to show that $\Sigma^{\ast}$ belongs to $\mathcal{Q}_{K}$
and attains the upper bound in (\ref{eq:relaxed_upper_bound}). We
have $\Sigma^{\ast}\succeq0$ because $M^{p}_{\ell^{\ast}}\succeq0$.
Moreover, using Lemma \ref{lem:lstar_fixed_point},
\[
\Sigma^{\ast}_{kk}=\frac{1}{4}\frac{\left(M^{p}_{\ell^{\ast}}\right)_{kk}}{\ell^{\ast}_{k}\tr{\left(M^{p}_{\ell^{\ast}}\right)}}=\frac{1}{4}.
\]
Plugging in $\Sigma^{\ast}$ into the objective yields
\[
\tr{\left((W\Sigma^{\ast}W')^{p/(p+1)}\right)}=\tr{\left(\left(\bar{W}_{\ell^{\ast}}\bar{\Sigma}^{\ast}_{\ell^{\ast}}\bar{W}_{\ell^{\ast}}'\right)^{p/(p+1)}\right)}
\]
where $\bar{\Sigma}^{\ast}_{\ell^{\ast}}=\frac{1}{4}\frac{M^{p}_{\ell^{\ast}}}{\tr{\left(M^{p}_{\ell^{\ast}}\right)}}$.
The non-zero eigenvalues of $\bar{W}_{\ell^{\ast}}\bar{\Sigma}^{\ast}_{\ell^{\ast}}\bar{W}_{\ell^{\ast}}'$
equal the non-zero eigenvalues of $\left(\bar{\Sigma}^{\ast}_{\ell^{\ast}}\right)^{1/2}\bar{W}_{\ell^{\ast}}'\bar{W}_{\ell^{\ast}}\left(\bar{\Sigma}^{\ast}_{\ell^{\ast}}\right)^{1/2}$.
Noting that $\bar{W}_{\ell^{\ast}}'\bar{W}_{\ell^{\ast}}=M_{\ell^{\ast}}$,
the latter matrix equals $\frac{M^{p/2}_{\ell^{\ast}}}{2\sqrt{\tr{\left(M^{p}_{\ell^{\ast}}\right)}}}M_{\ell^{\ast}}\frac{M^{p/2}_{\ell^{\ast}}}{2\sqrt{\tr{\left(M^{p}_{\ell^{\ast}}\right)}}}=\frac{M^{p+1}_{\ell^{\ast}}}{4\tr{\left(M^{p}_{\ell^{\ast}}\right)}}$.
Thus, non-zero eigenvalues of $\left(\bar{W}_{\ell^{\ast}}\bar{\Sigma}^{\ast}_{\ell^{\ast}}\bar{W}_{\ell^{\ast}}'\right)^{p/(p+1)}$
equal those of $\left(\frac{M^{p+1}_{\ell^{\ast}}}{4\tr{\left(M^{p}_{\ell^{\ast}}\right)}}\right)^{p/(p+1)}=\frac{M^{p}_{\ell^{\ast}}}{\left(4\tr{\left(M^{p}_{\ell^{\ast}}\right)}\right)^{p/(p+1)}}$,
from which we conclude that
\begin{equation}
\tr{\left((W\Sigma^{\ast}W')^{p/(p+1)}\right)}=\tr{\left(\frac{M^{p}_{\ell^{\ast}}}{\left(4\tr{\left(M^{p}_{\ell^{\ast}}\right)}\right)^{p/(p+1)}}\right)}=4^{-p/(p+1)}\tr(M^{p}_{\ell^{*}})^{1/(p+1)},\label{eq:relaxed_value}
\end{equation}
achieving the upper bound in (\ref{eq:relaxed_upper_bound}).

It remains to prove the claim about the optimal IV. Consider the thin
singular value decomposition that drops the zero singular values:
\[
\overline{W}_{\ell^{*}}=U\Lambda V',
\]
where $\Lambda$ is diagonal with strictly positive elements and $U$
and $V$ have orthonormal columns. Thus, $M_{\ell^{*}}=V\Lambda^{2}V'$
and
\[
\overline{\Sigma}^{*}_{\ell^{*}}=D^{1/2}_{\ell^{*}}\Sigma^{*}D^{1/2}_{\ell^{*}}=\frac{V\Lambda^{2p}V'}{4\tr{\left(M^{p}_{\ell^{*}}\right)}}.
\]
Moreover,
\[
S^{*}=W\Sigma^{*}W'=\overline{W}_{\ell^{*}}\overline{\Sigma}^{*}_{\ell^{*}}\overline{W}_{\ell^{*}}'=\frac{U\Lambda^{2p+2}U'}{4\tr{\left(M^{p}_{\ell^{*}}\right)}}.
\]
Thus, writing $a_{\ell^{*}}=\left(4\tr{\left(M^{p}_{\ell^{*}}\right)}\right)^{1/(p+1)}\ne0$,
\[
\left(S^{*\dagger}\right)^{1/(p+1)}=a_{\ell^{*}}U\Lambda^{-2}U'.
\]

We can therefore write, for the Theorem \ref{prop:OptimalIV} optimal
IV, using $W=\overline{W}_{\ell^{*}}D^{1/2}_{\ell^{*}}=U\Lambda V'D^{1/2}_{\ell^{*}}$,
\begin{align*}
W'z^{\ast} & =W'\left(S^{*\dagger}\right)^{1/(p+1)}W\widetilde{g}\\
 & =a_{\ell^{*}}W'U\Lambda^{-2}U'W\widetilde{g}\\
 & =a_{\ell^{*}}D^{1/2}_{\ell^{*}}VV'D^{1/2}_{\ell^{*}}\widetilde{g}.
\end{align*}
Since $\var{}{D^{1/2}_{\ell^{*}}\widetilde{g}}=\overline{\Sigma}^{*}_{\ell^{*}}=\frac{V\Lambda^{2p}V'}{4\tr{\left(M^{p}_{\ell^{*}}\right)}}$,
we have that, for every $v\perp\col(V)$, $\var{}{v'D^{1/2}_{\ell^{*}}\widetilde{g}}=0$.
Given the expectation of $D^{1/2}_{\ell^{*}}\widetilde{g}$ is also
zero, $v'D^{1/2}_{\ell^{*}}\widetilde{g}=0$ for any $\tilde{g}$
with positive probability under the design, i.e. $D^{1/2}_{\ell^{*}}\widetilde{g}\in\col\left(V\right)$.
Thus, $VV'D^{1/2}_{\ell^{*}}\widetilde{g}=D^{1/2}_{\ell^{*}}\widetilde{g}$
and $W'z^{\ast}=a_{\ell^{*}}D_{\ell^{*}}\widetilde{g}$ or, equivalently,
\[
\tilde{g}_{k}=\frac{1}{a_{\ell^{*}}}\frac{\sum_{i}w_{ik}z^{*}_{i}}{\ell^{*}_{k}},
\]
as claimed.

For later use, the same SVD also gives
\begin{equation}
z^{\ast}=a_{\ell^{*}}U\Lambda^{-1}V'D^{1/2}_{\ell^{*}}\tilde{g}=a_{\ell^{*}}\left(D^{-1/2}_{\ell^{*}}W'\right)^{\dagger}D^{1/2}_{\ell^{*}}\widetilde{g}.\label{eq:optimalIV_prop3}
\end{equation}

\paragraph{Proof of Corollary \ref{cor:p1}}

When $p=1$, 
\[
\tr(M^{p}_{\ell})=\tr(M_{\ell})=\sum^{K}_{k=1}\frac{\left\Vert w_{\cdot k}\right\Vert ^{2}_{2}}{\ell_{k}}.
\]
Thus, the weights solving (\ref{eq:lstar}) satisfy 
\[
\ell^{*}_{k}\propto\left\Vert w_{\cdot k}\right\Vert _{2}.
\]
Plugging these weights into the formula for $\Sigma^{\ast}$ yields
the corollary result. The weights also yield the optimal IV as a special
case of Proposition \ref{prop:Optimize_l}.

\paragraph{Proof of Corollary \ref{cor:sym_diagonal}}

We show that, under the conditions of the corollary, $\ell^{\ast}=(1/K)\cdot\one{}{}_{K}$.
This result follows by Lemma \ref{lem:lstar_fixed_point}: with such
$\ell^{\ast}$, $D_{\ell^{\ast}}=\frac{1}{K}I_{K}$, $M^{p}_{\ell^{\ast}}=\left(KW'W\right)^{p}$,
and thus
\[
\frac{\left(M^{p}_{\ell^{\ast}}\right)_{kk}}{\tr{\left(M^{p}_{\ell^{\ast}}\right)}}=\frac{\left(\left(W'W\right)^{p}\right)_{kk}}{\tr{\left(\left(W'W\right)^{p}\right)}}=\frac{1}{K}=\ell^{\ast}_{k}.
\]
Thus, in Proposition \ref{prop:Optimize_l}
\[
\Sigma^{\ast}=\frac{K}{4}\frac{\left(W'W\right)^{p}}{\tr{\left(\left(W'W\right)^{p}\right)}},
\]
as required.

For the optimal IV, we plug in $D_{\ell^{\ast}}=\frac{1}{K}I_{K}$
into (\ref{eq:optimalIV_prop3}). We have $a_{\ell^{*}}=\left(4K^{p}\tr{\left((W'W)^{p}\right)}\right)^{1/(p+1)}$,
$\left(D^{-1/2}_{\ell^{*}}W'\right)^{\dagger}D^{1/2}_{\ell^{*}}=\frac{1}{K}(W')^{\dagger}$,
and
\[
z^{*}=a\cdot(W')^{\dagger}\widetilde{g},\qquad\text{for }a=\left(\frac{4\tr{\left((W'W)^{p}\right)}}{K}\right)^{1/(p+1)},
\]
as claimed.

When $W$ is block-diagonal, the relaxed problem can be solved block-by-block,
implying that $\Sigma^{\ast}$ is the same. Moreover, since the implied
$S$ is block-diagonal, Theorem \ref{prop:OptimalIV} also applies
block-by-block.

\subsection{\label{app:Theorem1}Proof of Theorem \ref{thm:main} and Associated
Claims}

This section first proves Theorem \ref{thm:main}. We then provide
brief proofs of claims made in the main text about sufficient conditions
for Assumption \ref{ass:asymptotic}. Finally, we prove that $\hat{v}\toP v^{*}$.

Parts (a) and (b) of Assumption \ref{ass:asymptotic} hold with probability
converging to one, so we condition on the event that those claims
hold; this does not affect statements of convergence-in-probability
and convergence-in-distribution. We start by decomposing the estimator
error as: 
\begin{equation}
\sqrt{K}(\bhat-\beta)=\frac{K^{-1/2}b'\tilde{g}}{K^{-1}\tilde{g}'B\left(\tfrac{1}{2}\one+\tilde{g}\right)},\label{eq:ratio}
\end{equation}
where for this section we use the shorthand notation $\tilde{g}=g^{GR}-0.5\cdot\one{}{}$
for the recentered Gaussian rounded shocks, $S$ instead of $S_{GR}$
and $B=W'(S^{\dag})^{1/(p+1)}W$. The proof of the theorem then follows
by combining the decomposition in Equation (\ref{eq:ratio}) with
Lemma \ref{lem:numerator} and Lemma \ref{lem:denominator} by Slutsky's
theorem. Lemma \ref{lem:numerator} applies with $m=1$ and $\Psi=b$,
since $K^{-1}b'\Sigma^{GR}b=K^{-1}\varepsilon'S^{(p-1)/(p+1)}\varepsilon=v_{K}\toP v^{*}>0$
by Assumption \ref{ass:asymptotic}(d) and the third-moment condition
is Assumption \ref{ass:asymptotic}(e). Lemma \ref{lem:denominator}
applies to $B$, since, by Assumption \ref{ass:asymptotic}(a) and
(b),
\begin{align*}
\op B & \le\op W^{2}\op{(S^{\dag})^{1/(p+1)}}=\lambda_{\max}(W'W)\{\lambda^{+}_{\min}(S)\}^{-1/(p+1)}\le\bar{\lambda}\cdot\underline{\lambda}^{-1/(p+1)},\\
\expec{}{\tilde{g}'B\left(\tfrac{1}{2}\one+{\tilde{g}}\right)\mid W} & =\tr(B\Sigma^{GR})=\tr{\left\{ (S^{\dag})^{1/(p+1)}S\right\} }=\tr(S^{p/(p+1)})=Kh_{K}.
\end{align*}
The convergence in probability for the denominator to $h^{*}$ then
follows since $h_{K}\toP h^{*}$ by Assumption \ref{ass:asymptotic}(c).
For later use, we also note that $\op S\le\op W^{2}\op{\Sigma^{GR}}\le(\bar{d}+1)\bar{\lambda}/4$
by Assumption \ref{ass:asymptotic}(a) and the bound on $\op{\Sigma^{GR}}$
in the proof of Lemma \ref{lem:denominator}.

Next, we prove Lemmas for the numerator and denominator that are general
enough to apply to the multiple exposures case as well. For the numerator,
we use an existing result on dependency graph CLTs from \citet[Theorem 2.4]{chen2004normal}.\footnote{This result would allow us to grow the degree of $\Sigma^{*}$ slowly,
where the rate of growth has to be slower for higher $p$, but to
keep things simple we assume that the degree does not grow with $K$.} The denominator law of large numbers (LLN) provides a more direct
argument that also relies on sparsity in $\Sigma^{*}$.
\begin{lem}
\label{lem:numerator}Let $m$ be a fixed positive integer and let
$\Psi=(\Psi_{1},\dots,\Psi_{K})'$ be a random $K\times m$ matrix
that is a function of $(W,U,\varepsilon)$. Suppose that $d(\Sigma^{*})\le\bar{d}$,
that $\boldsymbol{v}_{K}=K^{-1}\Psi'\Sigma^{GR}\Psi\toP{\boldsymbol{v}}^{*}$
for a positive definite $m\times m$ matrix $\boldsymbol{v}^{*}$,
and that $K^{-3/2}\sum^{K}_{k=1}\left\Vert \Psi_{k}\right\Vert ^{3}=o_{p}(1)$.
Then,
\[
K^{-1/2}\sum^{K}_{k=1}\Psi_{k}\tilde{g}_{k}\toD\mathcal{N}(0,\boldsymbol{v}^{*}).
\]
\end{lem}
\begin{proof}
By the Cramér–Wold theorem, it suffices to show that $K^{-1/2}\sum^{K}_{k=1}\lambda'\Psi_{k}\tilde{g}_{k}\toD\mathcal{N}(0,\lambda'\boldsymbol{v}^{*}\lambda)$
for every $\lambda\in\mathbb{R}^{m}$. For $\lambda=0$ both sides
are degenerate at zero, so fix $\lambda\neq0$; then $\lambda'\boldsymbol{v}^{*}\lambda>0$
because $\boldsymbol{v}^{*}$ is positive definite. We first bound
the normal approximation error conditional on a sequence of $W$,
$U$ and $\varepsilon$. We write $Y_{k}=K^{-1/2}\lambda'\Psi_{k}\tilde{g}_{k}$,
so that 
\[
\var{}{\sum_{k}Y_{k}\mid W,U,\varepsilon}=\frac{1}{K}\lambda'\Psi'\Sigma^{GR}\Psi\lambda=\lambda'\boldsymbol{v}_{K}\lambda.
\]
Since $\lambda'\boldsymbol{v}_{K}\lambda\toP{\lambda}'\boldsymbol{v}^{*}\lambda>0$,
we have $\PR(\lambda'\boldsymbol{v}_{K}\lambda>0)\to1$ and, as above,
we work without loss of generality as if $\lambda'\boldsymbol{v}_{K}\lambda>0$.
Conditional on $W$, $U$ and $\varepsilon$, $Y$ is an element-wise
transformation of $\xi$, which is a Gaussian random vector. The transformation
means that independence properties pass from $\xi$ to $Y$ and Gaussianity
implies that if $\Sigma^{*}_{kl}=0$ for all $k$ in an index set
$A$ and all $l$ in an index set $B$, then the vector collecting
$Y_{k}$ for $k\in A$ and the vector collecting $Y_{l}$ for $l\in B$
are conditionally independent. Let $\mathcal{G}$ be the graph on
$K$ vertices, with an edge between vertex $k$ and $l\ne k$ if and
only if $\Sigma^{*}_{kl}\ne0$ and $k\ne l$.

Define the $s$-neighborhood of $k$ as $N^{s}[k]=\left\{ l:d_{\mathcal{G}}(k,l)\le s\right\} $,
where the distance measure is the shortest path between two vertices,
and we use the convention that $d_{\mathcal{G}}(k,k)=0$. $N^{1}[k]$,
$N^{2}[k]$ and $N^{3}[k]$ are three sets that play the role of the
sets $A_{i}$, $B_{i}$, $C_{i}$ from condition LD3 in \citet{chen2004normal}.
Specifically, conditionally on $W,U,\varepsilon$:
\begin{itemize}
\item $Y_{k}$ is independent of the vector collecting $Y_{t}$ for $t\notin N^{1}[k]$;
\item The vector collecting $Y_{l}$ for $l\in N^{1}[k]$ is independent
of the vector collecting $Y_{t}$ for $t\notin N^{2}[k]$ because
$N^{2}[k]$ contains all neighbors of those in $N^{1}[k]$;
\item The vector collecting $Y_{l}$ for $l\in N^{2}[k]$ is independent
of the vector collecting $Y_{t}$ for $t\notin N^{3}[k]$.
\end{itemize}
The next step in applying the theorem is computing an upper bound,
for any $k$, on:
\[
\begin{split} & \max\{|\{l:N^{3}[k]\cap N^{2}[l]\neq\emptyset\}|,|N^{3}[k]|\}\\
 & \leq\max\{|N^{5}[k]|,|N^{3}[k]|\}\\
 & \leq(\bar{d}+1)^{5}.
\end{split}
\]

Next, we define the conditional normal approximation error for the
centered and normalized random variable $\sum^{K}_{k=1}Y_{k}/\sqrt{\lambda'\boldsymbol{v}_{K}\lambda}$:
\[
\delta_{K}=\sup_{z}\left|\PR\left\{ \frac{\sum^{K}_{k=1}Y_{k}}{\sqrt{\lambda'\boldsymbol{v}_{K}\lambda}}\le z\mid W,U,\varepsilon\right\} -\Phi(z)\right|.
\]

We can now apply Theorem 2.4 of \citet{chen2004normal}, which gives
us the first line of the below display. 
\[
\begin{split}\delta_{K} & \leq75(\bar{d}+1)^{10}\sum^{K}_{k=1}\mathbb{E}\left[\lvert Y_{k}/\sqrt{\lambda'\boldsymbol{v}_{K}\lambda}\rvert^{3}\mid W,U,\varepsilon\right]\\
 & \leq75(\bar{d}+1)^{10}\frac{\left\Vert \lambda\right\Vert ^{3}K^{-3/2}}{8(\lambda'\boldsymbol{v}_{K}\lambda)^{3/2}}\sum^{K}_{k=1}\left\Vert \Psi_{k}\right\Vert ^{3}=o_{p}(1).
\end{split}
\]

The second line uses the fact that $|\tilde{g}_{k}|=1/2$, recalling
that $Y_{k}=K^{-1/2}\lambda'\Psi_{k}\tilde{g}_{k}$, and the Cauchy–Schwarz
inequality $|\lambda'\Psi_{k}|\le\left\Vert \lambda\right\Vert \left\Vert \Psi_{k}\right\Vert $.
The last equality uses the third-moment condition. Thus $\delta_{K}\toP0$.
Since $0\leq\delta_{K}\leq1$, the bounded convergence theorem gives
$\mathbb{E}[\delta_{K}]\to0$, and since $\bigl|\PR\{\sum^{K}_{k=1}Y_{k}/\sqrt{\lambda'\boldsymbol{v}_{K}\lambda}\le z\}-\Phi(z)\bigr|\le\mathbb{E}[\delta_{K}]$
for every $z$, we obtain the unconditional CLT, $\sum^{K}_{k=1}Y_{k}/\sqrt{\lambda'\boldsymbol{v}_{K}\lambda}\toD\mathcal{N}(0,1)$.
Since $\lambda'\boldsymbol{v}_{K}\lambda\toP{\lambda}'\boldsymbol{v}^{*}\lambda$,
Slutsky's theorem gives $\sum^{K}_{k=1}Y_{k}\toD\mathcal{N}(0,\lambda'\boldsymbol{v}^{*}\lambda)$.
\end{proof}

We finish by providing an LLN for the denominator of the estimator.
\begin{lem}
\label{lem:denominator} Suppose that $d(\Sigma^{*})\le\bar{d}$ and
let $B$ be a random $K\times K$ matrix that is a function of $W$
and $U$ and satisfies $\op B\le C_{B}$ for a finite constant $C_{B}$.
Then 
\[
K^{-1}\left\{ \tilde{g}'B\left(\tfrac{1}{2}\one+{\tilde{g}}\right)-\tr(B\Sigma^{GR})\right\} \toP0.
\]
\end{lem}
\begin{proof}
We have $\expec{}{\tilde{g}'B\left(\tfrac{1}{2}\one{}{}\right)\mid W,U}=0$,
and 
\[
\expec{}{\tilde{g}'B\tilde{g}\mid W,U}=\tr(B\Sigma^{GR}),
\]
so we start by bounding the conditional variance of $K^{-1}\tilde{g}'B\left(\tfrac{1}{2}\one+{\tilde{g}}\right)$
by $C/K$ for a finite constant $C$. The arcsin transformation does
not change sparsity, so the row-sum sparsity of $\Sigma^{*}$ passes
to $\Sigma^{GR}$. Since $|\Sigma^{GR}_{kl}|\le1/4$ for all $k,l$,
this means that every row has absolute sum of at most $(\bar{d}+1)/4$;
hence 
\[
\op{\Sigma^{GR}}\le(\bar{d}+1)/4.
\]
For the linear term, 
\[
\var{}{K^{-1}\tilde{g}'B(\tfrac{1}{2}\one){\mid}W,U}=K^{-2}\bigl(B(\tfrac{1}{2}\one){\bigr)}'\Sigma^{GR}\bigl(B(\tfrac{1}{2}\one){\bigr)}\le K^{-2}\op{\Sigma^{GR}}C^{2}_{B}K/4\le C_{1}/K,
\]
for a finite constant $C_{1}$. For the quadratic term, $\tilde{g}'B\tilde{g}=\tilde{g}'\tilde{B}\tilde{g}$
for the symmetric matrix $\tilde{B}=(B+B')/2$ with $\op{\tilde{B}}\le C_{B}$.
Then, using Lemma~\ref{lem:quad}, applied conditionally on $W$and
$U$, and the fact that $\rank(\tilde{B})\leq K$, we have: 
\[
\var{}{K^{-1}\tilde{g}'B\tilde{g}\mid W,U}\le CK^{-2}(\bar{d}+1)^{2}\Frob{\tilde{B}}^{2}\leq CK^{-2}(\bar{d}+1)^{2}\rank(\tilde{B})\op{\tilde{B}}^{2}\le C_{2}/K
\]
for a finite constant $C_{2}$. Using $\var{}{X+Y\mid W,U}\le2\var{}{X\mid W,U}+2\var{}{Y\mid W,U}$,
we get 
\[
\var{}{K^{-1}\tilde{g}'B\left(\tfrac{1}{2}\one+{\tilde{g}}\right)\,\middle|\,W,U}\le2(C_{1}+C_{2})/K.
\]
For every $\eta>0$, Chebyshev's inequality gives 
\[
\PR\left(\left.\left|K^{-1}\left\{ \tilde{g}'B\left(\tfrac{1}{2}\one+\tilde{g}\right)-\tr(B\Sigma^{GR})\right\} \right|>\eta\,\right|W,U\right)\le\frac{C}{\eta^{2}K}.
\]
Taking expectations yields 
\[
\PR\left(\left|K^{-1}\left\{ \tilde{g}'B\left(\tfrac{1}{2}\one+\tilde{g}\right)-\tr(B\Sigma^{GR})\right\} \right|>\eta\right)\le\frac{C}{\eta^{2}K}\to0.
\]
Therefore $K^{-1}\left\{ \tilde{g}'B\left(\tfrac{1}{2}\one+{\tilde{g}}\right)-\tr(B\Sigma^{GR})\right\} \toP0$.
\end{proof}

\begin{lem}
\label{lem:quad} Let $\tilde{g}_{k}\in\{-1/2,1/2\}$ be centered
random variables with a dependency graph whose maximum degree is at
most $\bar{d}$. Then, for every symmetric and fixed $K\times K$
matrix $B$, 
\[
\var{}{\tilde{g}'B\tilde{g}}\le C(\bar{d}+1)^{2}\Frob B^{2},
\]
where $C$ is a finite constant. 
\end{lem}
\begin{proof}
Because $\tilde{g}^{2}_{k}=1/4$, we can write 
\[
\tilde{g}'B\tilde{g}=\frac{1}{4}\tr(B)+Q,\qquad Q=2\sum_{k<l}B_{kl}\tilde{g}_{k}\tilde{g}_{l}.
\]
Hence $\var{}{\tilde{g}'B\tilde{g}}=\var{}Q$, and 
\[
\var{}Q=4\sum_{k<l}\sum_{a<b}B_{kl}B_{ab}\cov{}{\tilde{g}_{k}\tilde{g_{l}},\tilde{g}_{a}\tilde{g}_{b}}.
\]
We next count for fixed $k<l$, how many pairs $a<b$ can have nonzero
covariance with $\tilde{g}_{k}\tilde{g}_{l}$. As before, let $\mathcal{G}$
be the dependency graph of $\tilde{g}$ and let $N^{s}[k]=\{l:\dist_{\mathcal{G}}(k,l)\le s\}$.
To have non-zero covariance, at least one of $a$ or $b$ must be
in $N^{1}[k]\cup N^{1}[l]$, i.e. $\{a,b\}\cap\bigl(N^{1}[k]\cup N^{1}[l]\bigr)\ne\emptyset.$
After choosing one such endpoint, say $a\in N^{1}[k]\cup N^{1}[l]$,
the other endpoint must satisfy $b\in N^{1}[k]\cup N^{1}[l]\cup N^{1}[a]$;
otherwise the covariance is zero. So, the number of non-zero covariance
partners of the $kl$ pair is at most: 
\[
2\,\bigl|N^{1}[k]\cup N^{1}[l]\bigr|\max_{a\in N^{1}[k]\cup N^{1}[l]}\bigl|N^{1}[k]\cup N^{1}[l]\cup N^{1}[a]\bigr|\le12(\bar{d}+1)^{2}.
\]
Since $|\tilde{g}_{k}\tilde{g}_{l}|\le1/4$, all covariance terms
are bounded by a universal constant. Using the count above, the symmetry
of the covariance relation, and $|xy|\le(x^{2}+y^{2})/2$, 
\[
\begin{aligned}\var{}Q & \le C\sum_{\substack{k<l,\ a<b:\,\cov{}{\tilde{g}_{k}\tilde{g}_{l},\tilde{g}_{a}\tilde{g}_{b}}\ne0}
}|B_{kl}B_{ab}|\\
 & \le C\sum_{\substack{k<l,\ a<b:\,\cov{}{\tilde{g_{k}}\tilde{g}_{l},\tilde{g}_{a}\tilde{g}_{b}}\ne0}
}(B^{2}_{kl}+B^{2}_{ab})\\
 & \le C(\bar{d}+1)^{2}\sum_{k<l}B^{2}_{kl}\le C(\bar{d}+1)^{2}\Frob B^{2},
\end{aligned}
\]
where the constant $C$ can change from line to line. Thus $\var{}{\tilde{g}'B\tilde{g}}\le C(\bar{d}+1)^{2}\Frob B^{2}$.
\end{proof}

We next provide brief proofs of some claims made in the discussion
of Assumption \ref{ass:asymptotic}(a) in the main text.
\begin{prop}
Under the assumptions of Proposition \ref{prop:Optimize_l}, when
$p$ is an integer, the Proposition \ref{prop:Optimize_l} maximizer
$\Sigma^{\ast}$ of the relaxed problem (\ref{eq:relax2}), and therefore
its Gaussian rounding implementation $\Sigma^{GR}$, satisfy $d(\Sigma^{*})=d(\Sigma^{GR})\leq\left(d(W'W)\right)^{p}$.
\end{prop}
\begin{proof}
From Proposition \ref{prop:Optimize_l}, $\Sigma^{\ast}=\frac{1}{4}D^{-1/2}_{\ell^{\ast}}\frac{M^{p}_{\ell^{\ast}}}{\tr{\left(M^{p}_{\ell^{\ast}}\right)}}D^{-1/2}_{\ell^{\ast}}$,
where $M_{\ell}=D^{-1/2}_{\ell}W'WD^{-1/2}_{\ell}$, $D_{\ell}$ is
a non-singular diagonal matrix, and $\tr{\left(M^{p}_{\ell^{\ast}}\right)}$
is a non-zero scalar. Thus we can write
\[
\Sigma^{\ast}\propto D^{-1}_{\ell^{\ast}}\left(W'W\right)D^{-1}_{\ell^{\ast}}\left(W'W\right)\cdots D^{-1}_{\ell^{\ast}},
\]
with $p$ copies of $W'W$, or equivalently
\[
\Sigma^{\ast}_{kl}\propto\sum^{K}_{t_{1}=1}\dots\sum^{K}_{t_{p-1}=1}\frac{\left(W'W\right)_{kt_{1}}\cdot\left(W'W\right)_{t_{1}t_{2}}\cdot\ldots\cdot\left(W'W\right)_{t_{p-1}l}}{D_{kk}\cdot D_{t_{1}t_{1}}\cdot\ldots\cdot D_{ll}}.
\]
 Therefore, the non-zero entries of $\Sigma^{\ast}$ require at least
one path through $p$ non-zero entries of $W'W$. From any vertex
there are at most $\left(d(W'W)\right)^{p}$ such walks. Since the
non-zero elements are the same for $\Sigma^{\ast}$ and $\Sigma^{GR}$,
$d(\Sigma^{\ast})=d(\Sigma^{GR})$, so the inequality holds for $d(\Sigma^{GR})$
as well.
\end{proof}

For the next claim in the main text, if $W$ has row and column sparsity
bounded by constants $c$ and $r$ respectively, then $d(W'W)\leq cr$
and $\left(d(W'W)\right)^{p}$ are also bounded as long as $p$ is
a finite integer—implying the first part of Assumption \ref{ass:asymptotic}(a).
If the entries of $W$ are also bounded, then row and column sparsity
means that row and column sums of the absolute values of the entries
of $W$ are bounded. The largest eigenvalue of $W'W$, $\lambda_{\text{max}}(W'W)$,
which is the spectral norm of $W$, squared, is bounded by the product
of the maximum absolute column sum and maximum absolute row sum of
$W$. This yields the second part of Assumption \ref{ass:asymptotic}(a).

We next show $K\asymp N$ in non-degenerate cases. Suppose we have
both row and column sparsity, and a non-vanishing fraction of outcome
units are connected to at least some shocks as $K$ grows, meaning
there exist constants $\rho>0$ and $\underline{w}>0$ such that,
with probability $1-o(1)$, 
\[
\frac{1}{N}\sum^{N}_{i=1}\mathbf{1}\left\{ \sum^{K}_{k=1}W^{2}_{ik}\geq\underline{w}\right\} \geq\rho.
\]
On this event, 
\[
\begin{aligned}\rho\underline{w}N & \leq\sum^{N}_{i=1}\sum^{K}_{k=1}W^{2}_{ik}\\
 & =\tr(W'W)\\
 & \leq K\lambda_{\max}(W'W)\\
 & \leq\bar{\lambda}K.
\end{aligned}
\]
This implies that we cannot have $N/K\to\infty$.

Conversely, we have already shown that for a finite constant $C$,
$\lambda_{\max}(S)\leq C.$ So, 
\[
\begin{aligned}h_{K} & =\frac{1}{K}\tr{\left(S^{p/(p+1)}\right)}\\
 & \leq\frac{\rank(S)}{K}\lambda_{\max}(S_{GR})^{p/(p+1)}\\
 & \leq C^{p/(p+1)}\frac{N}{K},
\end{aligned}
\]
So $h_{K}\toP0$ if $N/K\to0$, contradicting Assumption \ref{ass:asymptotic}(c).
Therefore, $K\asymp N$.

Finally, we prove asymptotic validity of our confidence intervals
constructed using $\hat{v}$, as defined in the main text.
\begin{prop}
Under the assumptions of Theorem \ref{thm:main}, $\hat{v}\toP v^{*}$.
\end{prop}
\begin{proof}
We can write $\hat{v}=v_{K}-2(\bhat-\beta)K^{-1}x'M\varepsilon+(\bhat-\beta)^{2}K^{-1}x'Mx$,
where $M=S^{(p-1)/(p+1)}_{GR}$. By Theorem \ref{thm:main}, $\hat{\beta}-\beta=o_{p}(1)$
and, by Assumption \ref{ass:asymptotic}, $v_{K}-v^{\ast}=o_{p}(1)$.
To finish the proof, we need that $a_{K}=K^{-1}x'M\varepsilon$ and
$\omega^{2}_{K}=K^{-1}x'Mx$ are both $O_{p}(1)$. First, write $\omega^{2}_{K}\leq K^{-1}\op M\lambda_{\text{max}}(W'W)\left\Vert g\right\Vert ^{2}=O(1)$.
This is because $\op M=\lambda_{\max}(S)^{(p-1)/(p+1)}$ and, by the
argument above, $\lambda_{\max}(S)$\textbf{ }is bounded, as well
as $\left\Vert g\right\Vert ^{2}\leq K$. Second, by Cauchy-Schwarz,
since $M\succeq0$, $\left|a_{K}\right|\leq K^{-1}\sqrt{(x'Mx)(\varepsilon'M\varepsilon)}=\omega_{K}\sqrt{v_{K}}=O_{p}(1)$.
\end{proof}

\section{Additional Theoretical Results\label{appx:Additional-Results}}

\subsection{Interpretation with Heterogeneous Effects\label{appx:HetFX}}

Suppose the true outcome model is
\begin{align*}
y_{i} & =\beta_{i}x_{i}+\varepsilon_{i},
\end{align*}
where $\beta_{i}$ now represents a heterogeneous treatment effect
for unit $i$. For any design $\delta$, consider any partially-whitened
instrument from Theorem \ref{prop:OptimalIV}, $z^{\ast}=z^{\ast}_{\delta}$
with $z^{\ast}_{i}=\left(S^{\dag/(p+1)}_{\delta}\tilde{x}\right)_{i}$
for any $p$. Following \citet{borusyak_negative_2024,borusyak_optimal_2026},
we characterize  the IV estimand, i.e., ratio of $\expec{\delta}{\sum_{i}z_{i}y_{i}}$
and $\expec{\delta}{\sum_{i}z_{i}x_{i}}$. Treating $(\beta_{i},\varepsilon_{i})$
as fixed without loss of generality:
\begin{align*}
\frac{\expec{\delta}{\sum_{i}z^{\ast}_{i}y_{i}}}{\expec{\delta}{\sum_{i}z^{\ast}_{i}x_{i}}} & =\frac{\expec{\delta}{\sum_{i}z^{\ast}_{i}\beta_{i}x_{i}}}{\expec{\delta}{\sum_{i}z^{\ast}_{i}x_{i}}}\\
 & =\frac{\sum_{i}\cov{\delta}{z^{\ast}_{i},x_{i}}\beta_{i}}{\sum_{i}\cov{\delta}{z^{\ast}_{i},x_{i}}},
\end{align*}
which is a convex average of heterogeneous effects when $\cov{\delta}{z^{\ast}_{i},x_{i}}\ge0$
for all $i$. To see that this holds, write:
\[
\cov{\delta}{z^{\ast},x}=\expec{\delta}{z^{\ast}\tilde{x}'}=S^{\dag/(p+1)}_{\delta}S_{\delta}=S^{p/(p+1)}_{\delta}.
\]
This matrix is positive semidefinite so every diagonal element is
nonnegative. These convex weights are known given the known design
$\delta$ and exposure mapping of $x$.\footnote{We note, however, that the convex weighting result does not generally
extend to the nonlinear outcome model $y_{i}=y_{i}(x_{i})$, except
when the optimal instrument is unwhitened. The partially whitened
IV estimate identifies a weighted average of causal effects $\partial y_{i}(x)/\partial x$
at different margins $x$ with weights proportional to $\sum_{j}(S^{\dag/(p+1)}_{\delta})_{ij}\cov{\delta}{x_{j},\mathbf{1}[x_{i}\ge x]\mid y(\cdot)}$
(see Appendix S.1 of \citet{borusyak_optimal_2026}) which are not
guaranteed to be non-negative.}

The relaxed $p=1$ solution from Section \ref{subsec:Computation}
gives particularly interpretable weights. Combining the formula for
$\Sigma^{\ast}$ in Corollary \ref{cor:p1} with the above formulas
and $(S^{*}_{\delta})^{p/(p+1)}=(W\Sigma^{*}W^{\prime})^{1/2}$ gives:
\begin{align*}
\frac{\expec{\delta}{\sum_{i}z^{\ast}_{i}y_{i}}}{\expec{\delta}{\sum_{i}z^{\ast}_{i}x_{i}}} & =\sum_{i}\left(\frac{\sum_{k}w^{2}_{ik}/\left\Vert w_{\cdot k}\right\Vert _{2}}{\sum_{k}\left\Vert w_{\cdot k}\right\Vert _{2}}\right)\beta_{i}\\
 & =\sum_{i}\left(\sum_{k}\left(\frac{\left\Vert w_{\cdot k}\right\Vert _{2}}{\sum_{l}\left\Vert w_{\cdot l}\right\Vert _{2}}\right)\left(\frac{w^{2}_{ik}}{\left\Vert w_{\cdot k}\right\Vert ^{2}_{2}}\right)\right)\beta_{i}.
\end{align*}
Intuitively, each shock receives weight proportional to its column
norm $\left\Vert w_{\cdot k}\right\Vert _{2}$ and distributes that
weight across outcome units in proportion to their squared exposures.
Consequently, if all column norms are equal, the weight on $\beta_{i}$
is proportional to $\sum_{k}w^{2}_{ik}$; it is straightforward to
verify that these are exactly the weights that unwhitened recentered
IV applies under independent 50\% Bernoulli assignment.

In the model with both direct and spillover effects, this result generally
breaks down due to contamination bias \citep{goldsmith-pinkham_contamination_2024}
except in special cases. Specifically, rewrite (\ref{eq:model_direct})
with heterogeneous effects as:
\begin{align*}
y_{i} & =\beta_{i}x_{i}+\tau_{i}g_{i}+\varepsilon_{i},
\end{align*}
and note that since the instrument $z^{\ast}$ from Theorem \ref{prop:OptimalIV_direct}
is orthogonal to $g$ we can write the estimand for the spillover
effect as:
\[
\frac{\expec{\delta}{\sum_{i}z^{\ast}_{i}y_{i}}}{\expec{\delta}{\sum_{i}z^{\ast}_{i}x_{i}}}=\frac{\sum_{i}\cov{\delta}{z^{\ast}_{i},x_{i}}\beta_{i}}{\sum_{i}\cov{\delta}{z^{\ast}_{i},x_{i}}}+\frac{\sum_{i}\expec{\delta}{z^{\ast}_{i}\tilde{g}_{i}}\tau_{i}}{\sum_{i}\cov{\delta}{z^{\ast}_{i},x_{i}}}
\]
The first term is a weighted average of $\beta_{i}$, as before, but
the weights need not be convex.\footnote{Appendix A.1 of \citet{goldsmith-pinkham_contamination_2024} shows
convexity is guaranteed in an OLS regression with binary $x_{i}$,
since then the relevant weight matrix is proportional to an M-matrix.
For a given design, the signs of the weights can be checked directly
from the displayed covariance formula.} Moreover, the second term can be nonzero since we do not generally
have $\expec{\delta}{z^{\ast}_{i}\tilde{g}_{i}}=0$ for each $i$:
just $\expec{\delta}{\sum_{i}z^{\ast}_{i}\tilde{g}_{i}}=0$. An exception
is when direct effects are constant or otherwise sample-uncorrelated
with $\expec{\delta}{z^{\ast}_{i}\tilde{g}_{i}}$.

\subsection{Non-Recentered Estimators\label{appx:non-recentered}}

This appendix extends the main results to allow for non-recentered
instruments; specifically, to characterize the designs and instruments
which minimize a finite-sample approximate worst-case mean-squared
error (MSE) allowing for some bias. We show that the MSE-optimal design
and estimator have the same form as the baseline design and estimator
with recentering, but that the optimal design can have undesirable
properties. We further show that the MSE advantage of non-recentered
instruments disappears when appropriate covariates are included.

Consider an extension of the \citet{borusyak_optimal_2026} approximation,
now to approximate mean-squared error (MSE) of an IV estimator using
a possibly non-recentered instrument $z_{i}(g)$:
\begin{align*}
\mathcal{M}_{\delta,\mathcal{E}}\left[z\right] & =\frac{\expec{\delta,\mathcal{E}}{(\sum^{N}_{i=1}z_{i}(g)\varepsilon_{i})^{2}}}{\expec{\delta}{\sum^{N}_{i=1}z_{i}(g)x_{i}}^{2}}.
\end{align*}
Following the same steps as in the Appendix \ref{appx:Proof-Thm1}–\ref{appx:Proof-Prop1}
proofs of Theorem \ref{prop:OptimalIV} and Proposition \ref{prop:OptimalDesign},
one can show the MSE-optimal design and estimator 
\begin{align*}
\delta^{*},z^{*} & \in\arg\min_{\delta\in\mathcal{D},z}\max_{\mathcal{E}\in\mathcal{F}_{p}(\sigma)}\mathcal{M}_{\delta,\mathcal{E}}\left[z\right]
\end{align*}
are, for $M_{\delta}=\expec{\delta}{xx^{\prime}}$: 
\begin{align*}
\delta^{*} & \in\arg\max_{\delta\in\mathcal{D}}\tr{\left(M^{\frac{p}{p+1}}_{\delta}\right)},\\
z^{*} & \propto\left(M^{\dag}_{\delta^{*}}\right)^{\frac{1}{p+1}}x.
\end{align*}
The MSE-optimal instrument is generally not recentered (i.e., $\expec{\delta^{*}}{z^{*}_{i}(g)}\neq0$
in general), implying some small bias can help for this problem.

However, the optimal design can have undesirable properties. To see
this simply, consider the no-spillovers special case where $W=I$.
Previously, the optimal design was independent randomization with
$\expec{\delta^{*}}{g_{i}}=0.5$ for all $i$. Now we have the following:
\begin{prop}
\label{prop:mse_nospillovers}When $W=I$ and $p<\infty$, an optimal
design $\delta^{*}\in\arg\max_{\delta\in\mathcal{D}}\tr{\left(M^{\frac{p}{p+1}}_{\delta}\right)}$
is given by complete randomization with $K_{1}$ treated units, where
\begin{align*}
K_{1} & =\begin{cases}
\left\lfloor m^{*}\right\rfloor , & w.p.\,1-\theta^{*}\\
\left\lfloor m^{*}\right\rfloor +1, & w.p.\,\theta^{*}
\end{cases}
\end{align*}
for 
\begin{align*}
m^{*} & \in\arg\max_{m\in[0,N]}\left[\left(\frac{v(m)}{N}\right)^{\frac{p}{p+1}}+(N-1)\left(\frac{Nm-v(m)}{N(N-1)}\right)^{\frac{p}{p+1}}\right],\\
v(m) & =\left\lfloor m\right\rfloor ^{2}+(2\left\lfloor m\right\rfloor +1)(m-\left\lfloor m\right\rfloor ),\\
\theta^{*} & =m^{*}-\left\lfloor m^{*}\right\rfloor .
\end{align*}
\end{prop}
In the case where $m^{*}$ is an integer, this result shows complete
randomization with exactly $K_{1}=m^{*}$ units treated is MSE-optimal.
The expression for $m^{*}$ shows that when $p=1$ the optimal share
of treated units solves (ignoring discreteness)
\begin{align*}
\max_{t\in[0,1]} & t+\sqrt{N-1}\cdot\sqrt{t(1-t)}
\end{align*}
giving $t^{*}=(1+1/\sqrt{N})/2>1/2$. As $p\rightarrow\infty$ with
$N$ fixed, this treated share grows to one. Intuitively, without
recentering, the constant mean direction counts towards the “signal”
of the treatment leading to unusually high treatment rates.

A complementary result shows that the MSE advantage of non-recentered
estimators disappears (i.e., there is no MSE cost to recentering)
whenever $W\mathbf{1}$ is spanned by controls (included as in Appendix
\ref{appx:predetermined_covs}, below). In the previous no-spillovers
special case, this would hold when estimation includes an intercept:
\begin{prop}
\label{prop:mse_recentered}Let $R$ be a matrix of predetermined
controls, let $M_{R}=I-R(R^{\prime}R)^{\dag}R^{\prime}$ and assume
$W\mathbf{1}\in\text{Col}(R)$. Then there is no loss in only considering
recentered instruments when minimizing worst-case MSE; i.e. when solving
\begin{align*}
\delta^{*},z^{*} & \in\arg\min_{\delta\in\mathcal{D},z}\max_{\mathcal{E}\in\mathcal{F}_{p}(\sigma)}\frac{\expec{\delta,\mathcal{E}}{(z(g)^{\prime}M_{R}\varepsilon)^{2}}}{\expec{\delta}{z(g)^{\prime}M_{R}x}^{2}},
\end{align*}
it is without loss to limit to $z$ with $\expec{\delta}{z_{i}(g)}=0$
for all $i$.
\end{prop}
Intuitively, with an intercept, the non-recentered instrument can
no longer leverage the constant mean direction which can yield unusually
high treatment rates (while allowing for some bias). Once an intercept
or, more generally, the expected treatment is absorbed by the controls
(i.e. with $W\mathbf{1}$ controlled for) there is no reason to not
recenter. 

\subsection{Equicorrelated \texorpdfstring{$\varepsilon_{i}$}{εi}\label{appx:Equicorrelated}}

Consider $\varepsilon\in\mathbb{R}^{N}$ where $\expec{}{\varepsilon}=0$
and for some $\nu>0$ and $\rho\in[-(N-1)^{-1},1]$,
\begin{align*}
\Omega_{\rho} & =\expec{}{\varepsilon\varepsilon'}=\nu^{2}\left[(1-\rho)I+\rho\one{\one{}{}}'\right].
\end{align*}
Here $\nu^{2}=\var{}{\varepsilon_{i}}$ and $\rho=\corr{}{\varepsilon_{i},\varepsilon_{j}}$
for $j\ne i$. Since $\Omega_{\rho}=\nu^{2}\left[(1-\rho)(I-\one{\one{}{}}'/N)+(1+(N-1)\rho)\one{\one{}{}}'/N\right]$
by Lemma \ref{lem:equicor}: 
\begin{align*}
\left\Vert \Omega_{\rho}\right\Vert _{p} & =\left[\frac{1}{N}(1+(N-1)\rho)^{p}+\frac{N-1}{N}(1-\rho)^{p}\right]^{1/p}\cdot\nu^{2}N^{1/p}.
\end{align*}
Thus, the error distribution belongs to $\mathcal{F}_{p}(\sigma)$
for some $\sigma\ge\nu$ if and only if 
\begin{equation}
\left[\frac{1}{N}(1+(N-1)\rho)^{p}+\frac{N-1}{N}(1-\rho)^{p}\right]^{1/p}\le\frac{\sigma^{2}}{\nu^{2}}.\label{eq:equicorr_condition}
\end{equation}
This condition is illustrated in Figure \ref{fig:Equicorrelation}
which shows the maximum $\nu^{2}$ corresponding to each $\rho$,
depending on $N\in\left\{ 10,50\right\} $ and $p\in\left\{ 1,2,\infty\right\} $.
It is evident that smaller $p$ makes the condition weaker: For $p=1$,
it holds for any $\rho$ and $N$. But for larger $p$, the condition
penalizes correlations, especially in large samples: $\nu^{2}$ has
to be much smaller to stay within $\mathcal{F}_{p}(\sigma)$.

\subsection{Rank-1 Designs\label{appx:Rank-1-Designs}}
\begin{prop}
\label{prop:Rank1-designs}When $p=\infty$, there is always an optimal
design $\delta^{*}$ supported on a single pair of complementary assignments
$g^{*}$ and $\one-g^{*}$ so $\var{\delta}g$ has rank 1. Moreover,
when $W\ge0$ elementwise, one such design corresponds to setting
$g=\mathbf{1}$ with probability $0.5$ and $g=\mathbf{0}$ otherwise.
\end{prop}

\subsection{Symmetry Normalization\label{appx:Symmetry-Normalization}}
\begin{prop}
\label{prop:Symmetry}Suppose $\mathcal{G}$ is a finite group of
relabelings of observations and shocks represented by pairs of permutation
matrices $(P,Q)$ such that $PWQ^{\prime}=W$ for every $(P,Q)\in\mathcal{G}$.
Then there exists a solution $\delta^{*}$ to Proposition \ref{prop:OptimalDesign}
such that $\expec{\delta^{\ast}}{g_{k}}=0.5$, for all $k$ and $\delta^{*}$
is $\mathcal{G}$-invariant: i.e., if $g\sim\delta^{*}$ then $Qg$
has the same distribution as $g$ for every $(P,Q)\in\mathcal{G}$.
Equivalently, $\delta^{*}$ may be chosen to assign equal probability
to assignments in the same $\mathcal{G}$-orbit, i.e. for any $g,g'\in\left\{ 0,1\right\} ^{K}$
such that $g'=Qg$ for some $(P,Q)\in\mathcal{G}$. 
\end{prop}
In particular, this result shows that if $W$ is invariant under permutations
of a subset of shocks (possibly together with corresponding permutations
of observations), then we may without loss restrict attention to designs
that are exchangeable within that subset. Similarly, if $W$ has multiple
identical blocks, we may take the induced block-level design to be
exchangeable across those blocks. A version of this result also goes
through for the multiple-exposure extension, with $W$ replaced by
$W-cU$, as long as $W-cU$ is invariant under $\mathcal{G}$ for
every $c$.

\subsection{Special Case: Leave-Out Averages\label{appx:Leave-Out-Averages}}

In this special case the intervention units coincide with outcome
units, which are partitioned into $J$ clusters $C_{j}$ of sizes
$N_{j}>1$, with $x_{i}=\bar{g}_{-i}=\frac{1}{N_{j(i)}-1}\sum_{k\in C_{j(i)},k\neq i}g_{k}$.
We know from Proposition~\ref{prop:DesignProperties} (Proposition~\ref{prop:DesignProperties_direct}
with direct effects) that it is optimal to randomize independently
across clusters with $\expec{\delta}{g_{i}}=0.5$. For $p<\infty$,
we solve the relaxed problem (\ref{eq:relax2}), or (\ref{eq:relaxp})
with direct effects, cluster by cluster, and propose a feasible implementation
of this solution.

For cluster $j$, write $P_{j}=\one{}{}_{N_{j}}\one{}{}_{N_{j}}'/N_{j}$
and, for $c\in\mathbb{R}$, 
\[
d_{j}(c)=\frac{\left(N_{j}-1\right)\left|\frac{1}{N_{j}-1}+c\right|^{2p}+\left|1-c\right|^{2p}}{N_{j}}>0.
\]
Here $c=0$ in the case without direct effects, while with direct
effects ($h_{i}=g_{i}$, so $U=I_{N}$) $c$ is the residualization
coefficient in (\ref{eq:relaxp}). Focusing on cluster $j$, the within-cluster
(residualized) exposure matrix can be written as 
\begin{equation}
W_{(j),c}=W_{(j)}-cI_{N_{j}}=\frac{\one{}{}_{N_{j}}\one{}{}_{N_{j}}'-I_{N_{j}}}{N_{j}-1}-cI_{N_{j}}=-\left(\frac{1}{N_{j}-1}+c\right)\left(I_{N_{j}}-P_{j}\right)+\left(1-c\right)P_{j},\label{eq:wj}
\end{equation}
and, by Lemma~\ref{lem:equicor}, 
\[
\left(W_{(j),c}'W_{(j),c}\right)^{p}=\left|\frac{1}{N_{j}-1}+c\right|^{2p}\left(I_{N_{j}}-P_{j}\right)+\left|1-c\right|^{2p}P_{j},
\]
with constant diagonal elements $d_{j}(c)$. For fixed $c$, we thus
apply Corollary~\ref{cor:sym_diagonal} to get 
\begin{equation}
4\Sigma^{\ast}_{(j)}=\frac{\left|\frac{1}{N_{j}-1}+c\right|^{2p}\left(I_{N_{j}}-P_{j}\right)+\left|1-c\right|^{2p}P_{j}}{d_{j}(c)},\label{eq:sig_j}
\end{equation}
such that the off-diagonal elements of the within-cluster correlation
matrix $4\Sigma^{\ast}_{(j)}$ are all equal to 
\[
\psi_{j}(c)=\frac{\left|1-c\right|^{2p}-\left|\frac{1}{N_{j}-1}+c\right|^{2p}}{N_{j}d_{j}(c)}.
\]
The same corollary characterizes the optimal spillover instrument,
\begin{equation}
z^{\ast}_{i}=\left(4d_{j}(c)\right)^{1/(p+1)}\left[\left(1-c\right)^{\dag}\left(\bar{g}_{j}-0.5\right)+\left(\frac{1}{N_{j}-1}+c\right)^{\dag}\left(\bar{g}_{j}-g_{i}\right)\right],\qquad i\in C_{j},\label{eq:loo_iv_general}
\end{equation}
where $t^{\dag}=t^{-1}$ if $t\neq0$ and $0$ otherwise for scalar
$t$, and where the first term is cluster-level and the second term
averages to zero within the cluster.

\paragraph{No Direct Effects.}

Set $c=0$ and $r_{j}=(N_{j}-1)^{-2p}$. Then $d_{j}(0)=\left(1+(N_{j}-1)r_{j}\right)/N_{j}$
and 
\[
\psi_{j}=\psi_{j}(0)=\frac{1-r_{j}}{1+(N_{j}-1)r_{j}},
\]
so that $d_{j}(0)=\left(1+(N_{j}-1)\psi_{j}\right)^{-1}$. Since $r_{j}\in(0,1]$
and therefore $\psi_{j}\in[0,1]$, this relaxed problem solution can
be implemented by mixing clustered assignment and the Bernoulli design
with probabilities $\psi_{j}$ and $1-\psi_{j}$, independently across
clusters. It can also be implemented by a saturation design with saturations
0, 0.5, and 1. By (\ref{eq:loo_iv_general}), the optimal instrument
is 
\[
z^{\ast}_{i}\propto\left(1+(N_{j}-1)\psi_{j}\right)^{-1/(p+1)}\left[\left(\bar{g}_{j}-0.5\right)+\left(N_{j}-1\right)\left(\bar{g}_{j}-g_{i}\right)\right],\qquad i\in C_{j}.
\]

\paragraph{With Direct Effects.}

For fixed $c$, let $S_{(j),c}=W_{(j),c}\Sigma^{\ast}_{(j)}W_{(j),c}'$.
Taking the product of (\ref{eq:sig_j}) and (\ref{eq:wj}), using
that $P_{j}$ and $I_{N_{j}}-P_{j}$ are symmetric orthogonal projections
with $P_{j}(I_{N_{j}}-P_{j})=0$, 
\[
S_{(j),c}=\frac{\left|\frac{1}{N_{j}-1}+c\right|^{2p+2}\left(I_{N_{j}}-P_{j}\right)+\left|1-c\right|^{2p+2}P_{j}}{4d_{j}(c)},
\]
so Lemma~\ref{lem:equicor} with $\gamma=p/(p+1)$ gives
\[
\tr{\left(S^{p/(p+1)}_{(j),c}\right)}=\frac{\left(N_{j}-1\right)\left|\frac{1}{N_{j}-1}+c\right|^{2p}+\left|1-c\right|^{2p}}{\left(4d_{j}(c)\right)^{p/(p+1)}}=4^{-p/(p+1)}N_{j}d_{j}(c)^{1/(p+1)}.
\]

Then $c^{\ast}$ corresponding to the optimal design minimizes 
\begin{equation}
\sum^{J}_{j=1}N_{j}d_{j}(c)^{1/(p+1)},\label{eq:loo_c_min}
\end{equation}
and the relaxed problem solution and optimal spillover instrument
are the expressions above evaluated at $c=c^{\ast}$.

Let $\psi^{\perp}_{j}=\psi_{j}(c^{\ast})$. In general, we cannot
guarantee that $\psi^{\perp}_{j}\geq0$. However, since $\psi^{\perp}_{j}\geq-1/(N_{j}-1)$
by positive semi-definiteness of $4\Sigma^{\ast}_{(j)}$, for even
$N_{j}$ any such $\Sigma^{\ast}_{(j)}$ is implementable by mixing
cluster-level shocks, with probability $(1+(N_{j}-1)\psi^{\perp}_{j})/N_{j}$,
with complete randomization within the cluster, whose within-cluster
correlation is $-1/(N_{j}-1)$. Thus, it is implementable by a saturation
design with saturations 0, 0.5, and 1. \footnote{This is not exactly true for odd $N_{j}$, where the most negative
correlation implementable by binary shocks is $-1/N_{j}$ (compared
to the $-1/(N_{j}-1)$ bound guaranteed by positive semi-definiteness).
If $\psi^{\perp}_{j}\geq-1/N_{j}$, $\Sigma^{\ast}_{(j)}$ is implementable
by mixing cluster-level shocks with the design that randomly chooses
whether to treat $(N_{j}-1)/2$ or $(N_{j}+1)/2$ units in the cluster
and then randomly chooses which ones. }

We now consider two special cases. First, when all clusters are of
the same size, $N_{j}=n$, minimizing (\ref{eq:loo_c_min}) is equivalent
to minimizing the strictly convex function $d_{j}(c)$, with the first-order
condition that implies 
\[
\frac{\frac{1}{n-1}+c^{\ast}}{1-c^{\ast}}=(n-1)^{-1/(2p-1)},\qquad\text{i.e.,}\qquad c^{\ast}=\frac{(n-1)^{-1/(2p-1)}-(n-1)^{-1}}{1+(n-1)^{-1/(2p-1)}}\in\left(-\frac{1}{n-1},1\right).
\]
Plugging this in gives $r^{\perp}=\left(\frac{\frac{1}{n-1}+c^{\ast}}{1-c^{\ast}}\right)^{2p}=(n-1)^{-2p/(2p-1)}\in(0,1]$
with a common optimal mixing probability of 
\[
\psi^{\perp}=\psi_{j}(c^{\ast})=\frac{1-r^{\perp}}{1+(n-1)r^{\perp}}\in[0,1].
\]
Thus, the design is implementable by mixing cluster-level assignment
with Bernoulli randomization. Noting that $d_{j}(c^{\ast})$ and $1-c^{\ast}>0$
are common to all clusters, the optimal spillover instrument (\ref{eq:loo_iv_general})
simplifies to 
\[
z^{\ast}_{i}\propto\left(\bar{g}_{j(i)}-0.5\right)+(n-1)^{1/(2p-1)}\left(\bar{g}_{j(i)}-g_{i}\right).
\]

Second, when $p=1$ (but cluster sizes can vary), $d_{j}(c)=\frac{1}{N_{j}-1}+c^{2}$,
so every summand in (\ref{eq:loo_c_min}) is minimized at $c=0$ and
$c^{\ast}=0$. Thus, the design and optimal IV are the same as in
the setting without direct effects; see footnote \ref{fn:direct-doesnt-change-design}
for the intuition.

\subsection{Special Case: Spatial Spillovers\label{appx:Spatial-Spillovers}}

Fix $p=1$ and assume for simplicity that $2L+1$ divides $N$ and
$2L+1\le N/2$. That is, the circle can be split into an integer number
of blocks of length $2L+1$, at least two of them.

We first solve the relaxed problem. The structure of spillovers implies
$\left\Vert w_{\cdot k}\right\Vert _{2}=1/\sqrt{2L+1}$ and $w_{\cdot k}'w_{\cdot l}=\frac{1}{\left(2L+1\right)^{2}}\max\left\{ 2L+1-d(k,l),0\right\} $
for $k,l\in\left\{ 1,\dots,N\right\} $, where $d(k,l)$ is the circular
distance between units. Thus, by Corollary \ref{cor:p1} the covariances
in the relaxed problem's solution,
\[
\Sigma^{\ast}_{kl}=\frac{1}{4}\max\left\{ 1-\frac{d(k,l)}{2L+1},0\right\} ,
\]
decay linearly reaching zero at distance $2L+1$. This covariance
matrix can be implemented with binary shocks, specifically by the
procedure described in the main text. Indeed, the correlation between
$g_{k}$ and $g_{l}$ equals the probability that they are in the
same block, which equals $\max\left\{ 1-\frac{d(k,l)}{2L+1},0\right\} $,
as required.

For this symmetric design, Corollary \ref{cor:sym_diagonal} implies
$z^{\ast}\propto W^{\dag}\tilde{g}$. We now verify the characterization
of $z^{\ast}$ in the main text. First, since an $L$-neighborhood
of any unit $i$ includes the center of the block to which $i$ belongs
and no other block center, averaging the values of $z^{\ast}$ within
a neighborhood yields $\frac{1}{2L+1}\left((2L+1\right)\tilde{g}_{i}-2L\bar{\tilde{g}}+2L\bar{\tilde{g}})=\tilde{g}_{i}$.
Thus, $Wz^{\ast}=\tilde{g}$. For this $z^{\ast}$ to be $W^{\dag}\tilde{g}$,
it remains to show that $z^{\ast}\perp\ker\left(W\right)$, i.e. $z^{\ast}{}'v=0$
for any $v$ such that $Wv=0$. The structure of $W$ implies that
its kernel consists of vectors $v$ that are $(2L+1)$-periodic with
zero sums in each block: $v_{i+(2L+1)}=v_{i}$ and $\sum^{2L+1}_{i=1}v_{i}=0$.
Moreover, the structure of $z^{\ast}$ implies that its sums over
any $(2L+1)$-spaced orbit (i.e., set of units with the same position
relative to the block center) are the same, $\frac{N}{2L+1}\bar{\tilde{g}}$.
Thus, $\sum^{N}_{i=1}z^{\ast}_{i}v_{i}=\sum^{2L+1}_{i=1}\frac{N}{2L+1}\bar{\tilde{g}}v_{i}=0$,
which concludes the proof. We note that, despite $p=1$, this optimal
IV is not fully whitened because $S$ is singular in this setting.

\subsection{Predetermined Covariates\label{appx:predetermined_covs} }

This appendix extends the optimal design and instrument results to
specifications with predetermined controls. We allow the model error
to contain an arbitrary component in the span of the controls and
show that the controlled problem is exactly the baseline problem after
residualizing exposures and instruments: i.e., all formulas apply
with $W$ replaced by $M_{R}W$.

Let $R=(r_{i}')$ be an $N\times L$ matrix of predetermined covariates
and let $M_{R}=I-R(R^{\prime}R)^{\dag}R^{\prime}$ be the associated
residual-maker matrix. We expand the class of allowed error distributions
to include linear combinations of $r_{i}$ with unrestricted coefficients:
\begin{align*}
\mathcal{F}^{R}_{p}(\sigma) & =\left\{ \mathcal{E}:\varepsilon_{i}=\gamma'r_{i}+\tilde{\varepsilon}_{i},\,\gamma\in\mathbb{R}^{L},\,\left\Vert \expec{\mathcal{\mathcal{E}}}{\tilde{\varepsilon}\tilde{\varepsilon}^{\prime}}\right\Vert _{p}\le N^{1/p}\sigma^{2}\right\} .
\end{align*}
Let 
\begin{align*}
\tilde{x}_{\delta,R} & =M_{R}W\tilde{g}_{\delta},\\
S_{\delta,R}=\var{\delta}{\tilde{x}_{\delta,R}} & =M_{R}W\Sigma_{\delta}W'M_{R}.
\end{align*}
We then have the following results:
\begin{prop}
\label{prop:optimal_iv_wcovs}Fix $\delta$ and suppose $S_{\delta,R}\neq0$.
Then the minimax optimal recentered IV is:
\begin{align*}
z^{*}_{\delta,R}(g) & =\left(S^{\dag}_{\delta,R}\right)^{\frac{1}{p+1}}\tilde{x}_{\delta,R}\in\arg\min_{z\in\mathcal{Z}_{\delta}}\max_{\mathcal{E}\in\mathcal{F}^{R}_{p}(\sigma)}\V{\delta,\mathcal{E}}z,
\end{align*}
with the worst-case approximate variance of
\begin{align*}
\max_{\mathcal{E}\in\mathcal{F}^{R}_{p}(\sigma)}\V{\delta,\mathcal{E}}{z^{*}_{\delta,R}} & =N^{1/p}\sigma^{2}\tr(S^{p/(p+1)}_{\delta,R})^{-(p+1)/p}.
\end{align*}
\end{prop}
\begin{prop}
\label{prop:optimal_design_wcovs}Suppose there is $\delta\in\mathcal{D}$
such that $S_{\delta,R}\ne0$. Then any 
\begin{align*}
\delta^{*}_{R} & \in\arg\max_{\delta\in\mathcal{D}}\tr(S^{p/(p+1)}_{\delta,R})
\end{align*}
together with the instrument from Proposition \ref{prop:optimal_iv_wcovs}
solves the minimax problem.
\end{prop}
\begin{prop}
\label{prop:design_properties_wcovs} There exists a maximizer $\delta^{*}_{R}$
from Proposition \ref{prop:optimal_design_wcovs} such that $\expec{\delta^{*}_{R}}{g_{k}}=0.5$,
for all $k$. Moreover, if $M_{R}W$ is block-diagonal after a joint
partition of rows and columns, then there exists such a $\delta^{*}_{R}$
with independent subvectors of $g$ across blocks.
\end{prop}
The proofs here follow by showing the $R$-controlled problem is equivalent
to the original problem with $W$ replaced by the columnwise residualized
matrix $M_{R}W$. Notice that the optimal instrument $z^{*}_{\delta,R}(g)$
is always in-sample orthogonal to $R$; variation in directions absorbed
by $R$ is not rewarded by the objective and so is not used. A sufficient
condition for block-diagonality of $M_{R}W$ is that $M_{R}$ and
$W$ are jointly block-diagonal after the same joint partition of
rows and columns, such as when all covariates are included with interactions
with the dummies of blocks in $W$.

\subsection{Nonlinear Spillover Formulas\label{appx:nonlinear}}

It is straightforward to extend Theorem \ref{prop:OptimalIV} and
Proposition \ref{prop:OptimalDesign} to settings where $x_{i}=X_{i}(g)$
for a nonlinear (but still known) set of formulas $X(\cdot)=(X_{1}(\cdot),\dots,X_{N}(\cdot))'$
which can depend on $W$, as studied in \citet{BH1}. Indeed, the
Appendix \ref{appx:Proof-Thm1}-\ref{appx:Proof-Prop1} proofs extend
verbatim after defining the general $\tilde{x}_{\delta}=x-\expec{\delta}x$
and $S_{\delta}=\var{\delta}x$. The central challenge with this extension
is computation: unlike in the linear case, the optimal design objective
is not reducible to a search over covariance matrices of the $g$
shocks and is thus generally much more difficult. Still, a researcher
can compute and maximize this objective over a relatively small set
of candidate designs $\mathcal{D}_{r}\subset\mathcal{D}$:
\[
\max_{\delta\in\mathcal{D}_{r}}\tr{\left(\var{\delta}{X(g)}^{p/(p+1)}\right)}.
\]
The second and third claims of Proposition \ref{prop:DesignProperties}
extend similarly.

However, the first claim of Proposition \ref{prop:DesignProperties}—that
there is an optimal design with $\expec{\delta}{g_{k}}=0.5$—does
not generalize. We illustrate this and build intuition for nonlinear
solutions in a special case in which $N=K$ observations are clustered
into groups $C_{j}$ and $x_{i}=\mathbf{1}\left[\sum_{k\in C_{j(i),k\neq i}}g_{k}\ge1\right]$
indicates that $i$ has at least one shocked neighbor in his group.
Without loss, consider a design which optimizes group-by-group and
is exchangeable within groups. Focusing on one cluster and dropping
the $j$ subscript for brevity, let $T=\sum_{i\in C}g_{i}$ and
\begin{align*}
a & =Pr_{\delta}(T=0)\\
b & =Pr_{\delta}(T=1)\\
1-a-b & =Pr_{\delta}(T\ge2),
\end{align*}
with the uniquely shocked unit chosen at random when $T=1$. We then
have the following result:
\begin{prop}
\label{prop:one_treated_neighbor}Fix $n=|C|\ge3$ and $p\in[1,\infty)$
in the above nonlinear spillover formula case. Then among exchangeable
designs,
\begin{align*}
\max_{\delta} & \tr(S^{p/(p+1)}_{\delta})
\end{align*}
has a unique solution where $Pr_{\delta}(T=0)=a^{*}$ and $Pr_{\delta}(T=1)=1-a^{*}$
for $a^{*}$ solving
\begin{align*}
1-2a^{*} & =(a^{*})^{1/(p+1)}(n-1)^{-(p-1)/(p+1)}.
\end{align*}
Hence the optimal design never shocks more than one unit in the group,
and $\expec{\delta}{g_{k}}\ne1/2$. Moreover:
\begin{enumerate}
\item[(a)] \textup{}\emph{If $n=3$ or $p=1$ then $a^{*}=0.25$;}
\item[(b)] \textup{}As $p\rightarrow\infty$, $a^{*}\rightarrow\frac{n-2}{2(n-1)}$;
\item[(c)] \textup{}As $n\rightarrow\infty$ for $p>1$, $a^{*}\rightarrow1/2$.
\end{enumerate}
\end{prop}
Intuitively, the optimal design mixes between shocking nobody in a
cluster or shocking exactly one unit since assignments with two or
more shocks produce the same exposure vector. As isotropy becomes
less important ($p\rightarrow\infty$), the mixing probability approaches
the value which maximizes the total variance of the $x_{i}$.\footnote{One can extend the result to $p=\infty$; as before, the optimal exchangeable
design is not unique and here it does not generally set $Pr_{\delta}(T>1)=0$.} As group sizes get large ($n\rightarrow\infty$), the mixing probability
approaches $1/2$ because the identity of the shocked unit becomes
second-order and we approach the group-specific shock special case.

\subsection{\label{app:Theorem2_direct}Extension of Theorem \ref{thm:main}
to Multiple Exposures}

We extend Theorem \ref{thm:main} to the multiple-exposure model (\ref{eq:multiple_exposure}).
The design is the Gaussian rounding implementation of the solution
$\Sigma^{*}$ of the relaxed problem (\ref{eq:relaxp}), and the instrument
is the Theorem \ref{prop:OptimalIV_direct} optimal instrument vector
$\boldsymbol{z}^{*}_{\delta}$ evaluated at the rounded solution to
the relaxed problem, so that $c_{\delta}$ solves (\ref{eq:c_delta})
with $\Sigma_{\delta}=\Sigma^{GR}$.\footnote{Note that $c_{\delta}$ from (\ref{eq:c_delta}) from the instrument
is different from $c$ from (\ref{eq:relaxp}) because Gaussian rounding
is applied.} In addition to the notation of Appendix \ref{app:Theorem1}, we write
$W^{\perp}=W-c_{\delta}U$, $S^{\perp}=W^{\perp}\Sigma^{GR}W^{\perp\prime}$
and $\tilde{h}=U\tilde{g}$, and collect the instruments and exposures
in the $N\times2$ matrices $Z=(\boldsymbol{z}^{*}_{\delta}(g)_{i}')^{N}_{i=1}$
and $X=(\boldsymbol{x}_{i}')^{N}_{i=1}$, so that $\hat{\theta}=(Z'X)^{-1}Z'y$.

\renewcommand{\theassumption}{1$\,^\prime$} 
\begin{assumption}
\label{ass:asymptotic_direct}There are constants $\underline{\lambda}>0$,
$\bar{\lambda}\in(0,\infty)$, $\bar{d}<\infty$, an invertible $2\times2$
matrix $\boldsymbol{h}^{*}$, and a positive definite $2\times2$
matrix $\boldsymbol{v}^{*}$ such that:
\begin{enumerate}
\item[(a)] $d(\Sigma^{\ast})\le\bar{d}$ and $\max\left\{ \lambda_{\max}(W'W),\lambda_{\max}(U'U),\lambda_{\max}(W^{\perp\prime}W^{\perp})\right\} \le\bar{\lambda}$;
\item[(b)] $\lambda^{+}_{\min}(S^{\perp})>\underline{\lambda}$;
\item[(c)] $\boldsymbol{h}_{K}=K^{-1}\expec{}{Z'X\mid W,U}\toP{\boldsymbol{h}}^{*}$;
\item[(d)] $\boldsymbol{v}_{K}=K^{-1}\Psi'\Sigma^{GR}\Psi\toP{\boldsymbol{v}}^{*}$
for $\Psi=(b,a)$, $b=W^{\perp\prime}\left((S^{\perp})^{\dag}\right)^{1/(p+1)}\varepsilon$
and $a=U'\varepsilon$;
\item[(e)] $K^{-3/2}\sum^{K}_{k=1}\left(|b_{k}|^{3}+|a_{k}|^{3}\right)=o_{p}(1)$;
\end{enumerate}
where statements (a) and (b) hold with probability $1-o(1)$ and all
probability statements are with respect to the joint distribution
of $(W,U,\varepsilon,\xi)$.
\end{assumption}
\renewcommand{\thethm}{2$\,^\prime$} 
\begin{thm}
\label{thm:main_direct}Under Assumption~\ref{ass:asymptotic_direct}
and the Gaussian rounded design, the optimal recentered IV estimator
$\hat{\theta}=(\bhat,\hat{\tau})'$ satisfies 
\[
\sqrt{K}(\hat{\theta}-\theta)\toD\mathcal{N}\left(0,\boldsymbol{h}^{*-1}\boldsymbol{v}^{*}\boldsymbol{h}^{*-1\prime}\right).
\]
\end{thm}
\begin{proof}
As in Appendix \ref{app:Theorem1}, we condition on statements (a)
and (b) of Assumption~\ref{ass:asymptotic_direct} throughout. We
start by decomposing the estimator error as 
\begin{equation}
\sqrt{K}(\hat{\theta}-\theta)=\left(K^{-1}Z'X\right)^{-1}K^{-1/2}Z'\varepsilon.\label{eq:ratio_direct}
\end{equation}
 Since $\boldsymbol{z}^{*}_{\delta}(g)_{i}=\left(z^{*}_{\delta}(g)_{i},\tilde{h}_{i}\right)$
with $z^{*}_{\delta}(g)=\left((S^{\perp})^{\dag}\right)^{1/(p+1)}W^{\perp}\tilde{g}$,
we have $Z'\varepsilon=\Psi'\tilde{g}$ for $\Psi=(b,a)$ as in Assumption~\ref{ass:asymptotic_direct}(d).

For the numerator, Lemma \ref{lem:numerator} applies with $m=2$:
its conditions hold by Assumption \ref{ass:asymptotic_direct}(a),
(d) and (e), using $\left\Vert (b_{k},a_{k})\right\Vert ^{3}\le2^{3/2}\left(|b_{k}|^{3}+|a_{k}|^{3}\right)$.
Therefore $K^{-1/2}Z'\varepsilon=K^{-1/2}\Psi'\tilde{g}\toD\mathcal{N}(0,\boldsymbol{v}^{*})$.

For the denominator, writing $x=W\left(\tfrac{1}{2}\one+{\tilde{g}}\right)$
and $h=U\left(\tfrac{1}{2}\one+{\tilde{g}}\right)$, 
\[
Z'X=\begin{bmatrix}\tilde{g}'B_{11}\left(\tfrac{1}{2}\one+{\tilde{g}}\right) & \tilde{g}'B_{12}\left(\tfrac{1}{2}\one+{\tilde{g}}\right)\\
\tilde{g}'B_{21}\left(\tfrac{1}{2}\one+{\tilde{g}}\right) & \tilde{g}'B_{22}\left(\tfrac{1}{2}\one+{\tilde{g}}\right)
\end{bmatrix}
\]
for $B_{11}=W^{\perp\prime}\left((S^{\perp})^{\dag}\right)^{1/(p+1)}W$,
$B_{12}=W^{\perp\prime}\left((S^{\perp})^{\dag}\right)^{1/(p+1)}U$,
$B_{21}=U'W$, and $B_{22}=U'U$. By Assumption \ref{ass:asymptotic_direct}(a)
and (b), 
\[
\op{\left((S^{\perp})^{\dag}\right)^{1/(p+1)}}=\{\lambda^{+}_{\min}(S^{\perp})\}^{-1/(p+1)}\le\underline{\lambda}^{-1/(p+1)}
\]
and hence $\op{B_{jl}}\le\bar{\lambda}\max\{1,\underline{\lambda}^{-1/(p+1)}\}$
for all $j,l\in\{1,2\}$, so Lemma \ref{lem:denominator} applies
to each entry of $K^{-1}Z'X$. Since $\expec{}{\tilde{g}'B_{jl}\left(\tfrac{1}{2}\one+{\tilde{g}}\right)\mid W,U}=\tr(B_{jl}\Sigma^{GR})$,
this gives $K^{-1}Z'X-\boldsymbol{h}_{K}\toP0$ and, by Assumption
\ref{ass:asymptotic_direct}(c), $K^{-1}Z'X\toP{\boldsymbol{h}}^{*}$,
which is invertible. The result follows from (\ref{eq:ratio_direct})
by the continuous mapping theorem and Slutsky's theorem.
\end{proof}

Asymptotically valid inference follows from Theorem \ref{thm:main_direct}
using the plug-in estimator $\hat{\boldsymbol{h}}^{-1}\hat{\boldsymbol{v}}\hat{\boldsymbol{h}}^{-1\prime}$
of the asymptotic variance, where $\hat{\boldsymbol{h}}=K^{-1}Z'X$
and 
\begin{align*}
\hat{\boldsymbol{v}} & =K^{-1}\hat{\Psi}'\Sigma^{GR}\hat{\Psi},\\
\hat{\Psi} & =\left(W^{\perp\prime}\left((S^{\perp})^{\dag}\right)^{1/(p+1)}\hat{\varepsilon},\,U'\hat{\varepsilon}\right),\\
\hat{\varepsilon} & =y-X\hat{\theta}.
\end{align*}
Its consistency follows from $\hat{\boldsymbol{h}}\toP{\boldsymbol{h}}^{*}$,
established in the proof of Theorem \ref{thm:main_direct}, and from
$\hat{\boldsymbol{v}}\toP{\boldsymbol{v}}^{*}$, which follows from
consistency of $\hat{\theta}$ and Assumption \ref{ass:asymptotic_direct}
by the same kind of argument as for $\hat{v}$ in Appendix \ref{app:Theorem1}.

\renewcommand{\theassumption}{\arabic{assumption}}\setcounter{assumption}{1} \renewcommand{\thethm}{\arabic{thm}}\setcounter{thm}{2}

\subsection{Budget Constraints\label{appx:budgets}}

For the optimal design, consider first the analytically simpler case
where the researcher faces a constraint that the probability of treating
each unit be exactly $q$. Proposition \ref{prop:OptimalDesign} is
then modified to 
\[
\delta^{\ast}\in\arg\max_{\delta\colon\expec{\delta}g=q\one{}{}}\tr{\left\{ \left(W\var{\delta}gW'\right)^{p/(p+1)}\right\} }.
\]
In the case of a block-diagonal $W$, this problem still splits into
block-by-block parts, so this claim of Proposition \ref{prop:DesignProperties}
extends, too.

We propose to formulate the relaxed problem as an optimization over
the set of all correlation matrices multiplied by $q(1-q)$ instead
of $1/4$. Clearly, the objective $\tr{\left(W\Sigma W'\right)}^{p/(p+1)}$
is just a rescaling relative to $\Sigma\in\mathcal{Q}_{K}$, and the
optimal covariance matrix $\Sigma^{\ast}$ is the same, rescaled by
$4q(1-q)$. The Proposition \ref{prop:Optimize_l} characterization
of the relaxed problem's solution is therefore unchanged, too. The
Gaussian rounding approximation, in turn, is modified to drawing $\xi\sim\mathcal{N}(0,\Sigma^{\ast}/\left(q(1-q)\right))$
and setting $g_{k}=\one{\left[\xi_{k}\ge\Phi^{-1}(1-q)\right]}{}$
where $\Phi(\cdot)$ is the standard normal CDF. Since $\var{}{\xi_{k}}=1$,
we have $\expec{\delta}{g_{k}}=q$ for all $k$ as required.\footnote{Note that the optimization problem could impose an additional restriction
that $\Sigma^{\ast}_{kl}\ge-\min\left\{ q^{2},(1-q)^{2}\right\} $
for $k\ne l$, which has to hold for any binary vector with marginals
$q$. We do not recommend this, however, as the Gaussian rounding
procedure already ensures this inequality starting from any $\Sigma^{\ast}\succeq0$.
Imposing extra constraints on $\Sigma^{\ast}$ would unnecessarily
limit the set of Gaussian rounding designs.}

We now consider three alternative constraints the researcher may face.
First, suppose the budget constraint takes the inequality form $\expec{\delta}{g_{k}}\le q$
for $q\le1/2$. In general, we do not know whether the analog of the
first claim of Proposition \ref{prop:DesignProperties} holds, i.e.
that all of the inequality constraints have to bind, with $\expec{\delta}{g_{k}}=q$.
This is true, however, in the relaxed problem, as for any $\Sigma\succeq0$
with $\Sigma_{kk}\le q(1-q)$ and some inequalities strict, we can
consider $\tilde{\Sigma}=\Sigma+\diag\left(q(1-q)-\Sigma_{kk}\right)\succeq\Sigma$
which yields a weakly lower worst-case variance.

Second, suppose the budget constraint is soft, in the sense that the
researcher is willing to trade off the worst-case estimation variance
with the expected cost of treating more units. For instance, one can
choose $q$ by solving
\[
\min_{q\in[0,1/2]}\left\{ \left(\max_{\delta\colon\expec{\delta}g=q\one{}{}}\tr{\left(\left(W\var{\delta}gW'\right)^{p/(p+1)}\right)}\right)^{-(p+1)/p}+\lambda qK\right\} ,
\]
where the first term is the worst-case variance of the estimator and
the second term is the expected cost of the experiment with weight
$\lambda$ in units of $N^{1/p}\sigma^{2}$. This is a straightforward
scalar optimization, and it is especially simple if the inner optimization
is replaced by its relaxed version, as the optimal shock correlation
matrix does not depend on $q$.

The final scenario is when the researcher's constraint is on the realized
number of treated units: $\frac{1}{K}\sum_{k}g_{k}\le q$. This formulation
is difficult to implement within our approach, as it does not naturally
translate to a restriction on $\Sigma$ and it is not consistent with
Gaussian rounding. In cases like this, we recommend using our original
constraint on $\expec{\delta}{g_{k}}$ but either with a tighter $q$
(such that it is very unlikely that the realization violates the constraint
on the realized number of treated units, as formalized in \citet{sun2026empirical})
or with small \emph{ad hoc} modifications (e.g., switching a randomly
chosen set of $g_{k}=1$ to $g_{k}=0$ until the constraint holds).

\subsection{Precision Constraints with Multiple Exposures\label{appx:direct_fx_constraint}}

This appendix considers multiple-exposure scenarios where the researcher
is primarily interested in one exposure effect but requires a minimal
level of precision for the second exposure effect. We derive a constrained
solution to the problem in Section \ref{subsec:Direct-Effects} and
show it is equivalent to the problem of minimizing a weighted average
of worst-case variances for the two effects.

In the two-exposure model from Section \ref{subsec:Direct-Effects},
let $\boldsymbol{x}=((x_{i})^{N}_{i=1},(h_{i})^{N}_{i=1})$ and for
any recentered instrument vector $\boldsymbol{z}(g)=(z_{\beta}(g),z_{\tau}(g))\in\mathbb{R}^{N\times2}$
let $M_{\boldsymbol{z}}=\expec{\delta}{\boldsymbol{z}(g)^{\prime}\boldsymbol{x}}$.
Finally, for $j\in\{\beta,\tau\}$, let $\mathcal{\boldsymbol{\mathcal{V}}}^{j}_{\delta,\mathcal{E}}[\boldsymbol{z}]$
denote the appropriate diagonal element of $M^{-1}_{\boldsymbol{z}}\var{\delta,\mathcal{E}}{\boldsymbol{z}(g)^{\prime}\varepsilon}M^{-1\prime}_{\boldsymbol{z}}$,
with both marginal approximate variances defined to be infinite when
$M_{\boldsymbol{z}}$ is singular. We then study, for some $\bar{s}^{2}_{\tau}$:
\begin{align}
 & \min_{\delta\in\mathcal{D},\boldsymbol{z}\in\mathcal{Z}^{(2)}_{\delta}}\max_{\mathcal{E}\in\mathcal{F}_{p}(\sigma)}\boldsymbol{\mathcal{V}}^{\beta}_{\delta,\mathcal{E}}\left[\boldsymbol{z}\right]\,\,\,\text{s.t.}\max_{\mathcal{E}\in\mathcal{F}_{p}(\sigma)}\boldsymbol{\mathcal{V}}^{\tau}_{\delta,\mathcal{E}}\left[\boldsymbol{z}\right]\le\bar{s}^{2}_{\tau}\label{eq:constrained}
\end{align}
For example, $\bar{s}^{2}_{\tau}$ could be the worst-case variance
from some benchmark design and estimator, such as a 50\% Bernoulli
RCT and unwhitened recentered IV.

An initial result shows the direct-effect and indirect-effect problems
separate:
\begin{lem}
\label{lem:effectiveIVs} Fix $\delta$ and $\mathcal{E}$. For every
$\boldsymbol{z}\in\mathcal{Z}^{(2)}_{\delta}$ with nonsingular $M_{\boldsymbol{z}}$,
define $z^{\text{\emph{eff}}}_{j}(g)=\boldsymbol{z}(g)M^{-1\prime}_{\boldsymbol{z}}e_{j}$
for $j\in\{\beta,\tau\}$, where $e_{\beta}=(1,0)^{\prime}$ and $e_{\tau}=(0,1)^{\prime}$.
Then 
\begin{align}
\expec{\delta}{(z^{\text{\emph{eff}}}_{\beta}(g))^{\prime}x}=1, & \,\,\,\expec{\delta}{(z^{\text{\emph{eff}}}_{\beta}(g))^{\prime}h}=0,\label{eq:zeff_beta}\\
\expec{\delta}{(z^{\text{\emph{eff}}}_{\tau}(g))^{\prime}x}=0, & \,\,\,\expec{\delta}{(z^{\text{\emph{eff}}}_{\tau}(g))^{\prime}h}=1.\label{eq:zeff_tau}
\end{align}
and 
\begin{align}
\mathcal{V}^{j}_{\delta,\mathcal{E}}[\boldsymbol{z}] & =\var{\delta,\mathcal{E}}{(z^{\text{\emph{eff}}}_{j})^{\prime}\varepsilon},\,\,j\in\{\beta,\tau\}\label{eq:veff}
\end{align}
Conversely, any pair of recentered scalar instruments satisfying (\ref{eq:zeff_beta})-(\ref{eq:zeff_tau})
can be stacked as $\boldsymbol{z}=(z_{\beta},z_{\tau})$ to yield
(\ref{eq:veff}). Hence the sets of marginal variances form a Cartesian
product.
\end{lem}
\begin{proof}
Equations (\ref{eq:zeff_beta})-(\ref{eq:zeff_tau}) follow from the
definition of $z^{\text{eff}}_{j}$, since $\expec{\delta}{\boldsymbol{x}^{\prime}z^{\text{eff}}_{j}}=M^{\prime}_{\boldsymbol{z}}M^{-1\prime}_{\boldsymbol{z}}e_{j}=e_{j}$.
Equation (\ref{eq:veff}) then follows from the definitions of $\mathcal{V}^{j}_{\delta,\mathcal{E}}[\boldsymbol{z}]$
and $z^{\text{eff}}_{j}$. The converse similarly follows from these
definitions. 
\end{proof}

We next use this result to extend Theorem \ref{prop:OptimalIV_direct}:
\begin{prop}
\label{prop:OptimalIV_const}Fix $\delta$ and let 
\begin{align*}
c_{\delta} & \in\arg\min_{c\in\mathbb{R}}\tr{\left(\left((W-cU)\Sigma_{\delta}(W-cU)^{\prime}\right)^{p/(p+1)}\right)},\\
d_{\delta} & \in\arg\min_{d\in\mathbb{R}}\tr{\left(\left((U-dW)\Sigma_{\delta}(U-dW)^{\prime}\right)^{p/(p+1)}\right)},\\
\tilde{x}^{\perp}_{\delta} & =(W-c_{\delta}U)(g-\expec{\delta}g),\qquad S^{\perp}_{\beta,\delta}=(W-c_{\delta}U)\Sigma_{\delta}(W-c_{\delta}U)^{\prime},\\
\tilde{h}^{\perp}_{\delta} & =(U-d_{\delta}W)(g-\expec{\delta}g),\qquad S^{\perp}_{\tau,\delta}=(U-d_{\delta}W)\Sigma_{\delta}(U-d_{\delta}W)^{\prime}.
\end{align*}
Suppose $S^{\perp}_{\beta,\delta}\neq0$, $S^{\perp}_{\tau,\delta}\neq0$,
$\tr(U\Sigma_{\delta}U')>0$, and $\tr(W\Sigma_{\delta}W')>0$. Then,
for any $p\in[1,\infty]$ and $\sigma>0$:

(a) The smallest attainable marginal worst-case approximate variances
are, for $j\in\{\beta,\tau\}$:
\begin{align*}
\min\limits_{\bm{z}\in\mathcal{Z}^{(2)}_{\delta}}\max_{\mathcal{E}\in\mathcal{F}_{p}(\sigma)}\mathcal{\boldsymbol{\mathcal{V}}}^{j}_{\delta,\mathcal{E}}\left[\bm{z}\right] & =N^{1/p}\sigma^{2}\tr{\left(\left(S^{\perp}_{j,\delta}\right)^{p/(p+1)}\right)}^{-(p+1)/p}.
\end{align*}

(b) Assume that either $p>1$ or $p=1$ and $\rank\!\left((W-c_{\delta}U)\Sigma^{1/2}_{\delta}\right)=\rank\!\left((U-d_{\delta}W)\Sigma^{1/2}_{\delta}\right)=\min\left\{ N,\rank(\Sigma_{\delta})\right\} $.
Then the instrument $\boldsymbol{z}^{\ast}_{i}(g)=(z^{*}_{\beta,\delta}(g)_{i},z^{*}_{\tau,\delta}(g)_{i})$
attains the above approximate variances, where
\begin{align*}
z^{*}_{\beta,\delta}(g) & =\left((S^{\perp}_{\beta,\delta})^{\dag}\right)^{1/(p+1)}\tilde{x}^{\perp}_{\delta},\\
z^{*}_{\tau,\delta}(g) & =\left((S^{\perp}_{\tau,\delta})^{\dag}\right)^{1/(p+1)}\tilde{h}^{\perp}_{\delta}.
\end{align*}
\end{prop}
\begin{proof}
By Lemma \ref{lem:effectiveIVs}, the problem of minimizing the marginal
approximate variance of $\hat{\beta}$ is the same as the scalar problem
of choosing a recentered $z_{\beta}$ satisfying the constraints in
(\ref{eq:zeff_beta}). Moreover the first constraint of $\expec{\delta}{z^{\prime}_{\beta}x}=1$
can always be satisfied by an appropriate instrument rescaling. Theorem
\ref{prop:OptimalIV_direct} gives the above formulas for the optimal
approximate variance and unscaled instrument for $\beta$. Interchanging
($W,x,\beta$) and $(U,h,\tau$) gives the corresponding variance
and instrument formulas for $\tau$, while Lemma \ref{lem:effectiveIVs}
implies the two marginal variance optima can be attained simultaneously
by stacking $z^{*}_{\beta,\delta}(g)_{i}$ and $z^{*}_{\tau,\delta}(g)_{i}$.
\end{proof}

These results immediately imply a solution to the precision-constrained
problem:
\begin{prop}
Suppose there is $\delta\in\mathcal{D}$ such that the conditions
in Theorem \ref{prop:OptimalIV_direct} hold and $\tr{\left((S^{\perp}_{\tau,\delta})^{p/(p+1)}\right)}\ge(N^{1/p}\sigma^{2}/\bar{s}^{2}_{\tau})^{p/(p+1)}$.
Then
\[
\delta^{*},\boldsymbol{z}^{*}\in\arg\min_{\delta\in\mathcal{D},\boldsymbol{z}\in\mathcal{Z}^{(2)}_{\delta}}\max_{\mathcal{E}\in\mathcal{F}_{p}(\sigma)}\mathcal{\boldsymbol{\mathcal{V}}}^{\beta}_{\delta,\mathcal{E}}\left[\boldsymbol{z}\right]\,\,\,\text{\emph{s.t.}}\,\max_{\mathcal{E}\in\mathcal{F}_{p}(\sigma)}\boldsymbol{\mathcal{V}}^{\tau}_{\delta,\mathcal{E}}\left[\boldsymbol{z}\right]\le\bar{s}^{2}_{\tau}
\]
is solved by 
\begin{align*}
\delta^{*} & \in\arg\max_{\delta\in\mathcal{D}}\tr{\left((S^{\perp}_{\beta,\delta})^{p/(p+1)}\right)}\,\,\,\text{\emph{s.t.}}\,\tr{\left((S^{\perp}_{\tau,\delta})^{p/(p+1)}\right)}\ge(N^{1/p}\sigma^{2}/\bar{s}^{2}_{\tau})^{p/(p+1)}.
\end{align*}
The optimal instrument is $\boldsymbol{z}^{\ast}=\boldsymbol{z}^{\ast}_{\delta^{\ast}}$
from Proposition \ref{prop:OptimalIV_const}(b), as long as, for $p=1$,
\textup{$\rank\!\left((W-c_{\delta}U)\Sigma^{1/2}_{\delta}\right)=\rank\!\left((U-d_{\delta}W)\Sigma^{1/2}_{\delta}\right)=\min\left\{ N,\rank(\Sigma_{\delta})\right\} $
}.
\end{prop}
\noindent Analogs of Proposition \ref{prop:DesignProperties_direct}
and the relaxed problem also follow.

Finally, we note that in non-edge cases the constrained design problem
has an equivalent representation as the problem of minimizing a weighted
average of the two marginal worst-case variances, allowing the researcher
to trace the frontier between the precisions of the two exposure effects.
The Lagrangian for the constrained design problem is:
\begin{align*}
 & \max_{\mathcal{E}\in\mathcal{F}_{p}(\sigma)}\boldsymbol{\mathcal{V}}^{\beta}_{\delta,\mathcal{E}}\left[\boldsymbol{z}\right]+\lambda\left(\max_{\mathcal{E}\in\mathcal{F}_{p}(\sigma)}\boldsymbol{\mathcal{V}}^{\tau}_{\delta,\mathcal{E}}\left[\boldsymbol{z}\right]-\bar{s}^{2}_{\tau}\right).
\end{align*}
For any fixed multiplier $\lambda\ge0$, the $-\lambda\bar{s}^{2}_{\tau}$
term does not affect the optimizing design and instrument. Hence the
inner problem is:
\begin{align*}
 & \min_{\delta\in\mathcal{D},\boldsymbol{z}\in\mathcal{Z}^{(2)}_{\delta}}\left\{ \max_{\mathcal{E}\in\mathcal{F}_{p}(\sigma)}\boldsymbol{\mathcal{V}}^{\beta}_{\delta,\mathcal{E}}\left[\boldsymbol{z}\right]+\lambda\max_{\mathcal{E}\in\mathcal{F}_{p}(\sigma)}\boldsymbol{\mathcal{V}}^{\tau}_{\delta,\mathcal{E}}\left[\boldsymbol{z}\right]\right\} .
\end{align*}
This means that if the constraint binds there is a $\lambda^{*}>0$
such that the solution to the constrained problem also solves the
weighted-average problem
\begin{align*}
 & \min_{\delta\in\mathcal{D},\boldsymbol{z}\in\mathcal{Z}^{(2)}_{\delta}}\left\{ \frac{1}{1+\lambda^{*}}\max_{\mathcal{E}\in\mathcal{F}_{p}(\sigma)}\boldsymbol{\mathcal{V}}^{\beta}_{\delta,\mathcal{E}}\left[\boldsymbol{z}\right]+\frac{\lambda^{*}}{1+\lambda^{*}}\max_{\mathcal{E}\in\mathcal{F}_{p}(\sigma)}\boldsymbol{\mathcal{V}}^{\tau}_{\delta,\mathcal{E}}\left[\boldsymbol{z}\right]\right\} .
\end{align*}
Conversely, a solution to the weighted-average problem $(\delta^{*},\boldsymbol{z}^{*})$
solves the constrained problem with $\bar{s}_{\tau}=\sqrt{\max_{\mathcal{E}\in\mathcal{F}_{p}(\sigma)}\boldsymbol{\mathcal{V}}^{\tau}_{\delta^{*},\mathcal{E}}\left[\boldsymbol{z}^{*}\right]}$.

\subsection{\label{appx-adaptivity}Optimal Design with an Adaptive IV}

Theorem \ref{prop:OptimalIV} chooses the recentered IV $z^{*}_{\delta}$
using only the design and the restriction $\mathcal{E}\in\mathcal{F}_{p}(\sigma)$.
This is appropriate before outcomes are observed, but it may be conservative
for estimation. Once outcomes have been collected, the researcher
may be able to learn which of several recentered IVs gives highest
precision under the realized error structure. For example, whitening
may be valuable when the errors are strongly correlated, while the
unwhitened IV is efficient under spherical errors.

In this appendix we consider a broad class of infeasible adaptive
IV selection procedures and show the Proposition \ref{prop:OptimalDesign}
design remains optimal. We further propose a feasible two-step GMM
estimator that uses both the unwhitened exposure and the Theorem \ref{prop:OptimalIV}
optimal IV as two instruments. We focus on the single-exposure case
but the results extend to multiple exposures, analogously to Section
\ref{subsec:Direct-Effects}.

\paragraph{Infeasible IV Selection.}

We first consider an infeasible version of adaptive IV selection,
where the researcher learns and uses some features of $\mathcal{E}$
ex-post without estimation error and uses them to reduce estimation
variance:
\begin{prop}
\label{prop:adaptive-general}For every $\delta\in\mathcal{D}$ with
$S_{\delta}\ne0$ and every $\mathcal{E}\in\mathcal{F}_{p}(\sigma)$,
suppose an infeasible ex-post procedure selects $z^{A}_{\delta,\mathcal{E}}\in\mathcal{Z}_{\delta}$
satisfying
\[
\V{\delta,\mathcal{E}}{z^{A}_{\delta,\mathcal{E}}}\leq\V{\delta,\mathcal{E}}{z^{*}_{\delta}}.
\]
Then for every such $\delta$ the worst-case approximate variance
is the same as in Theorem \ref{prop:OptimalIV}: 
\[
\max_{\mathcal{E}\in\mathcal{F}_{p}(\sigma)}\V{\delta,\mathcal{E}}{z^{A}_{\delta,\mathcal{E}}}=N^{1/p}\sigma^{2}\left(\tr(S^{p/(p+1)}_{\delta})\right)^{-(p+1)/p}.
\]
Consequently, allowing this ex-post adaptation leaves the ex-ante
minimax design problem unchanged: the optimal design from Proposition
\ref{prop:OptimalDesign} also solves
\[
\min_{\delta\in\mathcal{D}}\max_{\mathcal{E}\in\mathcal{F}_{p}(\sigma)}\V{\delta,\mathcal{E}}{z^{A}_{\delta,\mathcal{E}}}.
\]
\end{prop}
Intuitively, under a least-favorable error distribution, the Theorem
\ref{prop:OptimalIV} optimal IV yields the most efficient estimator;
hence, no adaptive estimator can have a smaller worst-case variance.
Conversely, since the adaptive estimator is assumed to weakly improve
on $z^{*}_{\delta}$ for all $\mathcal{E}$, its worst-case variance
cannot be larger.\footnote{Note that, when denoting the adaptive instrument by $z^{A}_{\delta,\mathcal{E}}$,
we do not require that the researcher observes the entire distribution
$\mathcal{E}$, but only that the IV choice may depend on $\mathcal{E}$
via some observed features.}

Two special cases are worth pointing out. First, even if the researcher
could learn everything about $\mathcal{E}$ ex-post, the ex-ante design
problem would not change:
\begin{cor}
\label{cor:adaptive-oracle} The Proposition \ref{prop:OptimalDesign}
design $\delta^{\ast}$ remains minimax-optimal if the researcher
uses the oracle recentered IV estimator ex-post:
\[
\delta^{\ast}\in\arg\min_{\delta\in\mathcal{D}}\max_{\mathcal{E}\in\mathcal{F}_{p}(\sigma)}\min_{z\in\mathcal{Z}_{\delta}}\V{\delta,\mathcal{E}}z.
\]
\end{cor}
\begin{proof}
Let $z^{\mathrm{or}}_{\delta,\mathcal{E}}\in\arg\min_{z\in\mathcal{Z}_{\delta}}\mathcal{V}_{\delta,\mathcal{E}}[z]$
denote the oracle recentered IV. The proof then follows by Proposition
\ref{prop:adaptive-general}, noting that $z^{\mathrm{or}}_{\delta,\mathcal{E}}$
satisfies $\V{\delta,\mathcal{E}}{z^{\text{or}}_{\delta,\mathcal{E}}}\leq\V{\delta,\mathcal{E}}{z^{*}_{\delta}}$
by construction.
\end{proof}

Second, the ex-ante design problem also does not change if the researcher
combines $z^{\ast}$ with some other recentered instrument $z^{U}$
using the (infeasible) optimal weighting matrix that minimizes the
approximate estimation variance. We focus on $z^{U}=\tilde{x}$, such
that the researcher uses the best combination of the unwhitened exposure
and the optimally whitened IV (which are generally not perfectly collinear
when $p<\infty$), but the result applies more generally.
\begin{cor}
\label{cor:adaptive-twoIV}Let $\boldsymbol{z}=\left(z^{\ast},z^{U}\right)$
for $z^{\ast}$ from Theorem \ref{prop:OptimalIV} and some recentered
IV $z^{U}$, and suppose the researcher uses as an IV the infeasible
approximate-variance minimizing linear combination:
\[
\check{z}_{\delta,\mathcal{E}}=\boldsymbol{z}\check{\lambda}_{\delta,\mathcal{E}},\qquad\text{for }\check{\lambda}_{\delta,\mathcal{E}}\in\arg\min_{\lambda\in\mathbb{R}^{2}}\V{\delta,\mathcal{E}}{\boldsymbol{z}\lambda}.
\]
Then the Proposition \ref{prop:OptimalDesign} design $\delta^{\ast}$
remains minimax-optimal:
\[
\delta^{\ast}\in\arg\min_{\delta\in\mathcal{D}}\max_{\mathcal{E}\in\mathcal{F}_{p}(\sigma)}\V{\delta,\mathcal{E}}{\check{z}_{\delta,\mathcal{E}}}.
\]
\end{cor}
\begin{proof}
By Proposition \ref{prop:adaptive-general}, noting that the candidate
set of linear combinations contains $z^{*}_{\delta}$, by setting
$\lambda=(1,0)'$. Hence $\V{\delta,\mathcal{E}}{\check{z}_{\delta,\mathcal{E}}}=\min_{\lambda\in\mathbb{R}^{2}}\V{\delta,\mathcal{E}}{\boldsymbol{z}\lambda}\leq\V{\delta,\mathcal{E}}{z^{*}_{\delta}}.$
\end{proof}

One can further show that, outside a degenerate case where $\beta$
can be learned without noise (i.e., provided $\expec{\delta}{\boldsymbol{z}'x}\in\text{Im(\ensuremath{\var{\delta,\mathcal{E}}{\boldsymbol{z}'\varepsilon}})}$),
an optimal IV combination is:
\begin{equation}
\check{\lambda}_{\delta,\mathcal{E}}\propto\var{\delta,\mathcal{E}}{\boldsymbol{z}'\varepsilon}^{\dag}\expec{\delta}{\boldsymbol{z}'x},\label{eq:optimal-IV-combination}
\end{equation}
which is the oracle population version of the IV corresponding to
the optimally-weighted GMM using moment conditions $\expec{\delta,\mathcal{E}}{\boldsymbol{z}'\left(y-\beta x\right)}=0$.

\paragraph{Feasible GMM.}

While the results above show that ex-post IV adaptivity leaves the
ex-ante problem exactly intact, the corresponding estimators are infeasible.
Focusing on the two-IV estimator of Corollary \ref{cor:adaptive-twoIV},
we now discuss feasible GMM estimation. The two-step procedure would
first estimate $A_{\boldsymbol{z}}=\frac{1}{K}\var{\delta,\mathcal{E}}{\boldsymbol{z}'\varepsilon}$
by $\hat{A}_{\boldsymbol{z}}$ where
\[
\left(\hat{A}_{\boldsymbol{z}}\right)_{jl}=\frac{1}{K}\hat{\varepsilon}'\cov{\delta}{\boldsymbol{z}_{\cdot j},\boldsymbol{z}_{\cdot l}}\hat{\varepsilon},\qquad j,l\in\left\{ 1,2\right\} 
\]
where $\boldsymbol{z}_{\cdot j}$ is $j$th column of $\boldsymbol{z}$,
the matrices $\cov{\delta}{\boldsymbol{z}_{\cdot j},\boldsymbol{z}_{\cdot l}}$
are known by the design, and $\hat{\varepsilon}$ collects residuals
from some consistent preliminary estimation of $\beta$ (such as using
$z^{\ast}$ as a single IV). One would then solve the GMM problem
with the weight matrix $\hat{A}^{\dag}_{\boldsymbol{z}}$:
\[
\hat{\beta}_{GMM}\in\arg\min_{b}\left(y-bx\right)'\boldsymbol{z}\hat{A}^{\dag}_{\boldsymbol{z}}\boldsymbol{z}'\left(y-bx\right).
\]
This corresponds to using the instrument 
\[
\check{z}_{GMM}=\boldsymbol{z}\check{\lambda}_{GMM},\qquad\check{\lambda}_{GMM}\propto\hat{A}^{\dag}_{\boldsymbol{z}}\cdot\left(\frac{1}{K}\boldsymbol{z}'x\right),
\]
paralleling (\ref{eq:optimal-IV-combination}).

Minimax regret from using $\check{z}_{GMM}$ instead of $\check{z}_{\delta,\mathcal{E}}$
should asymptotically vanish when $\op{\hat{A}^{\dag}_{\boldsymbol{z}}-A^{\dag}_{\boldsymbol{z}}}\to0$
and $\frac{1}{K}\boldsymbol{z}'x-\frac{1}{K}\expec{\delta}{\boldsymbol{z}'x}\toP0$,
given a uniform CLT and estimation error bounds across both $\delta\in\mathcal{D}$
and $\mathcal{E}\in\mathcal{F}_{p}(\sigma)$. These would require
going beyond Theorem \ref{thm:main}, which only guarantees approximations
for a single asymptotic sequence and under restrictive conditions
on both design and error distribution. We expect that uniform approximations
can be obtained and feasible GMM can be shown to have vanishing minimax
regret in a more limited sense: over a lower-complexity set of designs
and over a restricted set of error distributions containing a least-favorable
distribution for each design under consideration. We leave formalizing
these conjectures to future work.

\subsection{Error Bound on the Gaussian Rounding Implementation\label{appx:approx_bound}}
\begin{prop}
\label{prop:approx_bound} We have:
\[
\tr{\left(S^{p/(p+1)}_{GR}\right)}\ge\left(\frac{2}{\pi}\right)^{p/(p+1)}\max_{\Sigma\in\mathcal{C}_{K}}\tr{\left(\left(W\Sigma W'\right)^{p/(p+1)}\right)}.
\]
Thus, the worst-case approximate variance of the Theorem \ref{prop:OptimalIV}
optimal recentered IV estimator under the Gaussian rounding design
based on the relaxed problem's solution from Proposition \ref{prop:Optimize_l}
is at most $\pi/2$ times that of the original Proposition \ref{prop:OptimalDesign}
problem. The same result holds for the multiple-exposure extension
in Sections \ref{subsec:Direct-Effects} and \ref{subsec:Computation}.
\end{prop}

\section{Proofs of Additional Theoretical Results\label{appx:Proofs-additional}}

\subsection{A Useful Lemma}

For $n\in\mathbb{N}$, denote $P_{n}=\one{}{}_{n}\one{'_{n}}/n$.
\begin{lem}
\label{lem:equicor}Consider $n\in\mathbb{N}$ and $\alpha,\beta,\gamma\in\mathbb{R}$
such that $\alpha,\beta\ge0$ and if $\gamma<0$ then $\alpha,\beta>0$.
Then the matrix $A=\alpha(I_{n}-P_{n})+\beta P_{n}$ is positive semi-definite
and we have:
\[
A^{\gamma}=\alpha^{\gamma}(I_{n}-P_{n})+\beta^{\gamma}P_{n}.
\]
Moreover,
\[
\tr{\left(A^{\gamma}\right)}=(n-1)\alpha^{\gamma}+\beta^{\gamma}
\]
and, when $\gamma\ge1$,
\[
\left\Vert A\right\Vert _{\gamma}=\left((n-1)\alpha^{\gamma}+\beta^{\gamma}\right)^{1/\gamma}.
\]
\end{lem}
\begin{proof}
The Lemma statement provides a spectral decomposition for matrix $A$,
since $P_{n}$ is a rank 1 orthogonal projector onto the span of $\one{}{}_{n}$
and $I_{n}-P_{n}$ is a rank $n-1$ orthogonal projector onto the
space orthogonal to $\one{}{}_{n}$. This means that $A$ has $n-1$
eigenvalues that are $\alpha$ and one eigenvalue that is $\beta$,
and the rest of the claims come from spectral properties of the trace
and the Schatten-$\gamma$ norm.
\end{proof}

\subsection{Proposition \ref{prop:mse_nospillovers}}

With no spillovers, the design problem is simply:
\begin{align*}
\max_{\delta} & \tr{\left(\expec{\delta}{gg^{\prime}}^{\frac{p}{p+1}}\right)}.
\end{align*}
By symmetry and concavity of $A\mapsto\tr(A^{p/(p+1)})$, there is
an exchangeable optimizer. Any such design can be represented by first
drawing $K_{1}=\mathbf{1}^{\prime}g$ from some distribution and then
choosing $K_{1}$ treated units at random. For such a design, $\expec{\delta}{g^{2}_{k}}=\expec{\delta}{g_{k}}=\expec{\delta}{\frac{K_{1}}{N}}$
and, for $k\ne l$, 
\[
\expec{\delta}{g_{k}g_{\ell}}=\expec{\delta}{\expec{\delta}{g_{k}g_{\ell}\mid K_{1}}}=\expec{\delta}{\frac{K_{1}(K_{1}-1)}{N(N-1)}}.
\]
Simple algebra yields $\expec{\delta}{gg'}=\lambda_{1,\delta}P_{N}+\lambda_{2,\delta}\left(I_{N}-P_{N}\right)$
for 
\[
\lambda_{1,\delta}=\frac{\expec{\delta}{K^{2}_{1}}}{N},\qquad\lambda_{2,\delta}=\frac{\expec{\delta}{K_{1}(N-K_{1})}}{N(N-1)}.
\]
By Lemma \ref{lem:equicor}, the objective is: 
\begin{align*}
\max_{\delta} & \lambda^{p/(p+1)}_{1,\delta}+(N-1)\lambda^{p/(p+1)}_{2,\delta}.
\end{align*}
Now fix $m=\expec{\delta}{K_{1}}\in[0,N]$. We have $\lambda_{1,\delta}\ge\lambda_{2,\delta}$
since $\lambda_{1,\delta}-\lambda_{2,\delta}=\expec{\delta}{K_{1}(K_{1}-1)}/(N-1)\ge0$,
with equality only if $K_{1}\in\left\{ 0,1\right\} $ a.s., i.e. $\expec{\delta}{K^{2}_{1}}=m$.
Moreover, the derivative of the objective with respect to $\expec{\delta}{K^{2}_{1}}$
is 
\[
\frac{p}{p+1}\cdot\left(\frac{1}{N}\lambda^{-1/(p+1)}_{1,\delta}-\frac{1}{N}\lambda^{-1/(p+1)}_{2,\delta}\right)\le0,
\]
with strict inequality except at a single point $\expec{\delta}{K^{2}_{1}}=m$.
Thus, the objective is strictly decreasing in $\expec{\delta}{K^{2}_{1}}$.
Hence, for fixed $m$, we minimize $\expec{\delta}{K^{2}_{1}}$. That
is achieved by a distribution for $K_{1}$ that has support only on
the two adjacent integers around $m$. Specifically, for $k=\left\lfloor m\right\rfloor $
and $\theta=m-k\in[0,1)$, it should place probability $1-\theta$
on $k$ and probability $\theta$ on $k+1$, such that
\begin{align*}
\expec{\delta}{K^{2}_{1}}=v(m) & =(1-\theta)k^{2}+\theta(k+1)^{2}=k^{2}+(2k+1)(m-k).
\end{align*}
The optimal $m$ is given by maximizing the objective:
\begin{align*}
m^{*} & \in\arg\max_{m\in[0,N]}\left[\left(\frac{v(m)}{N}\right)^{p/(p+1)}+(N-1)\left(\frac{Nm-v(m)}{N(N-1)}\right)^{p/(p+1)}\right],
\end{align*}
which yields the desired solution.

\subsection{Proposition \ref{prop:mse_recentered}}

It is straightforward to show (by the same steps as above) that for
fixed $\delta$ the optimal IV is
\begin{align*}
z^{*}_{\delta}(g) & \propto(A^{\dag}_{\delta})^{1/(p+1)}M_{R}Wg
\end{align*}
where
\begin{align*}
A_{\delta} & =M_{R}W\expec{\delta}{gg^{\prime}}W^{\prime}M_{R}
\end{align*}
and that optimal designs using this instrument are given by
\begin{align*}
\delta^{*} & \in\arg\max_{\delta}\tr(A^{p/(p+1)}_{\delta}).
\end{align*}
We now show it is without loss in this problem to restrict to designs
with equal marginals. For any $\delta$, let $\delta^{c}$ be the
complement design and let $\delta^{sym}=\frac{1}{2}\delta+\frac{1}{2}\delta^{c}$.
Since $x^{c}=W(\mathbf{1}-g)=W\mathbf{1}-x$ and $W\mathbf{1}\in\text{Col}(R)$,
we have 
\begin{align*}
M_{R}x^{c} & =M_{R}(W\mathbf{1}-x)=-M_{R}x.
\end{align*}
Therefore
\begin{align*}
A_{\delta^{c}} & =M_{R}\expec{\delta^{c}}{xx^{\prime}}M_{R}\\
 & =\expec{\delta}{(M_{R}x^{c})(M_{R}x^{c})^{\prime}}\\
 & =\expec{\delta}{(M_{R}x)(M_{R}x)^{\prime}}\\
 & =A_{\delta},
\end{align*}
and also $A_{\delta^{sym}}=\frac{1}{2}A_{\delta}+\frac{1}{2}A_{\delta^{c}}=A_{\delta}$.
Hence restricting to symmetrized designs is without loss, and they
all have equal marginals. 

Finally, note that under any equal-marginal design,
\begin{align*}
M_{R}W(g-\expec{\delta}g) & =M_{R}W(g-\frac{1}{2}\mathbf{1})=M_{R}Wg.
\end{align*}
Therefore the optimal instrument can be written 
\begin{align*}
z^{*}_{\delta}(g) & \propto(A^{\dag}_{\delta})^{1/(p+1)}M_{R}W(g-\expec{\delta}g)
\end{align*}
which is recentered. Therefore, it is without loss to restrict attention
to recentered instruments.

\subsection{Proposition \ref{prop:Rank1-designs}}

When $p=\infty$, the problem in Proposition \ref{prop:OptimalDesign}
becomes 
\begin{align*}
\max_{\delta}\tr(S_{\delta})= & \max_{\delta}\tr(W\var{\delta}gW^{\prime}).
\end{align*}
Moreover, by Proposition \ref{prop:DesignProperties} we may restrict
attention to designs with $\expec{\delta}g=\frac{1}{2}\mathbf{1}$.
Define $h=2g-1\in\left\{ -1,1\right\} ^{K}$ where $\expec{\delta}h=0$
and $\var{\delta}g=\frac{1}{4}\expec{\delta}{hh'}$. Hence,
\begin{align*}
\text{tr}(S_{\delta}) & =\tr(W\var{\delta}gW^{\prime})=\tr(W^{\prime}W\var{\delta}g)=\frac{1}{4}\expec{\delta}{h'W'Wh}.
\end{align*}
Choose $h^{*}$ to maximize $h^{\prime}W^{\prime}Wh$ over $\left\{ -1,1\right\} ^{K}$
and let $\delta^{*}$ put equal probability on $g^{*}=\frac{1}{2}(\mathbf{1}+h^{*})$
and $1-g^{*}=\frac{1}{2}(\mathbf{1}-h^{*})$. Here $\expec{\delta^{\ast}}g=\frac{1}{2}\mathbf{1}$
and 
\begin{align*}
\text{tr}(S_{\delta^{*}}) & =\frac{1}{4}h^{*\prime}W^{\prime}Wh^{*}\ge\frac{1}{4}\expec{\delta}{h'W'Wh}=\tr(S_{\delta})
\end{align*}
for every design $\delta$ with $\expec{\delta}g=\frac{1}{2}\one{}{}$.
Therefore $\delta^{*}$ is optimal. Finally, note that
\begin{align*}
\var{\delta^{\ast}}g & =\left(g^{*}-\frac{1}{2}\mathbf{1}\right)\left(g^{*}-\frac{1}{2}\mathbf{1}\right)^{\prime}=\frac{1}{4}h^{*}h^{*\prime}
\end{align*}
has rank 1. Moreover, when $W\ge0$, $W^{\prime}W\ge0$ so $h^{\prime}W^{\prime}Wh$
is maximized by $h=\pm\mathbf{1}$; i.e., setting all $g_{k}=1$ with
probability $0.5$ and otherwise setting all $g_{k}=0$.

\subsection{Proposition \ref{prop:Symmetry}}

By Proposition \ref{prop:DesignProperties}, we may restrict attention
to designs with $\expec{\delta}{g_{k}}=0.5$. For any $(P,Q)\in\mathcal{G}$,
let $\delta^{(P,Q)}$ be the distribution of $Qg$ for $g\sim\delta$.
Since $Q$ is a relabeling of shocks, each $\delta^{(P,Q)}$ also
satisfies $\expec{\delta^{(P,Q)}}g=\frac{1}{2}\mathbf{1}$. We have
$\Sigma_{\delta^{(P,Q)}}=Q\Sigma_{\delta}Q^{\prime}$. Moreover, since
$Q$ is a relabeling matrix, $Q'Q=I$ such that $PWQ^{\prime}=W$
implies $PW=WQ$. Thus,
\begin{align*}
S_{\delta^{(P,Q)}} & =W\Sigma_{\delta^{(P,Q)}}W^{\prime}=WQ\Sigma_{\delta}Q^{\prime}W^{\prime}=PW\Sigma_{\delta}W^{\prime}P^{\prime}=PS_{\delta}P^{\prime}.
\end{align*}
Since $P$ is an orthogonal matrix, $\tr{\left(S^{p/(p+1)}_{\delta^{(P,Q)}}\right)}=\tr{\left(S^{p/(p+1)}_{\delta}\right)}$.
Now average $\delta$ over the transformations in $\mathcal{G}$,
defining
\begin{align*}
\bar{\delta} & =\left|\mathcal{G}\right|^{-1}\sum_{(P,Q)\in\mathcal{G}}\delta^{(P,Q)}.
\end{align*}
Because all designs in this average have the same mean, $\Sigma_{\bar{\delta}}=\left|\mathcal{G}\right|^{-1}\sum_{(P,Q)\in\mathcal{G}}Q\Sigma_{\delta}Q^{\prime}$
and $S_{\bar{\delta}}=\left|\mathcal{G}\right|^{-1}\sum_{(P,Q)\in\mathcal{G}}PS_{\delta}P^{\prime}$.
And since $A\mapsto\text{tr}(A^{p/(p+1)})$ is weakly concave on the
positive semidefinite cone for $p\in[1,\infty]$,
\[
\tr{\left(S^{p/(p+1)}_{\bar{\delta}}\right)}\ge\left|\mathcal{G}\right|^{-1}\sum_{(P,Q)\in\mathcal{G}}\tr{\left(\left(PS_{\delta}P^{\prime}\right)^{p/(p+1)}\right)}=\tr{\left(S^{p/(p+1)}_{\delta}\right)}.
\]
Thus averaging over the symmetry group weakly improves the Proposition
\ref{prop:OptimalDesign} objective while producing a $\mathcal{G}$-invariant
design. Applying this argument to an optimal $\delta^{*}$ proves
the main result.

It remains only to justify the final equivalence. For the $\mathcal{G}$-invariant
optimal design $\delta^{*}$ and for $a\in\{0,1\}^{K}$ write $\delta^{*}(a)=\Pr{\delta^{\ast}}{g=a}$.
For every $(P,Q)\in\mathcal{G}$, if $g\sim\delta^{*}$ then $Qg\overset{d}{=}g$
(where $\overset{d}{=}$ stands for equal in distribution). Take any
two assignments $a,b\in\{0,1\}^{K}$ in the same $\mathcal{G}$-orbit.
Then there exists $(P,Q)\in\mathcal{G}$ such that $b=Qa$. By $\mathcal{G}$-invariance,
\[
\delta^{*}(b)=\Pr{\delta^{\ast}}{g=b}=\Pr{\delta^{\ast}}{Qg=b}=\Pr{\delta^{\ast}}{g=Q^{-1}b}=\Pr{\delta^{\ast}}{g=a}=\delta^{*}(a).
\]
Hence $\delta^{*}$ is constant on each $\mathcal{G}$-orbit, i.e.
it assigns equal probability to assignments in the same orbit.

Conversely, if a design $\delta$ is constant on each $\mathcal{G}$-orbit,
then for any $(P,Q)\in\mathcal{G}$ and any $A\subseteq\{0,1\}^{K}$,
\[
\Pr{\delta}{Qg\in A}=\sum_{a\in A}\Pr{\delta}{Qg=a}=\sum_{a\in A}\Pr{\delta}{g=Q^{-1}a}=\sum_{a\in A}\delta(Q^{-1}a)=\sum_{a\in A}\delta(a)=\Pr{\delta}{g\in A},
\]
because $Q^{-1}a$ and $a$ are in the same $\mathcal{G}$-orbit.
Thus $Qg\overset{d}{=}g$, proving equivalence.

\subsection{Propositions \ref{prop:optimal_iv_wcovs}, \ref{prop:optimal_design_wcovs}, and \ref{prop:design_properties_wcovs}}

The following facts about the projection matrix are useful: $M_{R}$
is symmetric and idempotent, $M_{R}R=0$, $\left\Vert M_{R}\right\Vert _{\infty}\le1$,
and for a vector $v$, $M_{R}v=v$ if and only if $R'v=0$.

For Proposition \ref{prop:optimal_iv_wcovs}, we first show that any
$z\in\mathcal{Z}_{\delta}$ with $\max_{\mathcal{E}\in\mathcal{F}^{R}_{p}(\sigma)}\mathcal{V}_{\delta,\mathcal{E}}\left[z\right]<\infty$
has to satisfy $R'z=0$ a.s. Indeed, for $\gamma\in\mathbb{R}^{L}$
let $\mathcal{E}_{\gamma}\in\mathcal{F}^{R}_{p}(\sigma)$ put probability
one on $\varepsilon=R\gamma$, such that
\[
\mathcal{V}_{\delta,\mathcal{\mathcal{E}_{\gamma}}}\left[z\right]\propto\var{\delta}{\gamma'R'z}=\gamma'\var{\delta}{R'z}\gamma.
\]
If $\var{\delta}{R'z}\ne0$, we can pick $\gamma$ to make $\mathcal{V}_{\delta,\mathcal{E}_{\gamma}}\left[z\right]$
arbitrarily large, such that $\sup\limits_{\mathcal{E}\in\mathcal{F}^{R}_{p}(\sigma)}\mathcal{V}_{\delta,\mathcal{E}}\left[z\right]\ge\sup\limits_{\gamma}\mathcal{V}_{\delta,\mathcal{\mathcal{E}_{\gamma}}}\left[z\right]=\infty$.
Thus, $\var{\delta}{R'z}=0$ and, since $\expec{\delta}{R'z}=0$,
we have $R'z=0$ a.s.

Thus, there is a solution to $\inf\limits_{z\in\mathcal{Z}_{\delta}}\max\limits_{\mathcal{E}\in\mathcal{F}^{R}_{p}(\sigma)}\V{\delta,\mathcal{E}}z$
that belongs to $\mathcal{Z}^{R}_{\delta}=\left\{ z\in\mathcal{Z}_{\delta}\colon R'z=0\text{ a.s.}\right\} $,
so we will now look for that solution. Let $W_{R}=M_{R}W$, $x_{R}=W_{R}g$,
$\tilde{x}_{R}=W_{R}\tilde{g}$, and $q=p/(p-1)$, and note that for
$z\in\mathcal{Z}^{R}_{\delta}$ and $\varepsilon=R\gamma+\tilde{\varepsilon}$
we have $z'\varepsilon=z'\tilde{\varepsilon}$. Moreover, using $z=M_{R}z$
and $\expec{\delta}z=0$ for the denominator, 
\[
\mathcal{V}_{\delta,\mathcal{E}}\left[z\right]=\frac{\var{\delta,\mathcal{E}}{z'\varepsilon}}{\expec{\delta}{z'Wg}^{2}}=\frac{\var{\delta,\mathcal{E}}{z'\tilde{\varepsilon}}}{\expec{\delta}{z'W_{R}\tilde{g}}^{2}}=\frac{\var{\delta,\mathcal{E}}{z'\tilde{\varepsilon}}}{\expec{\delta}{z'\tilde{x}_{R}}^{2}}.
\]
Thus, recalling that $Q_{z}=\expec{\delta}{zz'}$, analogously to
(\ref{eq:minimax_var}),
\begin{equation}
\max_{\mathcal{E}\in\mathcal{F}^{R}_{p}(\sigma)}\mathcal{V}_{\delta,\mathcal{E}}\left[z\right]=\max_{\tilde{\mathcal{E}}\in\mathcal{F}_{p}(\sigma)}\frac{\var{\delta,\tilde{\mathcal{E}}}{z'\tilde{\varepsilon}}}{\expec{\delta}{z'\tilde{x}_{R}}^{2}}=N^{1/p}\sigma^{2}\frac{\left\Vert Q_{z}\right\Vert _{q}}{\expec{\delta}{z'\tilde{x}_{R}}^{2}}.\label{eq:obj_R}
\end{equation}
Our goal is to find $z\in\mathcal{Z}^{R}_{\delta}$ that minimizes
the right-hand side of (\ref{eq:obj_R}). Minimizing that function
over the larger set $\mathcal{Z}_{\delta}$ is the Theorem \ref{prop:OptimalIV}
optimization problem with the exposure matrix $W_{R}$ instead of
$W$, which has the desired solution $z^{\ast}_{\delta,R}$ and the
worst-case variance given in the proposition. It is therefore the
solution over $\mathcal{Z}^{R}_{\delta}$ as long as $z^{\ast}_{\delta,R}\in\mathcal{Z}^{R}_{\delta}$.
To see this, note that $S_{\delta,R}=M_{R}W\Sigma_{\delta}W'M_{R}=M_{R}S_{\delta,R}M_{R}$
and therefore $\col\left(S_{\delta,R}\right)\subseteq\col(M_{R})$.
Moreover, $\left(S^{\dag}_{\delta,R}\right)^{1/(p+1)}$ has the same
column space as $S_{\delta,R}$ (including for $p=\infty$ under our
convention that $\left(S^{\dag}_{\delta,R}\right)^{0}=S_{\delta,R}S^{\dag}_{\delta,R}$).
It follows that, for every $g$, $z^{\ast}_{\delta,R}=\left(S^{\dag}_{\delta,R}\right)^{1/(p+1)}\tilde{x}_{R}\in\col(M_{R})$
and therefore $M_{R}z^{\ast}_{\delta,R}=z^{\ast}_{\delta,R}$ a.s.
or, equivalently, $R'z^{\ast}_{\delta,R}=0$ a.s., completing the
proof.

Similarly, applying Proposition \ref{prop:OptimalDesign} to the same
residualized problem proves Proposition \ref{prop:optimal_design_wcovs},
while Proposition \ref{prop:design_properties_wcovs} follows from
Proposition \ref{prop:DesignProperties}.

\subsection{Proposition \ref{prop:one_treated_neighbor}}

Let $u=Pr_{\delta}(x_{i}=0)=a+\frac{b}{n}$, which does not depend
on $i$ by exchangeability. Then:
\begin{align*}
\expec{\delta}{x_{i}} & =1-u\\
\var{\delta}{x_{i}} & =u(1-u)\\
\cov{_{\delta}}{x_{i},x_{j}} & =a-u^{2},\qquad i\ne j
\end{align*}
Thus,
\[
S_{\delta}=\left(u-a\right)\left(I_{n}-P_{n}\right)+\left(nu(1-u)-(n-1)(u-a)\right)P_{n}
\]
and, by Lemma \ref{lem:equicor}, 
\begin{align*}
\tr{\left(S^{p/(p+1)}_{\delta}\right)} & =(n-1)(u-a)^{p/(p+1)}+\left(nu(1-u)-(n-1)(u-a)\right)^{p/(p+1)}.
\end{align*}
Letting $\lambda=u-a$, we can rewrite the objective function as 
\begin{align*}
\Phi_{u}(\lambda) & =(n-1)\lambda^{p/(p+1)}+(nu(1-u)-(n-1)\lambda)^{p/(p+1)}.
\end{align*}
Since $p/(p+1)\in(0,1)$, $\Phi_{u}(\cdot)$ is concave and from the
first order condition is maximized at $\lambda=u(1-u)$. But feasibility
requires 
\begin{align*}
1-a-b & =1-u-(n-1)\lambda\ge0\\
\implies\lambda & \le\frac{1-u}{n-1}.
\end{align*}
Hence for fixed $u$ the maximizer is
\begin{align*}
\lambda^{*}(u) & =\min\left\{ u(1-u),\frac{1-u}{n-1}\right\} .
\end{align*}
Consider first the case $u\le1/(n-1)$. Then the $u(1-u)$ solution
is feasible and
\begin{align*}
\tr(S^{p/(p+1)}_{\delta}) & =n(u(1-u))^{p/(p+1)},
\end{align*}
which is increasing in $u$ on $[0,1/(n-1)]$. Hence the optimum over
this region is attained on the boundary $u=1/(n-1)$. Thus, without
loss we can restrict attention to the other case, $u\ge1/(n-1)$ and
$\lambda=(1-u)/(n-1)$. In this case, $b^{*}=1-a^{*}$ such that $Pr_{\delta^{*}}(T\ge2)=0$.
Moreover, $u=a+\frac{1-a}{n}$ and the objective simplifies to 
\[
\tr(S^{p/(p+1)}_{\delta})=(n-1)\left(\frac{1-a}{n}\right)^{p/(p+1)}+\left(\frac{a(1-a)(n-1)^{2}}{n}\right)^{p/(p+1)},
\]
with the constraints $a+(1-a)/n\ge1/(n-1)$, or equivalently $a\ge1/(n-1)^{2}$.
This objective is concave in $a$. Taking first-order conditions and
simplifying, we have
\begin{align}
1-2a^{*} & =(a^{*})^{1/(p+1)}(n-1)^{-(p-1)/(p+1)}.\label{eq:astar}
\end{align}
The left hand side is strictly decreasing and positive on $(0,0.5)$
while the right-hand side is weakly increasing and always positive,
so the $a^{*}\in(0,0.5)$ satisfying this is unique. Moreover, this
solution satisfies the feasibility constraint: at $a=1/(n-1)^{2}$,
the right-hand side equals $1/(n-1)\le1/2$ while $1-2/(n-1)^{2}\ge1/2$
for $n\ge3$; thus, $a^{\ast}\ge1/(n-1)^{2}$.

We now show the remaining claims. We prove $\expec{\delta}{g_{k}}\ne1/2$
by contradiction. If $\expec{\delta}{g_{k}}=1/2$, $\expec{\delta}T=n/2>1$
by exchangeability, while we have shown that $\expec{\delta}T=1-a^{\ast}\le1$.

To prove (a), note that when $n=3$ the above first-order condition
becomes
\begin{align*}
1-2a^{*} & =(a^{*})^{1/(p+1)}2^{-(p-1)/(p+1)},
\end{align*}
which is uniquely solved by $a^{*}=0.25$. Similarly, note the first-order
condition at $p=1$ is
\begin{align*}
1-2a^{*} & =\sqrt{a^{*}}
\end{align*}
which is uniquely solved by $a^{*}=0.25$. For (b), fix $n$ and let
$p\rightarrow\infty$, denoting the corresponding solution $a^{\ast}(p)$;
since we know 
\begin{align*}
1-2a^{*}(p) & \le(n-1)^{-(p-1)/(p+1)}\\
\implies a^{*}(p) & \ge\frac{1-(n-1)^{-(p-1)/(p+1)}}{2}\rightarrow\frac{n-2}{2(n-1)}>0,
\end{align*}
we have by passing through the limit:
\begin{align*}
1-2a^{*}(\infty) & =\lim_{p\rightarrow\infty}(a^{*}(p))^{1/(p+1)}(n-1)^{-(p-1)/(p+1)}=(n-1)^{-1}
\end{align*}
Hence $a^{*}(\infty)=(n-2)/(2(n-1))$; (c) similarly follows taking
limits of $n\rightarrow\infty$ with $p>1$.

\subsection{Proposition \ref{prop:adaptive-general}}

Fix a design $\delta$ with $S_{\delta}\neq0$. Since $\max_{\mathcal{E}\in\mathcal{F}_{p}(\sigma)}\V{\delta,\mathcal{E}}{z^{*}_{\delta}}=N^{1/p}\sigma^{2}\left(\tr(S^{p/(p+1)}_{\delta})\right)^{-(p+1)/p}$
by Theorem \ref{prop:OptimalIV} we immediately have
\begin{align*}
\max_{\mathcal{E}\in\mathcal{F}_{p}(\sigma)}\V{\delta,\mathcal{E}}{z^{A}_{\delta,\mathcal{E}}} & \leq N^{1/p}\sigma^{2}\left(\tr(S^{p/(p+1)}_{\delta})\right)^{-(p+1)/p}.
\end{align*}
For the reverse inequality, let $\mathcal{E}^{*}$ have the least-favorable
second-moment matrix $\Omega^{*}_{\varepsilon}\propto S^{1/(p+1)}_{\delta}$
constructed in (\ref{eq:omega_star}) at the end of Appendix \ref{appx:Proof-Thm1}.
Then, since $\tilde{x}_{\delta}\in\text{Im}(S_{\delta})=\text{Im}(\Omega^{*}_{\varepsilon})$
for every $g$ in the support of $\delta$, we have for any $z\in\mathcal{Z}_{\delta}$:
\begin{align*}
\expec{\delta}{z^{\prime}x}^{2} & =\expec{\delta}{z^{\prime}\tilde{x}_{\delta}}^{2}\\
 & =\expec{\delta}{z'\Omega^{*}_{\varepsilon}\Omega^{*\dagger}_{\varepsilon}\tilde{x}_{\delta}}^{2}\\
 & \leq\expec{\delta}{z'\Omega^{*}_{\varepsilon}z}\expec{\delta}{\tilde{x}_{\delta}'\Omega^{*\dagger}_{\varepsilon}\tilde{x}_{\delta}},
\end{align*}
where the last line follows by Cauchy-Schwarz. This implies, writing
$\mathcal{V}_{\delta,\mathcal{E}^{*}}[z]=\expec{\delta}{z'\Omega^{*}_{\varepsilon}z}/\expec{\delta}{z^{\prime}x}^{2}$
as the approximate variance of the recentered IV estimator using $z$
when the error distribution has second moment matrix $\Omega^{\ast}_{\varepsilon}$
(and $\expec{\delta}{z^{\prime}x}\neq0$, with $\mathcal{V}_{\delta,\mathcal{E}^{*}}[z]=\infty$
otherwise), that we have:
\begin{align*}
\mathcal{V}_{\delta,\mathcal{E}^{*}}[z] & \ge\expec{\delta}{\tilde{x}^{\prime}_{\delta}\Omega^{*\dagger}_{\varepsilon}\tilde{x}_{\delta}}^{-1}=\left(\tr\left(\Omega^{*\dagger}_{\varepsilon}S_{\delta}\right)\right)^{-1},
\end{align*}
where the final equality comes from plugging in the equation for $\Omega^{\ast}_{\varepsilon}$
from (\ref{eq:omega_star}) and using $\expec{\delta}{\tilde{x}_{\delta}\tilde{x}^{\prime}_{\delta}}=S_{\delta}$.
This implies
\[
\V{\delta,\mathcal{E^{*}}}{z^{A}_{\delta,\mathcal{E}^{*}}}\geq N^{1/p}\sigma^{2}\left(\tr(S^{p/(p+1)}_{\delta})\right)^{-(p+1)/p}.
\]
Hence $\max_{\mathcal{E}\in\mathcal{F}_{p}(\sigma)}\V{\delta,\mathcal{E}}{z^{A}_{\delta,\mathcal{E}}}\ge N^{1/p}\sigma^{2}\left(\tr(S^{p/(p+1)}_{\delta})\right)^{-(p+1)/p}$,
completing the proof.

\subsection{Proposition \ref{prop:approx_bound}}

We prove the more general result with multiple exposures: let $J(\Sigma)=\min_{c\in\mathbb{R}}\tr{\left(((W-cU)\Sigma(W-cU)^{\prime})^{p/(p+1)}\right)}$
and let
\begin{align*}
\Sigma^{*} & \in\arg\max_{\Sigma\in\mathcal{Q}_{K}}J(\Sigma),
\end{align*}
where $\mathcal{Q}_{K}$ and $\mathcal{C}_{K}$ are as before; we
again implement $\Sigma^{*}$ by Gaussian rounding with $\Sigma^{\mathrm{GR}}$
denoting the variance of $g^{\mathrm{GR}}-\frac{1}{2}\mathbf{1}$.
We prove that 
\begin{align*}
J(\Sigma^{GR}) & \ge\left(\frac{2}{\pi}\right)^{p/(p+1)}\max_{\Sigma\in\mathcal{C}_{K}}J(\Sigma).
\end{align*}

Let $R=4\Sigma^{*}$ be the relaxed correlation matrix. Under Gaussian
rounding, we have $\Sigma^{GR}=\frac{1}{2\pi}\arcsin(R)$ where $\arcsin(\cdot)$
is applied entrywise. By the power series
\begin{align*}
\arcsin(x) & =x+\sum^{\infty}_{m=1}a_{m}x^{2m+1},\,\,a_{m}\ge0
\end{align*}
we have entrywise:
\begin{align*}
\arcsin(R)-R & =\sum^{\infty}_{m=1}a_{m}R^{\circ(2m+1)}
\end{align*}
where each Hadamard power $R^{\circ(2m+1)}$ is positive semidefinite
by the Schur product theorem (the Hadamard product of two positive
semidefinite matrices is also positive semidefinite). Hence $\arcsin(R)-R\succeq0$
and
\begin{align*}
\Sigma_{GR} & \succeq\frac{1}{2\pi}R=\frac{2}{\pi}\Sigma^{*}.
\end{align*}
For any $c$, this gives:
\begin{align*}
(W-cU)\Sigma_{GR}(W-cU)^{\prime} & \succeq\frac{2}{\pi}(W-cU)\Sigma^{*}(W-cU)^{\prime}.
\end{align*}
Since $p/(p+1)\in(0,1]$, the map $A\mapsto\tr(A^{p/(p+1)})$ is monotone
on the positive semidefinite cone. Hence, for any fixed $c$:
\begin{align*}
\tr{\left(((W-cU)\Sigma_{GR}(W-cU)^{\prime})^{p/(p+1)}\right)} & \ge\left(\frac{2}{\pi}\right)^{p/(p+1)}\tr{\left(((W-cU)\Sigma^{*}(W-cU)^{\prime})^{p/(p+1)}\right)}.
\end{align*}
Because this inequality holds pointwise, we have by minimizing $c$
on both sides:
\begin{align*}
J(\Sigma^{GR}) & \ge\left(\frac{2}{\pi}\right)^{p/(p+1)}J(\Sigma^{*})\ge\left(\frac{2}{\pi}\right)^{p/(p+1)}\max_{\Sigma\in\mathcal{C}_{K}}J(\Sigma),
\end{align*}
and thus
\[
N^{1/p}\sigma^{2}J(\Sigma^{GR})^{-(p+1)/p}\le\frac{\pi}{2}\cdot N^{1/p}\sigma^{2}\left(\max_{\Sigma\in\mathcal{C}_{K}}J(\Sigma)\right)^{-(p+1)/p}.
\]
Here by Theorem \ref{prop:OptimalIV_direct} the left-hand side is
the worst-case approximate variance corresponding to the Gaussian
rounding design while the right-hand side is $\pi/2$ times the lowest
possible worst-case approximate variance.

\section{Additional Empirical Results\label{sec:appx-empirics}}

This appendix summarizes auxiliary results for the three applications
in Section \ref{sec:Applications}.

\subsection{Cai et al. (2015)\label{subsec:appx-more-Cai}}

Figure \ref{fig:cai_bipartite_draws} illustrates the optimal designs
by reporting one example draw of the vector of assignments $g_{k}$
and the (unwhitened) recentered exposures $\tilde{x}_{i}$ for the
bipartite graph of the Figure \ref{fig:Applications}(a) village.
The increased correlation of assignment of connected intervention
units can be seen as clustered black and white circles on the left
side of the corresponding panels. Higher assignment correlation induces
more variation in the outcome units' shares of treated friends (around
the unit-specific expectation), which shows as brighter colors on
the right side for our proposed designs relative to the Bernoulli
RCT.

Figure \ref{fig:cai_exposure_heatmaps}(a) instead averages draws
across the randomization distribution, showing that the variances
of $\tilde{x}_{i}$ are higher with higher-$p$ optimal designs (higher
signal) although covariances are higher too (lower isotropy). Panel
(b) illustrates the role of whitening by reporting variances and covariances
of the whitened exposures, defined as the population first-stage fitted
values from the regression of $\tilde{x}_{i}$ on $z^{\ast}_{i}$,
i.e. $\hat{x}_{i}=\frac{\sum_{i}\expec{\delta}{\tilde{x}_{i}z^{\ast}_{i}}}{\sum_{i}\expec{\delta}{z^{\ast2}_{i}}}z^{\ast}_{i}$.
For $p=2$ and especially $p=1$, whitening partially equalizes the
variances of $\hat{x}_{i}$ and reduces their covariances.

Figure \ref{fig:sd_exposures}(a) confirms these patterns hold more
generally by plotting the distribution of exposure signal across units.
Panel (a) reports that the distribution of exposure variation, measured
by $\sd{\delta}{\tilde{x}_{i}}$, shifts to the right when moving
from Bernoulli RCT to $p=1$ and $p=2$ without whitening and to $p=\infty$,
meaning that there is more experimental variation in exposures. Whitening,
in turn, partially equalizes exposure signal across units (panel (b)
that compares $\sd{\delta}{\tilde{x}_{i}}$ and $\sd{\delta}{\hat{x}_{i}}$).

Finally, Figure \ref{fig:worst_case_variance} complements Table \ref{tab:cai}
by illustrating how all of the optimal designs beat Bernoulli RCT
but with relative gains that depend on the complexity of the error
DGP. We measure complexity by the error second-moment lying in $\mathcal{F}_{p_{0}}(1)$,
parameterized by some true $p_{0}$. We imagine a researcher who does
not know $p_{0}$ \emph{ex ante }and thus chooses among our optimal
designs for $p=1,2,\infty$. The figure reports the worst-case approximate
variance for each design–estimator pair against $p_{0}$. The theoretical
predictions are confirmed. First, the $p=1$ design is best for the
most complex (low-$p_{0}$) errors, with $p=2$ optimal in the intermediate
range, and $p=\infty$ dominating for large $p_{0}$. Second, Bernoulli
RCT is uniformly worse than all optimal designs, with the exception
of $p=\infty$ when the errors are complex. Finally, while more complex
errors increase estimation variance in all cases, this increase is
weaker with lower-$p$ optimal designs, showcasing their minimax robustness.

\subsection{Miguel and Kremer (2004)\label{subsec:appx-more-MK}}

Figures \ref{fig:frontier}(b), \ref{fig:sd_exposures}(c, d), and
\ref{fig:worst_case_variance}(b) are direct analogs to panels (a),
(a, b), and (a) of the same figures for the \citet{cai_social_2015}
application, respectively, with unchanged conclusions. A small exception
is that for Bernoulli RCT dominates the $p=1$ optimal design when
the errors are close to spherical (Figure \ref{fig:worst_case_variance}(b)).

Figure \ref{fig:mk_exposure_networks} visualizes the variance-covariance
matrix of indirect exposures without and with whitening, in parallel
to Figure \ref{fig:cai_exposure_heatmaps} except that the units in
\citet{Miguel2004} can be placed on a map. Consistent with the assignment
correlations in Figure \ref{fig:Applications}(b), the presence of
the direct exposure induces exposure covariances to be negative, with
both positive and negative covariances especially strong with $p=\infty$.
Such covariances are removed by whitening for $p=2$ and especially
$p=1$.\footnote{We have not found a good way to visualize a single draw for this application,
to parallel Figure \ref{fig:cai_bipartite_draws}, since both direct
and indirect effects apply to the same units that are located on a
map.}

\subsection{Egger et al. (2022)\label{subsec:appx-Egger}}

This application uses the cash-transfer experiment of \citet{egger2022general},
who randomized an unconditional transfer of approximately USD 1,000
to eligible households in 653 villages in Kenya and estimate both
the direct effect of a village's own treatment and the spillover from
transfers to households in nearby villages. The intervention units
are villages $i$. The authors' saturation design groups villages
into 68 groups of 3 to 24 villages, largely based on Kenyan administrative
units. Each group is assigned to high or low saturation with probability
50\%; then the cash transfer is assigned to all eligible households
in exactly 2/3 of the villages in the high-saturation group and 1/3
of villages in the low-saturation group. They then consider two alternative
specifications for exposures $x_{i}$. The first is the fraction of
treated villages (including village $i$) within the saturation group
of village $i$, and the second is the fraction of eligible households
treated in all villages within 2km of village $i$. Our results apply
in both cases, and they may be especially valuable in the latter case
(which is also the authors' preferred specification), as there are
no naturally defined clusters \citep[e.g.,][]{viviano_causal_2025}.
However, to better illustrate the comparison with saturation designs,
we abstract from this issue and choose the former $x_{i}$ definition
for our semi-synthetic simulation. We take as the outcome $y_{i}$
the village mean of annualized household expenditure among eligible
households, so that $N=K=653$. 

We estimate a village-level regression $y_{i}=\alpha+\beta x_{i}+\tau g_{i}+\upsilon_{i}$
using recentered IV without whitening, and without additional covariates.\footnote{While \citet{egger2022general} weight each village by size to maintain
a household-level interpretation of the estimand, we estimate unweighted
specifications for simplicity.} This yields $\hat{\tau}=316.5$ and $\hat{\beta}=-14.4$. For $S=10,000$
simulations we generate $y_{i}=\beta^{*}x_{i}+\tau^{*}g_{i}+\varepsilon_{i}$
with the six error DGPs of the simulation from Section \ref{subsec:mk},
drawing the cluster shock at the saturation group level. We calibrate
the simulations like in the previous applications: to match the estimated
standard error of the unwhitened recentered IV estimator and targeting
the same shares of variation explained by different error components.\footnote{The corresponding parameters are as follows: Homoskedastic: $\nu=727$;
heteroskedastic: $\gamma=0.088$, $\nu=520$; degree-in-mean: $\gamma=45$,
$\nu=192$; cluster-correlated: $\nu=219$; network-correlated: $\gamma=747$,
$\nu=241$.}

\subparagraph{Results.}

We solve for the feasible design optimized for the indirect effect
in a model that includes a direct effect. We also solve a version
of the design in Appendix \ref{appx:direct_fx_constraint}, adding
a constraint that the worst-case approximate direct effect variance
is no more than 25\% higher than that of unwhitened recentered IV
under independent assignment with a 50\% treatment probability.

Appendix Table \ref{tab:RMSE-Egger} reports effective sample size
gains relative to the baseline \citet{egger2022general} saturation
design. For the spillover effect in Panel (a), the unconstrained-optimal
designs and estimators again deliver large RMSE reductions—on the
order of 100-200\% or larger increases in effective sample size—except
for $p=1$ in the simplest error DGPs and with estimated residuals.
The degradation of direct effect precision in Panel (b) is, however,
much more pronounced than in the earlier \citet{Miguel2004} application.
The constraint solves this issue, resulting in losses of around 20\%
of the effective sample size for $p=1$ or $p=2$.\footnote{By Proposition \ref{prop:Symmetry}, an unconstrained optimal design
is fully characterized by the within-cluster assignment correlation
for each cluster and can be implemented via a cluster-specific saturation
design. We thus optimize the objective directly, without a Gaussian
rounding step. For the constrained optimization problem in Equation
(\ref{eq:constrained}) with finite $p$, maximizing the Lagrangian
at a fixed value of the multiplier decomposes into $B$ separable
one-dimensional concave optimization problems, where $B$ is the number
of clusters. The constraint is satisfied by a one-dimensional search
over the Lagrange multiplier. Note we do not apply the constraint
for $p=\infty$, as we have not developed a computationally efficient
algorithm for this case.} Notably, the spillover effect is still much more precisely estimated
with this constraint than under the baseline saturation design: effective
sample size grows by around 100-250\%. 

\addcontentsline{toc}{section}{Additional Exhibits}

\newpage{}

\section*{Additional Exhibits}

\begin{algorithm}[H]
\caption{\label{alg:full}Feasible Optimal Experimental Design and IV}

\begin{enumerate}
\item[0.] \textbf{Input:}
\begin{itemize}
\item $N\times K$ matrix of the indirect exposure of interest, $W$
\item \emph{Optional:} $N\times K$ matrix of the other included exposure,
$U$; e.g., $U=I$ for direct exposure
\item \emph{Optional:} $N\times L$ matrix of predetermined covariates,
$R$
\item Tuning parameter $p\in[1,\infty)$. Set lower $p$ for higher robustness
to non-spherical errors
\end{itemize}
\item[1.] If $R\ne\emptyset$, replace $W$ and $U$ with their columnwise
residuals after projecting on $R$
\item[2.] Split $(W,U)$ jointly into independent blocks of rows and columns,
$(W_{(b)},U_{(b)})$, if any
\begin{itemize}
\item \emph{Optional: }to ease computation, at the cost of some approximation
error, also split approximately independent blocks (e.g., not fully
independent because of the intercept)
\end{itemize}
\item[3.] Fix $c=0$. Solve for the relaxed-optimal shock covariance matrix
$\Sigma^{\ast}$ block-by-block:
\begin{itemize}
\item If $p=1$, use the Corollary \ref{cor:p1} closed-form solution
\item If all diagonal elements of $(W'_{(b)}W_{(b)})^{p}$ are equal, use
the Corollary \ref{cor:sym_diagonal} closed-form solution
\item Else, use the Proposition \ref{prop:Optimize_l} numerical optimization
\end{itemize}
\item[4.] If $U\ne\emptyset$, choose $c_{\delta}$ that solves the optimization
problem in (\ref{eq:relaxp}) by repeating Step 3 for $c\ne0$ and
$W-cU$ replacing $W$ (on a grid or using gradient-based methods)
\item[5.] Randomize shocks using the Gaussian rounding design $\delta$ from
Section \ref{subsec:Feasible-Implementation}. \textbf{Output $g$}
\item[6.] Using $\var{\delta}g=\Sigma^{GR}$ from (\ref{eq:V}), choose the
recentered instrument:
\begin{itemize}
\item If $U=\emptyset$, follow Theorem \ref{prop:OptimalIV}. \textbf{Output
$z^{\ast}_{\delta}(g)$}
\item If $U\ne\emptyset$, follow Theorem \ref{prop:OptimalIV_direct}.
\textbf{Output $\boldsymbol{z}^{\ast}_{\delta}(g)$}
\end{itemize}
\end{enumerate}
\end{algorithm}

\begin{figure}[H]
\caption{Correlation Distortion from Gaussian Rounding\label{fig:arcsin}}

\begin{centering}
\includegraphics[width=0.4\textwidth]{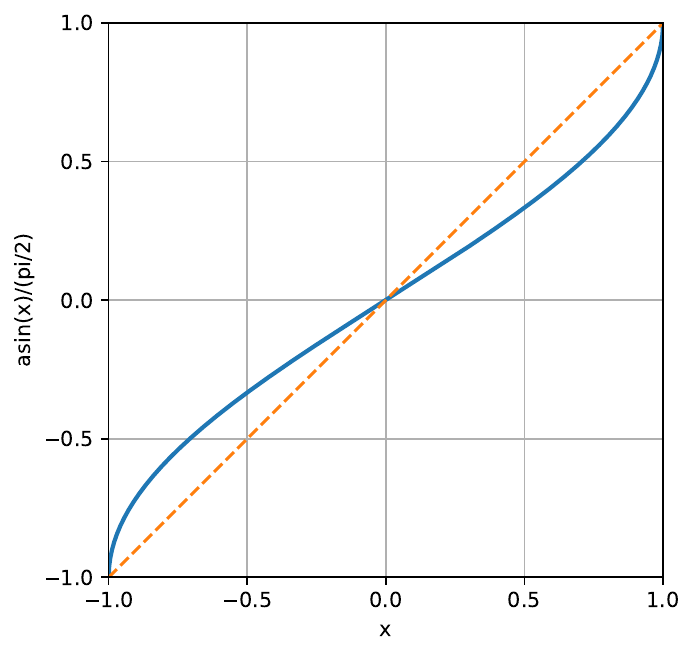}\smallskip{}
\par\end{centering}
{\small\emph{Notes}}{\small : For each value of $\corr{}{\xi_{k},\xi_{l}}$
on the horizontal axis, the vertical axis shows $\corr{}{g_{k},g_{l}}$
from the Gaussian rounding procedure of Section \ref{subsec:Feasible-Implementation}.
The dashed 45-degree line indicates the benchmark of no distortion.}{\small\par}
\end{figure}

\begin{figure}[H]
\caption{Signal and Isotropy Frontiers\label{fig:frontier}}

\begin{centering}
\begin{tabular}{cc}
(a) \citet{cai_social_2015}, $K=995$ & (b) \citet{Miguel2004}, $K=49$\tabularnewline
\includegraphics[width=0.45\textwidth]{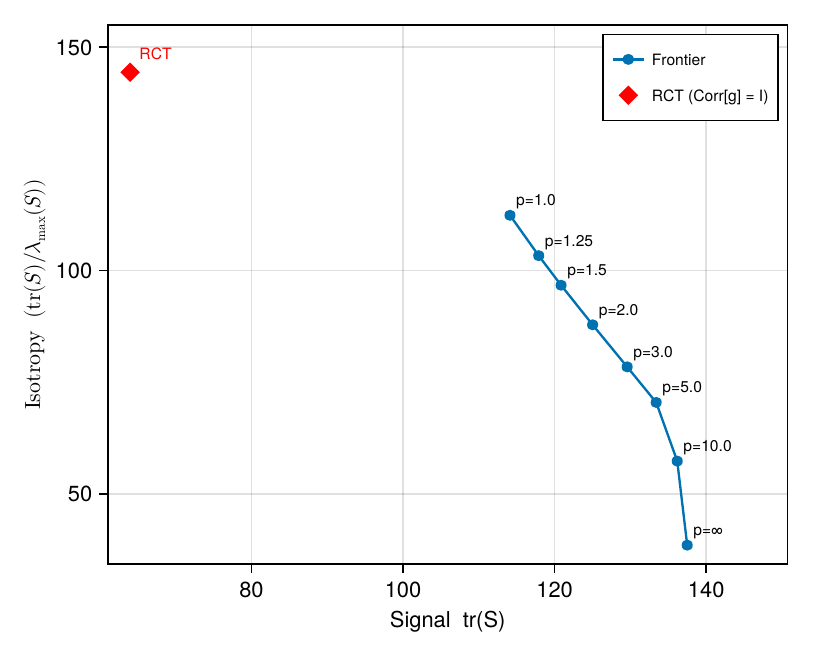} & \includegraphics[width=0.45\textwidth]{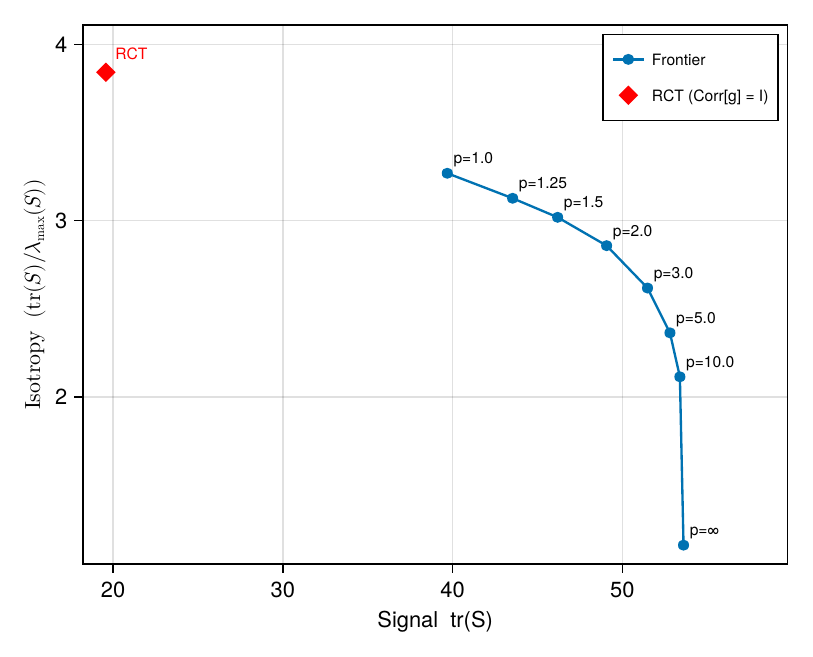}\tabularnewline
\end{tabular}\smallskip{}
\par\end{centering}
{\small\emph{Notes}}{\small : This figure shows the isotropy and signal
of optimal designs, for different values of $p$. The signal and isotropy
of the baseline (independent Bernoulli assignment) is included for
reference. Signal is defined as $\tr(S)$ for $S=\var{\delta}x$,
while isotropy is defined as $\tr(S)/\lambda_{\max}(S)$.}{\small\par}
\end{figure}

\begin{figure}[H]
\caption{The Tradeoff Between Variance and Mutual Correlation in the Equicorrelated
Setting\label{fig:Equicorrelation}}

\begin{centering}
\begin{tabular}{cc}
(a) $N=10$ & (b) $N=50$\tabularnewline
\includegraphics[width=0.45\textwidth]{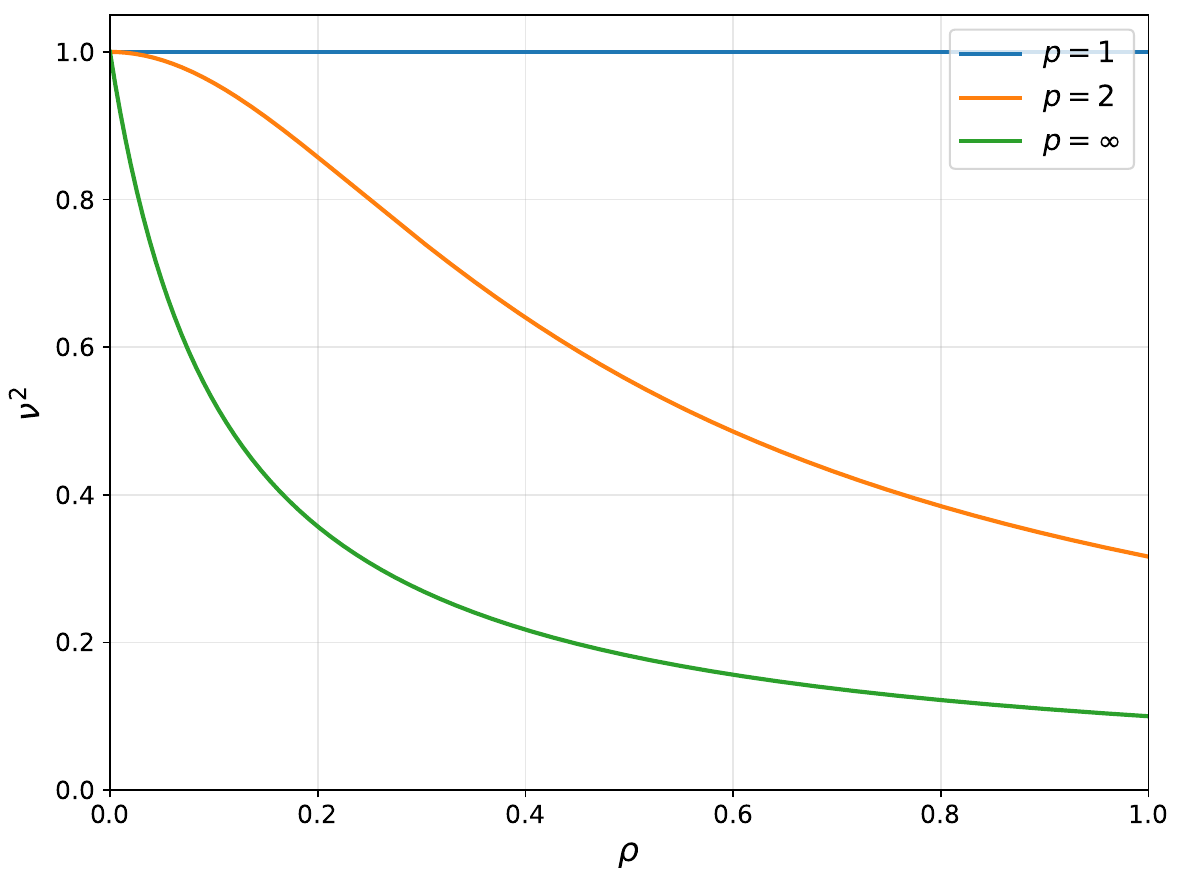} & \includegraphics[width=0.45\textwidth]{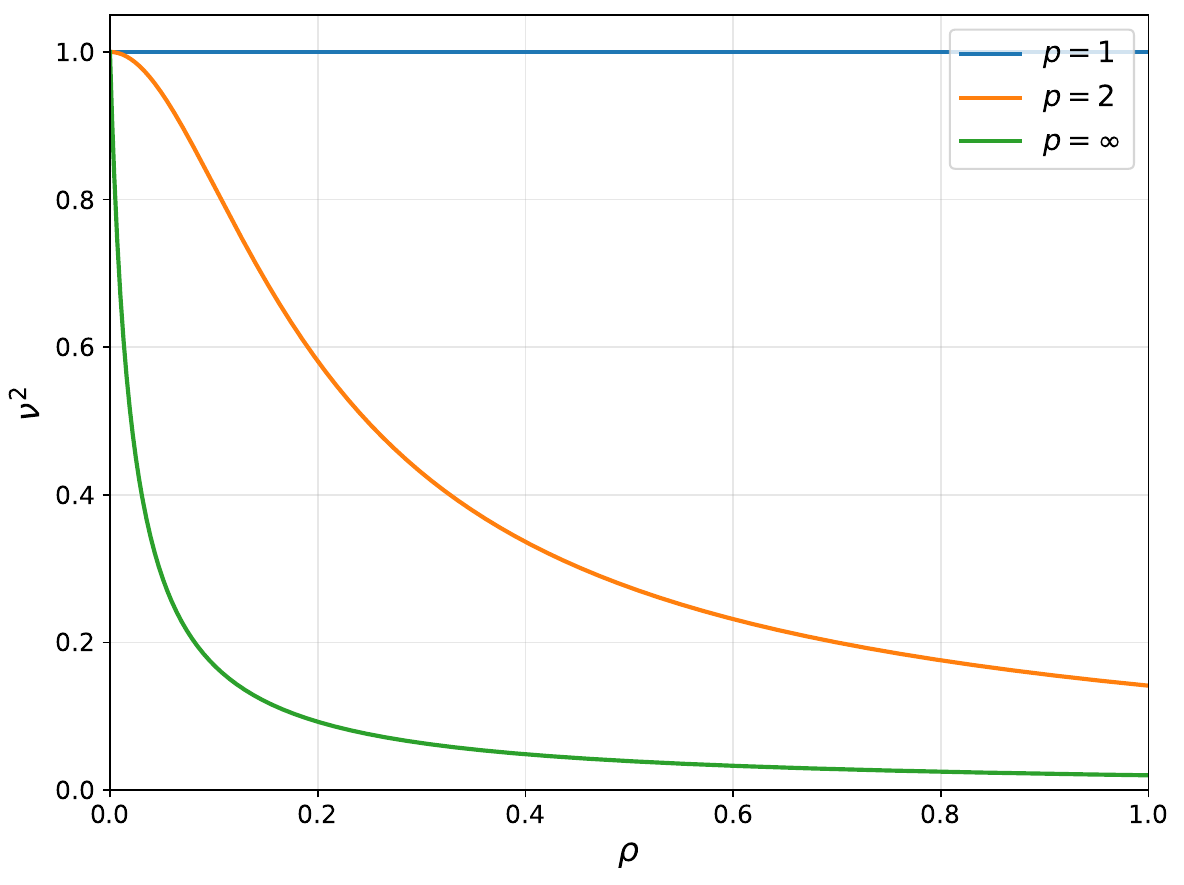}\tabularnewline
\end{tabular}
\par\end{centering}
{\small\emph{Notes: }}{\small In the setting of Appendix \ref{appx:Equicorrelated},
Panel (a) shows the maximum variance of errors $\nu^{2}$, assumed
equal across all $N=10$ observations, such that the error distribution
with equal mutual correlations of errors equal to $\rho\in[0,1]$
is within $\mathcal{F}_{p}(\sigma)$ for $\sigma=1$. The computation
follows condition (\ref{eq:equicorr_condition}). Each line corresponds
to a different choice of $p$. Panel (b) repeats this analysis for
$N=50$.}{\small\par}
\end{figure}

\begin{figure}[H]
\caption{Example Draws from Alternative Experimental Designs, \citet{cai_social_2015}
\label{fig:cai_bipartite_draws}}

\begin{centering}
\includegraphics[width=1\textwidth]{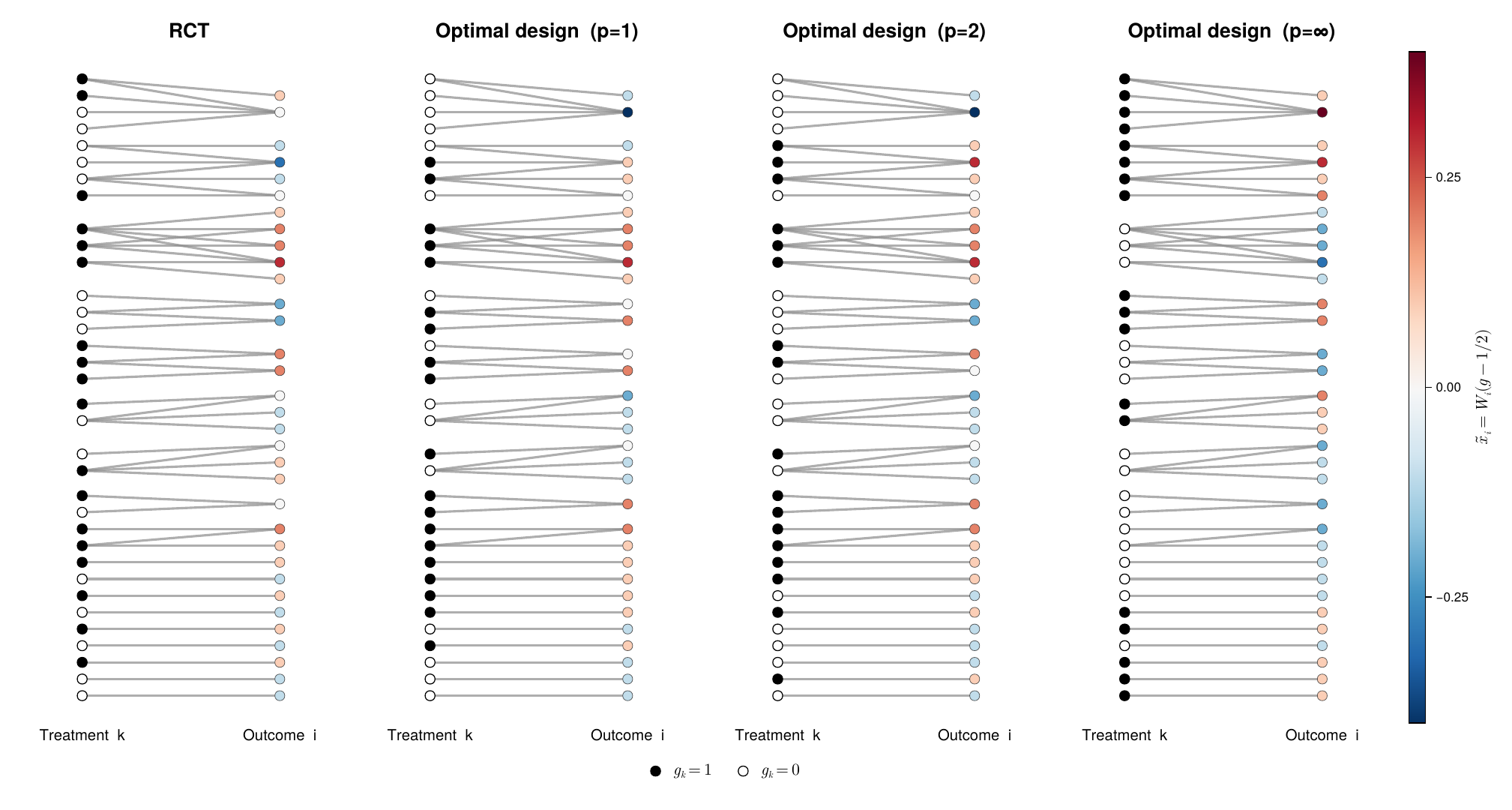}
\par\end{centering}
\begin{centering}
\smallskip{}
\par\end{centering}
{\small\emph{Notes}}{\small : Each panel shows the bipartite graph
of friendships for the example village from Figure \ref{fig:Applications}(a).
For each randomization design, the left of each panel shows one random
draw of the treatment assignment $g_{k}$, with black and white circles
representing treated and untreated intervention units. The right side
of each panel in turn shows the corresponding recentered exposure
$\tilde{x}_{i}$.}{\small\par}
\end{figure}

\begin{figure}[H]
\caption{Variance-Covariance Matrix of Exposure Signal for Alternative Experimental
Designs, \citet{cai_social_2015} \label{fig:cai_exposure_heatmaps}}

\begin{centering}
(a) Without whitening
\par\end{centering}
\begin{centering}
\includegraphics[width=0.95\textwidth]{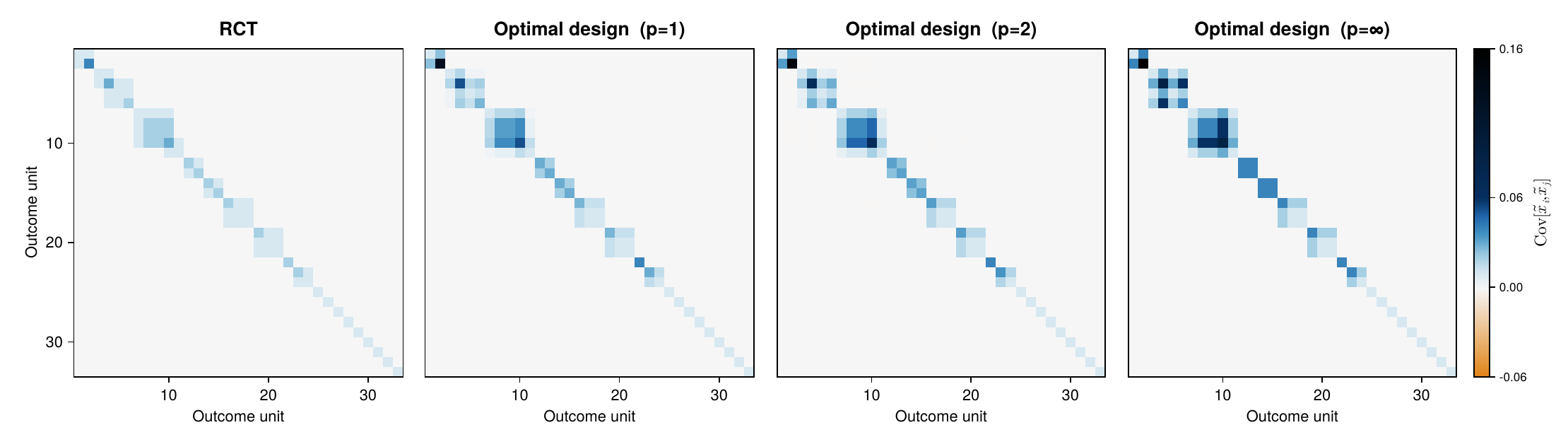}
\par\end{centering}
\begin{centering}
\medskip{}
(b) With whitening
\par\end{centering}
\begin{centering}
\includegraphics[width=0.95\textwidth]{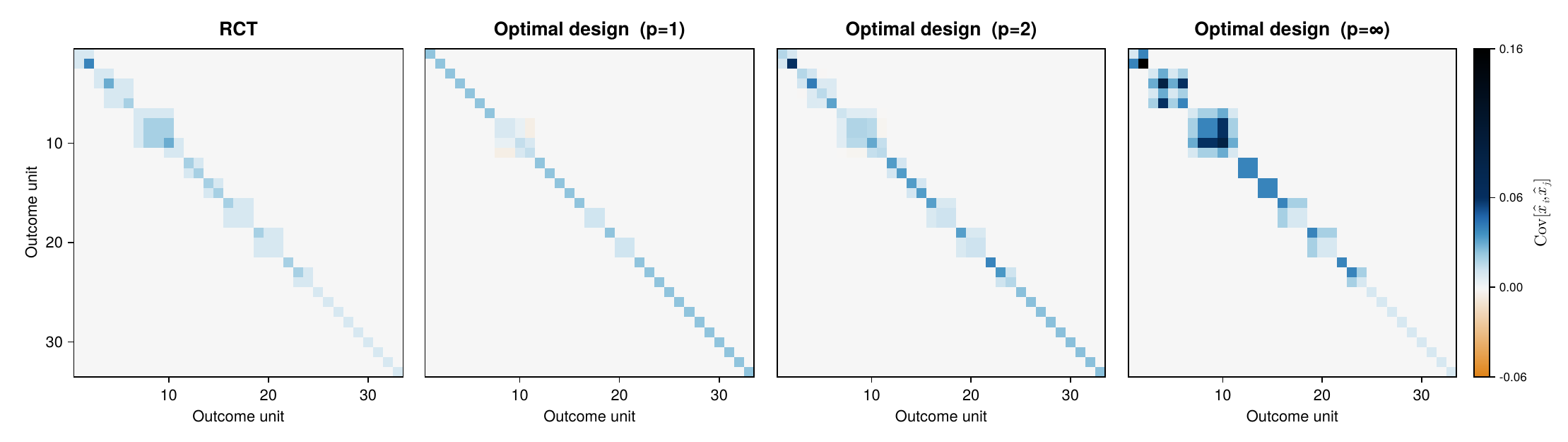}
\par\end{centering}
\begin{centering}
\smallskip{}
\par\end{centering}
{\small\emph{Notes}}{\small : For the example village from Figure \ref{fig:Applications}(a),
panel (a) shows the variance-covariance matrix of the recentered exposure
$\tilde{x}_{i}$, while panel (b) shows it for the first-stage fitted
values $\hat{x}_{i}=\frac{\sum_{i}\expec{\delta}{\tilde{x}_{i}z^{\ast}_{i}}}{\sum_{i}\expec{\delta}{z^{\ast2}_{i}}}z^{\ast}_{i}$.}{\small\par}
\end{figure}

\begin{figure}[H]
\caption{Distribution of Exposure Signal across Outcome Units \protect \\
for Alternative Experimental Designs \label{fig:sd_exposures}}

\begin{centering}
\begin{tabular}{cc}
(a) \citet{cai_social_2015}, Without whitening & (b) \citet{cai_social_2015}, The role of whitening\tabularnewline
\includegraphics[width=0.45\textwidth]{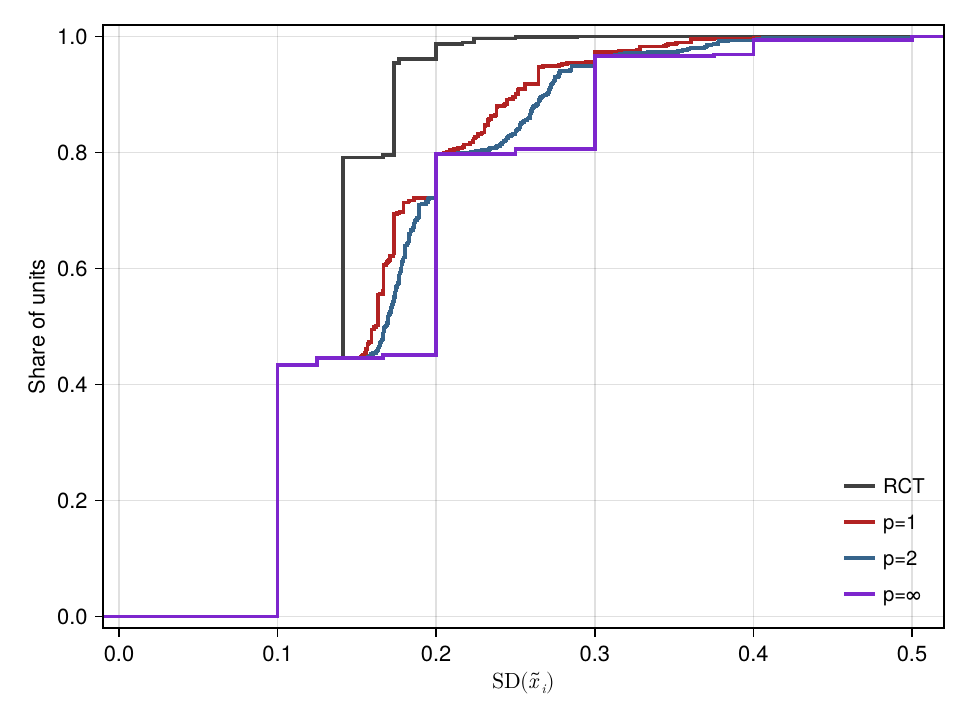} & \includegraphics[width=0.45\textwidth]{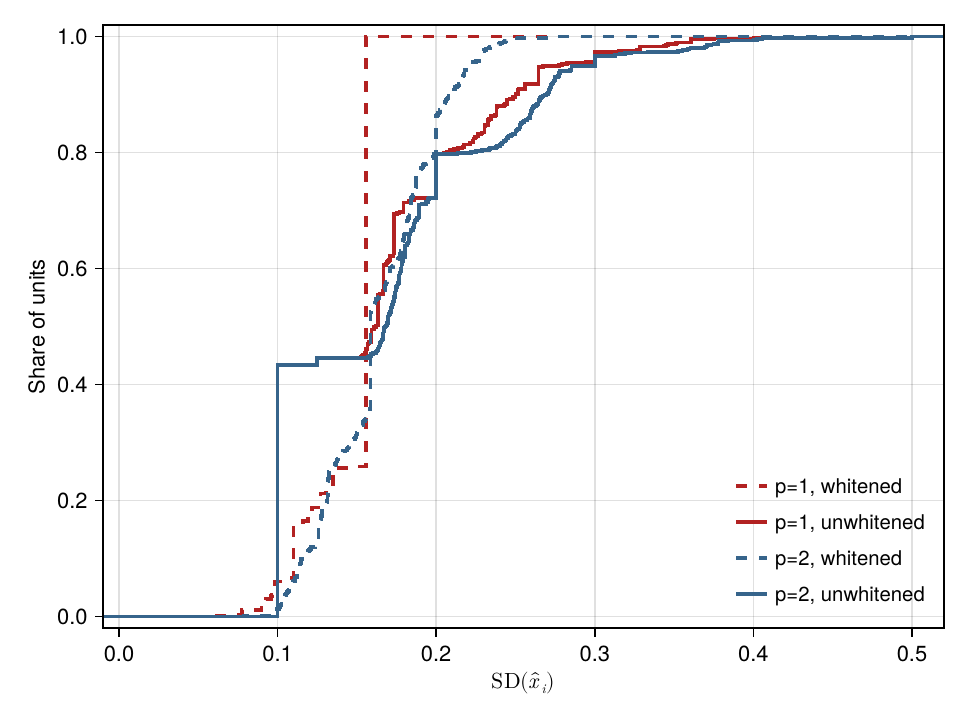}\tabularnewline
 & \tabularnewline
(c) \citet{Miguel2004}, & (d) \citet{Miguel2004},\tabularnewline
Without whitening & The role of whitening\tabularnewline
\includegraphics[width=0.45\textwidth]{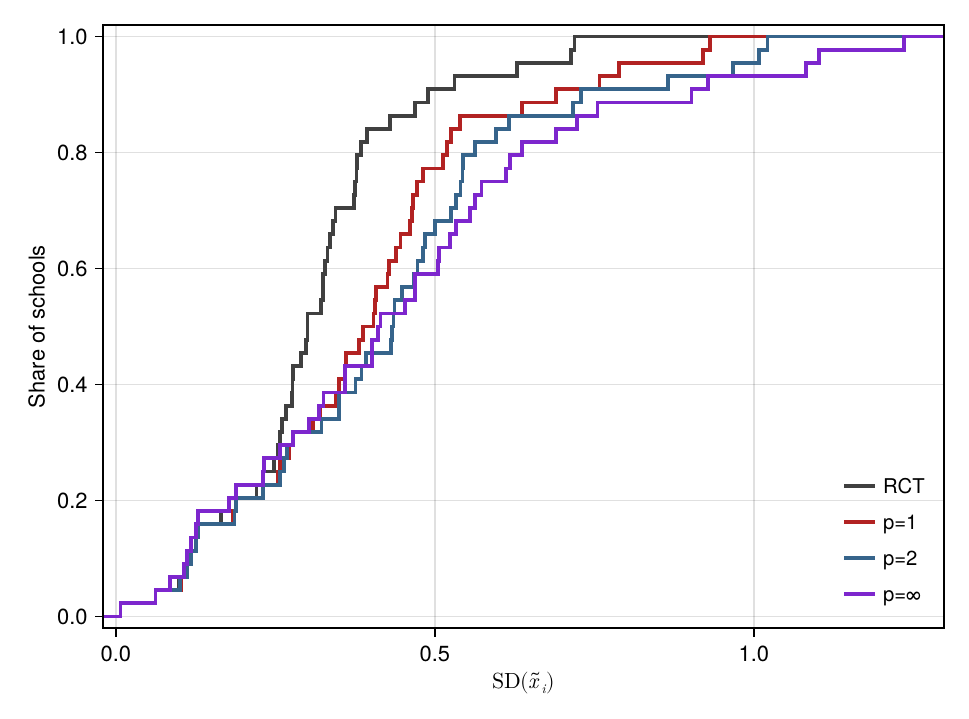} & \includegraphics[width=0.45\textwidth]{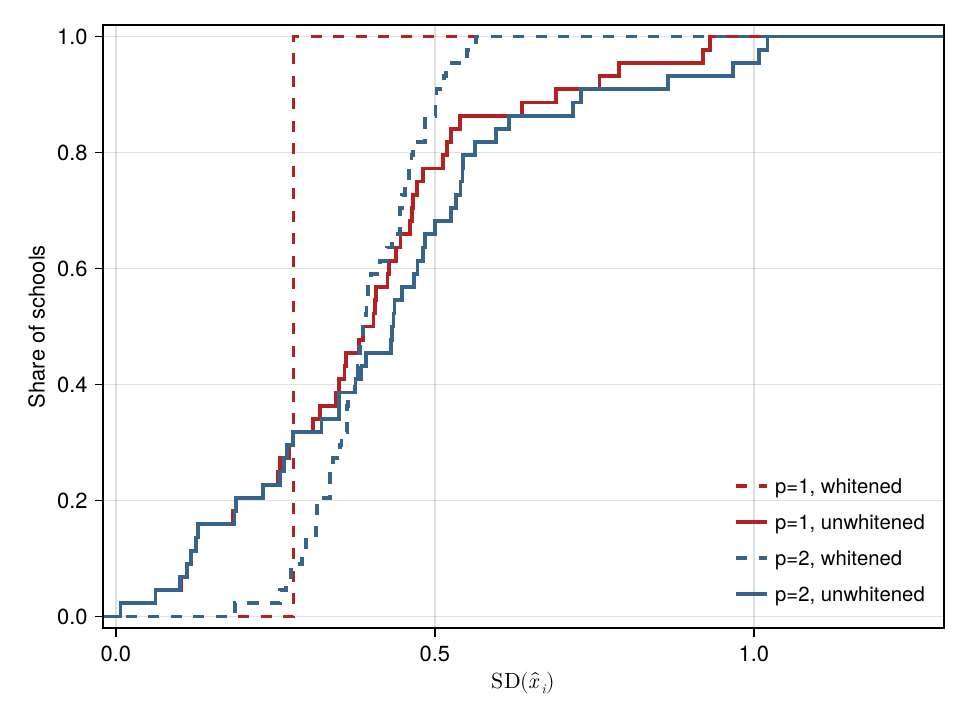}\tabularnewline
\end{tabular}
\par\end{centering}
\begin{centering}
\smallskip{}
\par\end{centering}
{\small\emph{Notes}}{\small : In the \citet{cai_social_2015} application,
for Bernoulli RCT and optimal designs, Panel (a) shows the CDF of
standard deviations of unwhitened recentered exposure $\tilde{x}_{i}$
across outcome units. Panel (b) compares the CDF of standard deviations
of unwhitened exposure $\tilde{x}_{i}$ and whitened exposure $\hat{x}_{i}$,
defined as in Figure \ref{fig:cai_exposure_heatmaps}. Panels (c)
and (d) repeat the same analysis in the \citet{Miguel2004} application.}{\small\par}
\end{figure}

\begin{figure}[H]
\caption{Error Complexity and Worst-Case Approximate Variances \protect \\
of Alternative Experimental Designs \label{fig:worst_case_variance}}

\begin{centering}
\begin{tabular}{cc}
(a) \citet{cai_social_2015} & (b) \citet{Miguel2004}\tabularnewline
\includegraphics[width=0.45\textwidth]{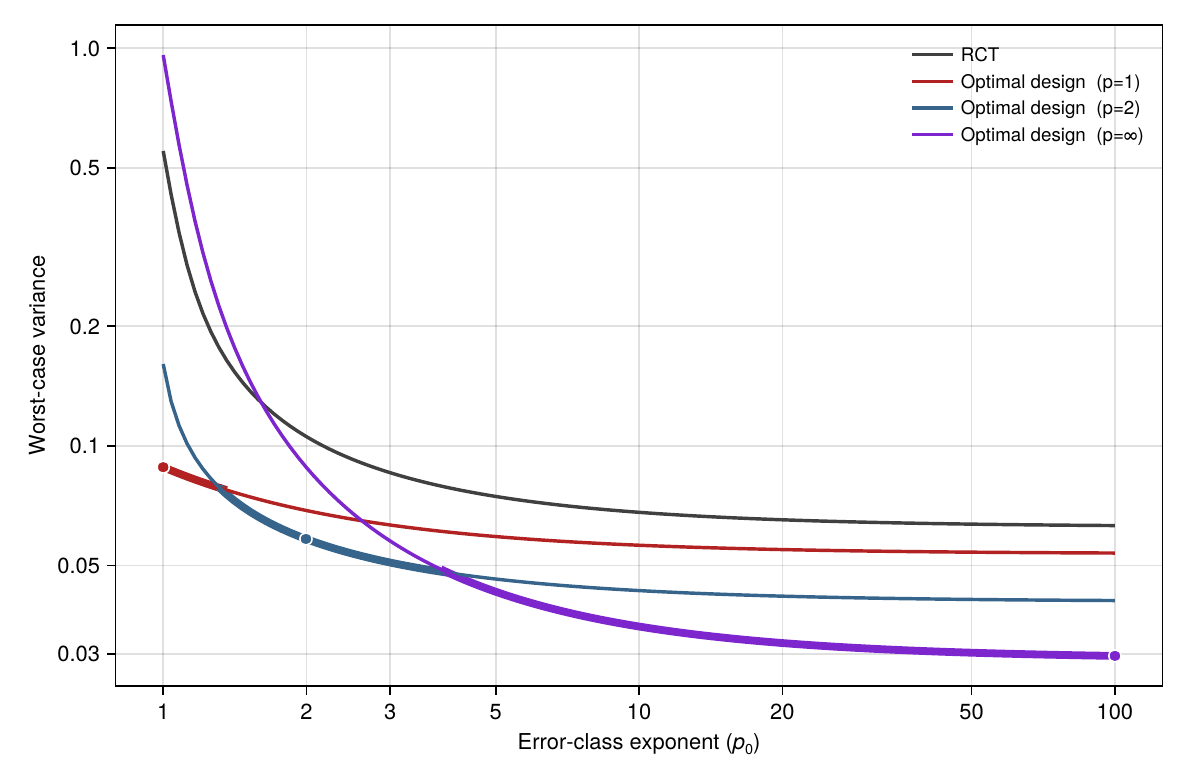} & \includegraphics[width=0.45\textwidth]{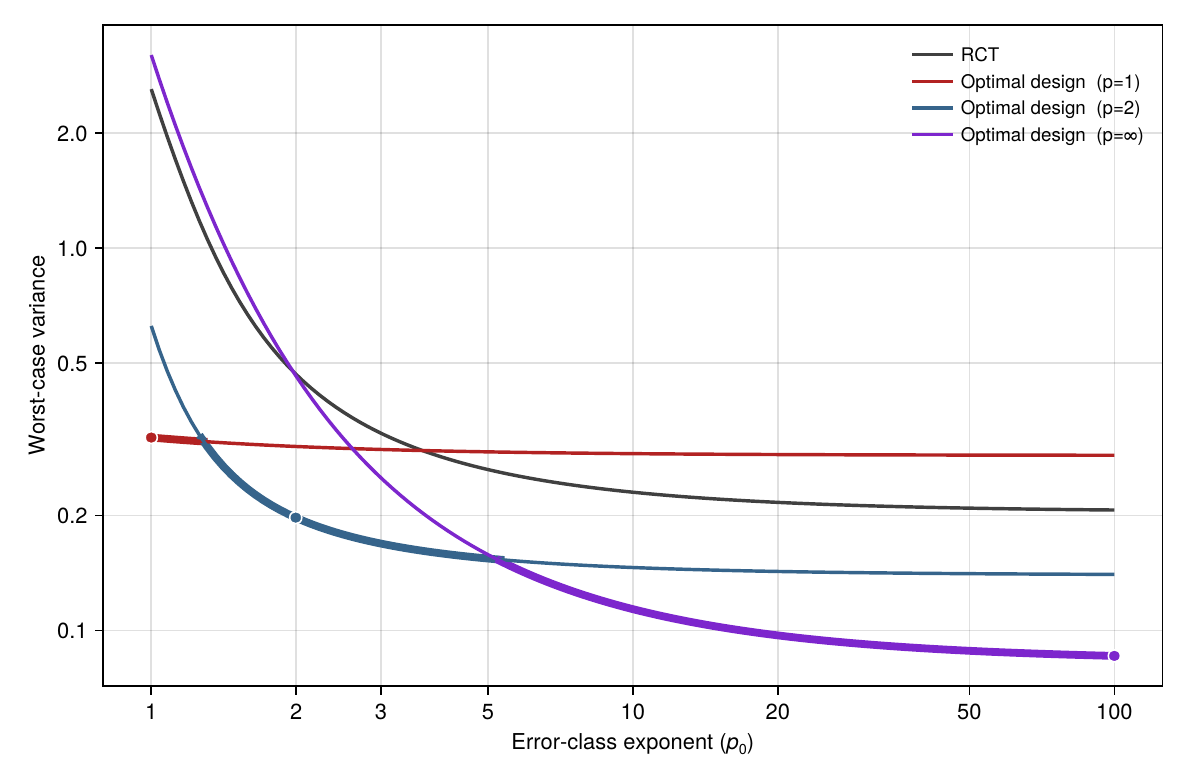}\tabularnewline
\end{tabular}
\par\end{centering}
\begin{centering}
\smallskip{}
\par\end{centering}
{\small\emph{Notes}}{\small : For Bernoulli RCT and optimal design–instrument
combinations, this figure shows the worst-case approximate variance
when the error DGP lies within the class $\mathcal{F}_{p_{0}}(1)$,
for different values of $p_{0}$ shown on the horizontal axis.}{\small\par}
\end{figure}

\begin{sidewaysfigure}[H]
\caption{Variance-Covariance Matrix of Exposure Signal for Alternative Experimental
Designs, \citet{Miguel2004}\label{fig:mk_exposure_networks}}

\begin{centering}
(a) Without whitening
\par\end{centering}
\begin{centering}
\includegraphics[width=0.95\textwidth]{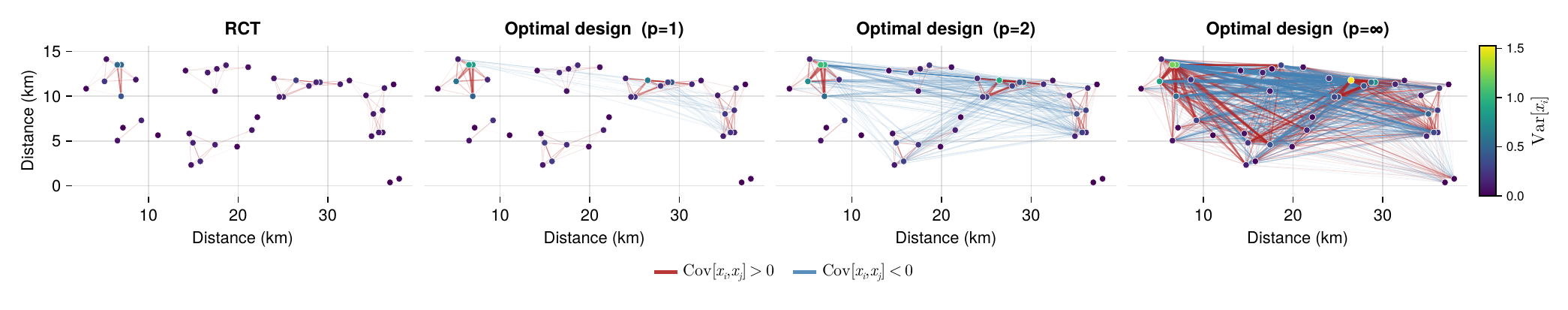}
\par\end{centering}
\begin{centering}
\medskip{}
(b) With whitening
\par\end{centering}
\begin{centering}
\includegraphics[width=0.95\textwidth]{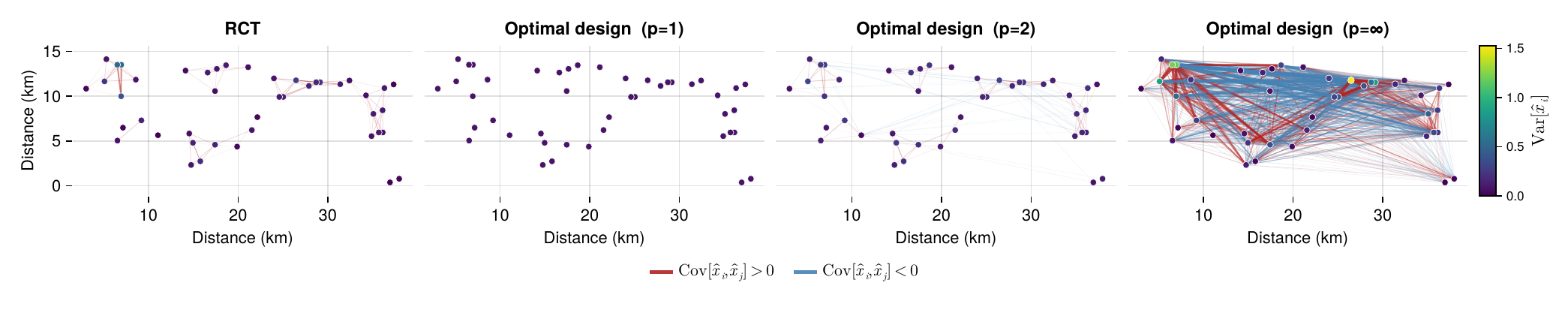}
\par\end{centering}
\begin{centering}
\smallskip{}
\par\end{centering}
{\small\emph{Notes}}{\small : Each node corresponds to a school. Node
colors indicate $\var{\delta}{\tilde{x}_{i}}$ in panel (a) and $\var{\delta}{\hat{x}_{i}}$
in panel (b). Red edges correspond to strictly positive mutual covariances
of $\tilde{x}_{i}$ or $\hat{x}_{i}$and blue edges indicate strictly
negative covariances, while edges with zero covariances are not shown.
The thickness of the edge reflects the magnitude of the corresponding
covariance. Since exact school coordinates are not made available
by \citet{Miguel2004}, relative locations were inferred by triangulating
the published distance matrix and the map orientation is arbitrary.}{\small\par}
\end{sidewaysfigure}

\begin{table}[H]
\caption{RMSE and Effective Sample Size Gains for the \citet{cai_social_2015}
Application, Additional Estimators\label{tab:cai-extra}}

\begin{centering}
\begin{tabular}{l*{6}{c}}
\toprule
 & \shortstack{Homo- \\ skedastic} & \shortstack{Hetero- \\ skedastic} & \shortstack{Deg.-in- \\ Mean} & \shortstack{Village- \\ Corr.} & \shortstack{Network- \\ Corr.} & \shortstack{Estimated \\ Residuals} \\
 & (1) & (2) & (3) & (4) & (5) & (6) \\
\midrule
Effective Sample Gain: \\
\quad Optimal, $p=1$ (whitened) & +19\% & +33\% & +40\% & +70\% & +70\% & +32\% \\
\quad Optimal, $p=1$ (unwhitened) & +63\% & +52\% & +27\% & +63\% & +50\% & +63\% \\
\quad Optimal, $p=1$ (GMM) & +62\% & +56\% & +43\% & +77\% & +72\% & +63\% \\
\quad Optimal, $p=2$ (whitened) & +55\% & +60\% & +44\% & +84\% & +78\% & +67\% \\
\quad Optimal, $p=2$ (unwhitened) & +81\% & +65\% & +25\% & +64\% & +56\% & +75\% \\
\quad Optimal, $p=2$ (GMM) & +81\% & +68\% & +44\% & +83\% & +77\% & +76\% \\
\addlinespace
RMSE for RCT Benchmark & 0.135 & 0.133 & 0.131 & 0.133 & 0.132 & 0.133 \\
\bottomrule
\end{tabular}
\smallskip{}
\par\end{centering}
{\small\emph{Notes}}{\small : This table reports, for the \citet{cai_social_2015}
application, differences in effective sample size for different optimal
designs and estimators relative to the benchmark of the Bernoulli
RCT with the unwhitened recentered IV estimator. The rows correspond
to the Algorithm \ref{alg:full} feasible optimal designs with $R=\one$
(the intercept) for $p=1$ and $p=2$. For each, we compare the performance
of the optimal IV estimator, the unwhitened recentered IV estimator,
and the GMM estimator discussed in Appendix \ref{appx-adaptivity}
that combines them. The last row gives the RMSE of the benchmark design
and estimator. The columns are the six error DGPs described in the
text. Differences in effective sample sizes are computed, for each
DGP, as the squared ratio of inverse RMSE of a given scenario relative
to the benchmark, minus one.}{\small\par}
\end{table}

\begin{table}[H]
\caption{RMSE and Effective Sample Size Gains for the \citet{Miguel2004} Application,
Additional Estimators\label{tab:MK-extra}}

\begin{centering}
(a) Spillover Effect\smallskip{}
\par\end{centering}
\begin{centering}
\begin{tabular}{l*{6}{c}}
\toprule
 & \shortstack{Homo- \\ skedastic} & \shortstack{Hetero- \\ skedastic} & \shortstack{Deg.-in- \\ Mean} & \shortstack{Cluster- \\ Corr.} & \shortstack{Network- \\ Corr.} & \shortstack{Estimated \\ Residuals} \\
 & (1) & (2) & (3) & (4) & (5) & (6) \\
\midrule
Effective Sample Gain: \\
\quad Optimal, $p=1$ (whitened) & +3\% & +16\% & +92\% & +94\% & +169\% & +118\% \\
\quad Optimal, $p=1$ (unwhitened) & +109\% & +106\% & +138\% & +19\% & +99\% & +81\% \\
\quad Optimal, $p=1$ (GMM) & +71\% & +75\% & +137\% & +118\% & +170\% & +96\% \\
\quad Optimal, $p=2$ (whitened) & +106\% & +122\% & +190\% & +164\% & +237\% & +172\% \\
\quad Optimal, $p=2$ (unwhitened) & +164\% & +170\% & +223\% & +150\% & +153\% & +124\% \\
\quad Optimal, $p=2$ (GMM) & +140\% & +145\% & +216\% & +149\% & +220\% & +126\% \\
\addlinespace
RMSE for RCT Benchmark & 0.152 & 0.159 & 0.165 & 0.162 & 0.159 & 0.188 \\
\bottomrule
\end{tabular}

\par\end{centering}
\begin{centering}
\bigskip{}
\par\end{centering}
\begin{centering}
(b) Direct Effect\smallskip{}
\par\end{centering}
\begin{centering}
\begin{tabular}{l*{6}{c}}
\toprule
 & \shortstack{Homo- \\ skedastic} & \shortstack{Hetero- \\ skedastic} & \shortstack{Deg.-in- \\ Mean} & \shortstack{Cluster- \\ Corr.} & \shortstack{Network- \\ Corr.} & \shortstack{Estimated \\ Residuals} \\
 & (1) & (2) & (3) & (4) & (5) & (6) \\
\midrule
Effective Sample Gain: \\
\quad Optimal, $p=1$ (whitened) & -20\% & -19\% & -6\% & -13\% & -21\% & -12\% \\
\quad Optimal, $p=1$ (unwhitened) & -7\% & -1\% & +7\% & -9\% & -12\% & -3\% \\
\quad Optimal, $p=1$ (GMM) & -13\% & -6\% & +3\% & -11\% & -18\% & -12\% \\
\quad Optimal, $p=2$ (whitened) & -15\% & -12\% & +3\% & -12\% & -21\% & -8\% \\
\quad Optimal, $p=2$ (unwhitened) & -10\% & -3\% & +10\% & -10\% & -19\% & -3\% \\
\quad Optimal, $p=2$ (GMM) & -15\% & -9\% & +6\% & -16\% & -26\% & -15\% \\
\addlinespace
RMSE for RCT Benchmark & 0.075 & 0.072 & 0.058 & 0.062 & 0.049 & 0.064 \\
\bottomrule
\end{tabular}
\smallskip{}
\par\end{centering}
{\small\emph{Notes}}{\small : Panel (a) reports, for the \citet{Miguel2004}
application, differences in effective sample size for different optimal
designs and estimators for the spillover effect relative to the benchmark
of the Bernoulli RCT with the unwhitened recentered IV estimator.
Panel (b) reports the gains for the direct effect under the same design
and estimation choice. The rows correspond to the Algorithm \ref{alg:full}
feasible optimal designs with $R=\one$ (the intercept) in the presence
of direct effects, for $p=1$ and $p=2$. For each, we compare the
performance of the optimal IV estimator, the unwhitened recentered
IV estimator, and the GMM estimator discussed in Appendix \ref{appx-adaptivity}
that combines them (along with the instrument for the direct effect).
The last row gives the RMSE of the benchmark design and estimator.
The columns are the six error DGPs described in the text. Differences
in effective sample sizes are computed, for each DGP, as the squared
ratio of inverse RMSE of a given scenario relative to the benchmark,
minus one.}{\small\par}
\end{table}

\begin{table}
\caption{Coverage for \citet{cai_social_2015} Simulation\label{tab:Coverage-Cai}}

\begin{centering}
\begin{tabular}{lcccccc}
\toprule
 & \shortstack{Homo- \\ skedastic} & \shortstack{Hetero- \\ skedastic} & \shortstack{Deg.-in- \\ Mean} & \shortstack{Village- \\ Corr.} & \shortstack{Network- \\ Corr.} & \shortstack{Estimated \\ Residuals} \\
 & (1) & (2) & (3) & (4) & (5) & (6) \\
\hline
RCT Benchmark & $0.951$ & $0.951$ & $0.950$ & $0.947$ & $0.947$ & $0.946$ \\
Optimal, $p=1$ & $0.950$ & $0.951$ & $0.951$ & $0.950$ & $0.950$ & $0.951$ \\
Optimal, $p=2$ & $0.948$ & $0.952$ & $0.949$ & $0.949$ & $0.948$ & $0.951$ \\
Optimal, $p=\infty$ & $0.948$ & $0.945$ & $0.946$ & $0.947$ & $0.947$ & $0.945$ \\
\bottomrule
\end{tabular}
\smallskip{}
\par\end{centering}
{\small\emph{Notes}}{\small : This table reports the empirical coverage
for the whitened recentered IV estimator with different designs (rows)
and different data-generating processes for the errors (columns) for
the \citet{cai_social_2015} application, using a confidence interval
based on the normal approximation $\hat{\beta}\pm1.96\cdot\frac{\sqrt{\hat{v}}}{\sqrt{K}\left|h_{K}\right|}$,
as defined in Section \ref{subsec:clt-inference}. Coverage is computed
as the fraction of times in 10,000 replications of the Monte-Carlo
simulation that the confidence interval contains the true $\beta$.}{\small\par}
\end{table}

\begin{table}
\caption{Coverage for the \citet{Miguel2004} Application\label{tab:Coverage-MK}}

\begin{centering}
(a) Spillover Effect\smallskip{}
\par\end{centering}
\begin{centering}
\begin{tabular}{lcccccc}
\toprule
 & \shortstack{Homo- \\ skedastic} & \shortstack{Hetero- \\ skedastic} & \shortstack{Deg.-in- \\ Mean} & \shortstack{Cluster- \\ Corr.} & \shortstack{Network- \\ Corr.} & \shortstack{Estimated \\ Residuals} \\
 & (1) & (2) & (3) & (4) & (5) & (6) \\
\hline
RCT Benchmark & $0.949$ & $0.943$ & $0.922$ & $0.921$ & $0.904$ & $0.885$ \\
Optimal, $p=1$ & $0.973$ & $0.967$ & $0.952$ & $0.964$ & $0.952$ & $0.949$ \\
Optimal, $p=2$ & $0.960$ & $0.960$ & $0.939$ & $0.939$ & $0.934$ & $0.915$ \\
Optimal, $p=\infty$ & $0.651$ & $0.655$ & $0.609$ & $0.491$ & $0.590$ & $0.424$ \\
\bottomrule
\end{tabular}

\par\end{centering}
\begin{centering}
\bigskip{}
\par\end{centering}
\begin{centering}
(b) Direct Effect\smallskip{}
\par\end{centering}
\begin{centering}
\begin{tabular}{lcccccc}
\toprule
 & \shortstack{Homo- \\ skedastic} & \shortstack{Hetero- \\ skedastic} & \shortstack{Deg.-in- \\ Mean} & \shortstack{Cluster- \\ Corr.} & \shortstack{Network- \\ Corr.} & \shortstack{Estimated \\ Residuals} \\
 & (1) & (2) & (3) & (4) & (5) & (6) \\
\hline
RCT Benchmark & $0.953$ & $0.950$ & $0.958$ & $0.952$ & $0.955$ & $0.951$ \\
Optimal, $p=1$ & $0.954$ & $0.956$ & $0.955$ & $0.949$ & $0.948$ & $0.940$ \\
Optimal, $p=2$ & $0.947$ & $0.945$ & $0.953$ & $0.942$ & $0.942$ & $0.931$ \\
Optimal, $p=\infty$ & $0.690$ & $0.731$ & $0.655$ & $0.653$ & $0.728$ & $0.413$ \\
\bottomrule
\end{tabular}

\par\end{centering}
\begin{centering}
\smallskip{}
\par\end{centering}
{\small\emph{Notes}}{\small : Panel (a) reports the empirical coverage
for the whitened recentered IV estimator for the indirect effect with
different designs (rows) and different data-generating processes for
the errors (columns) for the \citet{Miguel2004} application, using
a confidence interval based on the normal approximation, applying
an extension of Theorem \ref{thm:main} to multiple exposures from
Appendix }\ref{app:Theorem2_direct}{\small . Panel (b) reports the
same for the direct effect. For each panel, coverage is computed as
the fraction of times in 10,000 replications of the Monte-Carlo simulation
that the confidence interval contains the true target parameter.}{\small\par}
\end{table}

\begin{table}[tp]
\caption{RMSE and Effective Sample Size Gains for the \citet{egger2022general}
Application\label{tab:RMSE-Egger}}

\begin{centering}
(a) Spillover Effect\smallskip{}
\par\end{centering}
\begin{centering}
\begin{tabular}{l*{6}{c}}
\toprule
 & \shortstack{Homo- \\ skedastic} & \shortstack{Hetero- \\ skedastic} & \shortstack{Deg.-in- \\ Mean} & \shortstack{Group- \\ Corr.} & \shortstack{Network- \\ Corr.} & \shortstack{Estimated \\ Residuals} \\
 & (1) & (2) & (3) & (4) & (5) & (6) \\
\midrule
Effective Sample Gain: \\
\quad Optimal, $p=1$ & -19\% & -18\% & +479\% & +340\% & +296\% & -14\% \\
\quad Optimal, $p=2$ & +123\% & +118\% & +554\% & +446\% & +394\% & +131\% \\
\quad Optimal, $p=\infty$ & +160\% & +151\% & +333\% & +316\% & +316\% & +163\% \\
\quad Constrained Optimal, $p=1$ & +105\% & +121\% & +266\% & +176\% & +143\% & +103\% \\
\quad Constrained Optimal, $p=2$ & +96\% & +111\% & +260\% & +173\% & +136\% & +95\% \\
\addlinespace
RMSE for Saturation Design & 185.501 & 186.470 & 182.562 & 182.122 & 185.770 & 182.701 \\
\bottomrule
\end{tabular}
\bigskip{}
\par\end{centering}
\begin{centering}
(b) Direct Effect\smallskip{}
\par\end{centering}
\begin{centering}
\begin{tabular}{l*{6}{c}}
\toprule
 & \shortstack{Homo- \\ skedastic} & \shortstack{Hetero- \\ skedastic} & \shortstack{Deg.-in- \\ Mean} & \shortstack{Group- \\ Corr.} & \shortstack{Network- \\ Corr.} & \shortstack{Estimated \\ Residuals} \\
 & (1) & (2) & (3) & (4) & (5) & (6) \\
\midrule
Effective Sample Gain: \\
\quad Optimal, $p=1$ & -90\% & -90\% & -93\% & -91\% & -91\% & -91\% \\
\quad Optimal, $p=2$ & -65\% & -66\% & -71\% & -66\% & -67\% & -67\% \\
\quad Optimal, $p=\infty$ & -45\% & -45\% & -44\% & -45\% & -43\% & -45\% \\
\quad Constrained Optimal, $p=1$ & -19\% & -19\% & -21\% & -17\% & -16\% & -20\% \\
\quad Constrained Optimal, $p=2$ & -20\% & -19\% & -21\% & -19\% & -18\% & -19\% \\
\addlinespace
RMSE for Saturation Design & 60.684 & 61.189 & 15.891 & 18.328 & 20.206 & 56.682 \\
\bottomrule
\end{tabular}
\smallskip{}
\par\end{centering}
{\small\emph{Notes}}{\small : Panel (a) reports, for the \citet{egger2022general}
application, differences in effective sample size for different optimal
designs and estimators for the spillover effect relative to the \citet{egger2022general}
saturation design with the unwhitened recentered IV estimator. Panel
(b) reports the gains or losses for the direct effect under the same
design and estimation choice. The rows correspond to the optimal designs
and constrained optimal designs discussed in Appendix \ref{appx:direct_fx_constraint}
with $R=\one$ (the intercept) in the presence of direct effects,
for $p=1$, $p=2$, and $p=\infty$. The constraint imposed (for $p<\infty$)
is that the worst-case approximate asymptotic direct effect variance
is no more than 25\% higher compared to the unwhitened estimator with
a 50\% Bernoulli RCT. The last row gives the RMSE of the benchmark
design and estimator, where the benchmark design randomizes uniformly
between completely randomized designs with 1/3rd and 2/3rd treatment
fractions. The treated number is rounded to the nearest integer for
group sizes that are not divisible by 3. The columns are the six error
DGPs described in the text. Differences in effective sample sizes
are computed, for each DGP, as the squared ratio of inverse RMSE of
a given scenario relative to the benchmark, minus one.}{\small\par}
\end{table}

\end{document}